\documentclass[11pt,USletter]{article}
\def\withcolors{0}
\def\withnotes{0}

\usepackage{amsmath}
\usepackage{enumerate}
\newcommand{\e}{\varepsilon}
\newcommand{\iprod}[1]{\langle #1 \rangle}
\newcommand{\mper}{\, .}

\newcommand{\mcl}{\mathcal}
\DeclareMathOperator{\Gram}{Gram}
\DeclareMathOperator{\Sum}{Sum}

\newcommand{\fS}{\mathfrak{S}}
\newcommand{\bone}{\mathbf{1}}
\newcommand{\Filter}{\textsf{Filter}}

\usepackage{ccanonne}
\usepackage{times}

\newcommand{\sidebysidecaption}[4]{\raggedright \begin{minipage}[t]{#1}
    \vspace*{0pt}
    #3
  \end{minipage}
  \hfill \begin{minipage}[t]{#2}
    \vspace*{0pt}
    #4
\end{minipage}}

\definecolor{identifiercolor}{rgb}{.4,.6,.56}
\definecolor{stringcolor}{gray}{0.5}
\definecolor{inactivecolor}{rgb}{0.15,0.15,0.5}
\usepackage{listings}
\algrenewcommand\algorithmicrequire{\textbf{Input:}}
\algrenewcommand\algorithmicreturn{\textbf{Output:}}

\title{The Full Landscape of Robust Mean Testing: Sharp Separations between Oblivious and Adaptive Contamination
}

\author{Cl\'{e}ment Canonne\thanks{\texttt{clement.canonne@sydney.edu.au}. Supported by an ARC DECRA (DE230101329) and an unrestricted gift from Google Research.} \\ University of Sydney \and Samuel B. Hopkins\thanks{\texttt{samhop@mit.edu}. Supported by NSF Award No. 2238080 and MLA@CSAIL} \\ MIT  \and Jerry Li\thanks{\texttt{jerrl@microsoft.com}.} \\ Microsoft Research \and Allen Liu\thanks{\texttt{cliu568@mit.edu} Supported by an NSF Graduate Research Fellowship and a Fannie and John Hertz Foundation Fellowship.} \\ MIT \and Shyam Narayanan\thanks{\texttt{shyamsn@mit.edu}. Supported by an NSF Graduate Fellowship and a Google Fellowship.} \\ MIT}

\date{\today}

\begin{document}

\maketitle

\begin{abstract}
We consider the question of Gaussian mean testing, a fundamental task in high-dimensional distribution testing and signal processing, subject to adversarial corruptions of the samples. We focus on the relative power of different adversaries, and show that, in contrast to the common wisdom in robust statistics, there exists a strict separation between adaptive adversaries (strong contamination) and oblivious ones (weak contamination) for this task. Specifically, we resolve both the information-theoretic and computational landscapes for robust mean testing. In the exponential-time setting, we establish the tight sample complexity of testing $\cN(0,I)$ against $\cN(\alpha v, I)$, where $\|v\|_2 = 1$, with an $\varepsilon$-fraction of adversarial corruptions, to be
\[
    \tilde{\Theta}\!\left(\max\left(\frac{\sqrt{d}}{\alpha^2}, \frac{d\varepsilon^3}{\alpha^4},\min\left(\frac{d^{2/3}\varepsilon^{2/3}}{\alpha^{8/3}}, \frac{d \varepsilon}{\alpha^2}\right)\right) \right) \,,
\]
while the complexity against adaptive adversaries is
\[
    \tilde{\Theta}\!\left(\max\left(\frac{\sqrt{d}}{\alpha^2}, \frac{d\varepsilon^2}{\alpha^4} \right)\right) \,,
\]
which is strictly worse for a large range of vanishing $\varepsilon,\alpha$. To the best of our knowledge, ours is the first separation in sample complexity between the strong and weak contamination models. 

In the polynomial-time setting, we close a gap in the literature by providing a polynomial-time algorithm against adaptive adversaries achieving the above sample complexity $\tilde{\Theta}(\max({\sqrt{d}}/{\alpha^2}, {d\varepsilon^2}/{\alpha^4} ))$, and a low-degree lower bound (which complements an existing reduction from planted clique) suggesting that all efficient algorithms require this many samples, even in the oblivious-adversary setting.
\end{abstract}

\thispagestyle{empty}

\if\withnotes=1
\listoftodos
\fi

\newpage

\setcounter{tocdepth}{2}

\thispagestyle{empty}
\bgroup\small
\tableofcontents
\egroup
\newpage 
\setcounter{page}{1}
\section{Introduction}
Among all high-dimensional distribution testing (i.e., hypothesis testing) problems, \emph{Gaussian mean testing} is one of the most basic, with connections to signal processing where it corresponds to \emph{signal detection} under white noise. Given $n$ independent samples $X_1,\ldots,X_n \in \R^d$, the goal is to decide between two hypotheses:
\begin{quote}
    $\mathbf{H}_0$: $X_1,\ldots,X_n$ were drawn from $\cN(0,I)$, an origin-centered identity-covariance Gaussian.\\
    $\mathbf{H}_1$: $X_1,\ldots,X_n$ were drawn from $\cN(\mu,I)$ for some vector $\mu$ with $\normtwo{\mu} \geq \dst$.
\end{quote}
The following simple tester uses only $\Theta(\sqrt{\dims}/\dst^2)$ samples, the information-theoretic optimum: reject the null iff the norm of the empirical mean \smash{$\big\|\frac 1 n \sum_{i=1}^n X_i \big\|_2$} is larger than some well-chosen threshold. The number of samples scales as the square root of the dimension: in contrast, $\Theta(\dims/\dst^2)$ samples (linear in the dimension) are needed to \emph{learn} the mean $\mu$ of a Gaussian $\cN(\mu,I)$ up to $\lp[2]$ error $\dst$.
This $d$-vs-$\sqrt d$ gap is a prime example of a core theme in the literature on distribution testing: testing requires fewer samples than learning.

This simple tester is not robust to even a small fraction of adversarially corrupted samples.
Concretely, suppose that an $\cor$-fraction of the samples $X_1,\ldots,X_n$ are chosen by a malicious adversary.
Even after preprocessing the dataset by removing obvious outliers~--~say, $X_i$ such that $\normtwo{X_i} \gg \E \normtwo{X_i} \approx \sqrt{d}$ -- the simple tester with $\Theta(\sqrt{\dims}/\dst^2)$ samples can be fooled by just a single corrupted sample.

Robust distribution testing has been extensively studied in robust statistics (the sub-field of statistics dealing with adversarially-corrupted data) \cite{DiakonikolasKS17,DiakonikolasK:ARS}, and yet basic questions about robust mean testing remain open. 
Most importantly: \emph{what is the sample-optimal robust mean tester?}
As we show, the answer to this question is intimately intertwined with another unanswered question in robust statistics: \emph{how much does it matter if the adversary sees the uncorrupted portion of the dataset?}

We find the latter question interesting for (at least) two reasons.
First, it is a foundational question about the power of statistical adversaries -- since modeling assumptions can have a strong effect on algorithm design, it is important to understand the consequences of basic assumptions.
We are not the first to ask the question from this perspective; see also recent work of Blanc, Lange, Malik, and Tan~\cite{blanc2022power}.
Second, the question is pertinent to \emph{data poisoning attacks} in machine learning
\cite{deng2021separation,goldblum2022dataset}, where an adversary injects a small amount of malicious training data into a machine learning pipeline.
Such attacks can be feasible in practice 
 and hence are a significant concern \cite{kumar2020adversarial}.
If an \emph{oblivious} adversary is strictly less powerful than an \emph{adaptive} one, then keeping the training data secret is a potential (partial) defense against data poisoning.

It turns out that oblivious and adaptive adversaries have equal power for robust mean testing's close (and intensely studied \cite{DiakonikolasK:ARS}) cousin, robust mean \emph{estimation}.\footnote{Here we mean that the \emph{sample complexity} of robust mean estimation is insensitive to details of the adversary's power. However, some separations are known, for instance between \emph{additive} versus \emph{additive and subtractive} adversaries in the polynomial-time setting \cite{diakonikolas2018robustly}. See \cref{sec:related} for further discussion.}
Here, the goal is to estimate $\mu$ up to $\lp[2]$ error $\dst$ -- in both adaptive and oblivious cases this requires $\Theta(\tfrac{\dims}{\dst^2})$ samples.
Indeed, this appears to be the case for a range of robust estimation problems, including covariance estimation and linear regression.
This suggests a conventional wisdom in robust statistics: adaptivity does not buy statistical adversaries additional power.

Returning to robust mean testing, recent work by Narayanan \cite{narayanan2022privatetesting} shows that the sample complexity of robust mean testing against an adaptive adversary is $\tilde{\Theta}(\max({\sqrt{\dims}}/{\dst^2}, {\dims \cor^2}/{\dst^4}))$.\footnote{Narayanan's work focuses on differentially private mean testing, but this result can be extracted using known reductions between robustness and privacy.}
This brings us to:
\vspace{-0.35em}
\begin{center}
    \textbf{Main Question:} \emph{What is the optimal robust mean tester against an \textbf{oblivious} adversary? Are the sample complexities of testing against adaptive and oblivious adversaries the same, as they are in robust estimation?}
\end{center}

We answer this question by showing that the common wisdom~--~being resilient to stronger adversaries comes essentially ``for free''~--~does \emph{not} extend to mean testing, where being robust against an oblivious adversary is strictly easier than against a fully adaptive one (\cref{theo:oblivious:informal})!
In fact, we resolve (up to log factors) the sample complexity of robust Gaussian mean testing in the presence of an oblivious adversary, by designing a new robust mean tester and proving a nearly-matching information-theoretic lower bound.

To make the landscape even more interesting, we also show that this separation vanishes when one requires the tester to be \emph{computationally efficient}.
We first give a polynomial-time (in fact, quadratic time) variant of Narayanan's tester, and then we obtain a lower bound against a large class of efficient algorithms (``low-degree algorithms'') which shows a matching sample complexity against both oblivious and adaptive adversaries (\cref{theo:sos-lb:informal}).
(This complements a reduction from planted clique by Brennan and Bresler \cite{brennan2020reducibility} which also suggests that efficient algorithms require $\tfrac{\dims \cor^2}{\dst^4}$ samples even against oblivious adversaries.)
One consequence is a new statistical-computational gap for robust mean testing against an oblivious adversary.

In order to discuss our results in more detail, we describe in the next section the standard adversarial corruption models we consider in our work, and how they relate.
Then we state our results and provide an overview of the new techniques and ideas that underlie our proofs and algorithms.

\subsection{Types of Adversaries} 
    \label{ssec:adv-models}

We focus on two main types of adversarial corruptions: namely, the \emph{adaptive} (strong) and \emph{oblivious} corruption models. These have a long history in Statistics and Algorithmic Robust Statistics; see~\cite{DiakonikolasK19,DiakonikolasK:ARS} for a more thorough discussion. In what follows, we assume that the corruption rate $\cor$ is provided to the algorithm. Note that this is without loss of generality, as, given $\dims$, $\dst$, and the expressions of the sample complexities, the algorithm can compute the largest value of $\cor$ it can tolerate for a given number $\ns$ of samples.

The first corruption model allows an \emph{adaptive} adversary to look at the samples, and choose an $\cor$-fraction of them to alter arbitrarily. Which subset of the samples was corrupted is unknown to the algorithm.
\begin{definition}[Strong contamination model] 
    \label{def:adaptive}
In the strong contamination model, $\ns$ \iid samples $X'_1,\dots,X'_\ns$ are drawn from the underlying unknown distribution $\cD$. The adversary, upon observing $X'_1,\dots,X'_\ns$, chooses $\cor \ns$ indices $i_1,\dots, i_{\cor\ns}$ and values $X''_{i_1}, \dots, X''_{i_{\cor\ns}}$. The algorithm then receives the sequence $X_1,\dots, X_\ns$, where
$X_{i_j} = X''_{i_j}$ for all $j\in[\cor\ns]$, and 
$X_i=X'_i$ otherwise.
Crucially, both the $\cor\ns$ indices and the values $X''_i$ can depend on the ``uncorrupted'' samples $X'_1,\dots,X'_\ns$.
\end{definition}

In contrast, in the \emph{oblivious} contamination model, the adversary must commit to which fraction of the samples it will corrupt, and how, \emph{before} observing the actual realization of the samples. (It is, however, allowed knowledge of both the specification of the algorithm and the underlying distribution.)
\begin{definition}[Oblivious contamination model] 
    \label{def:oblivious}
The adversary chooses $\cor \ns$ indices $i_1,\dots, i_{\cor\ns}$ and values $X''_{i_1}, \dots, X''_{i_{\cor\ns}}$. Then $\ns$ \iid samples $X'_1,\dots,X'_\ns$ are drawn from the underlying unknown distribution $\cD$, and the algorithm is provided with the sequence $X_1,\dots, X_\ns$, as in~\cref{def:adaptive}.
\end{definition}
This definition does allow the corrupted samples to be chosen in a correlated fashion; however, they cannot depend on the realizations of the uncorrupted points themselves. 
This oblivious model can be further weakened, leading to what is known as the \emph{Huber contamination model} where the corrupted data points themselves must be chosen independently of each other:
\begin{definition}[Huber contamination model] 
In the Huber contamination model, the adversary chooses a corruption distribution $\tilde{\cD}$ (possibly a function of the algorithm and underlying unknown distribution $\cD$). Then $\ns$ \iid samples $X_1,\dots,X_\ns$ are drawn from the mixture $(1-\cor)\cD + \cor \tilde{\cD}$, and provided to the algorithm.
\end{definition}
While the focus of our work is on the adaptive and oblivious contamination models, some of our lower bounds apply even to the weaker Huber contamination model.
\vspace{-0.5em}
\subsection{Our Results}
Our main result settles the sample complexity of robust mean testing under oblivious contamination, and establishes a strict separation between oblivious and adaptive contamination models.
In what follows, $\tilde{O},\tilde{\Theta},\tilde{\Omega}$ hide polylogarithmic factors in the argument, and we always assume\footnote{We note that, for identity-covariance Gaussian distributions, mean $\lp[2]$ distance $\dst$ corresponds (for small $\dst$) to total variation (TV) distance $\Theta(\dst)$. Thus, $\lp[2]$ mean testing corresponds to TV testing, which motivates the regime $\dst \ll 1$ as of particular interest.} $\dst \leq O(1)$ and $\cor \leq \dst / (\log \ns)^{O(1)}$ (except in \cref{thm:poly-time-intro}), which is information-theoretically necessary, up to the factor $(\log \ns)^{O(1)}$.\todonote{@shyam: can we add your way of removing this assumption in a note somewhere. LOL not anymore...}

\begin{theorem}[Obliviously-robust mean testing (Informal; see~\cref{thm:oblivious-ub,thm:huber-lb,thm:oblivious-lb})]
    \label{theo:oblivious:informal}
    In the \emph{oblivious} contamination model, there is a mean tester which is robust to $\cor$-contamination, which uses
    \begin{equation}
        \label{eq:main:oblivious:ub}
        \tilde{\Theta}\Paren{\max\Paren{\frac{\sqrt{\dims}}{\dst^2}, \frac{\dims\cor^3}{\dst^4},\min\Paren{\frac{\dims^{2/3}\cor^{2/3}}{\dst^{8/3}}, \frac{\dims \cor}{\dst^2}}} }\,,
    \end{equation}
    samples in the \emph{oblivious} contamination model, and this is information-theoretically tight up to logarithmic factors.
     Moreover, $\tilde{\Omega}\Paren{\max\Paren{\frac{\sqrt{\dims}}{\dst^2}, \frac{\dims\cor^3}{\dst^4}}}$ samples are needed even in the weaker Huber contamination model.
\end{theorem}
We offer a little interpretation of the (surprisingly complex) expression \eqref{eq:main:oblivious:ub}.
If $\dims$ dominates the other parameters, \ie $\dims \gg  1/{\poly(\dst)}, 1/{\poly(\cor)}$, then $\tfrac{\dims \cor^3}{\dst^4}$ is the dominant term.
But if $d,1/{\dst}, 1/{\cor}$ are within small polynomial factors, any of the four terms in \eqref{eq:main:oblivious:ub} can dominate.
\begin{figure}[H]\centering
\label{fig:oblivious:complexity}
\sidebysidecaption{0.555\linewidth}{0.42\linewidth}{\includegraphics[width=1.0\textwidth]{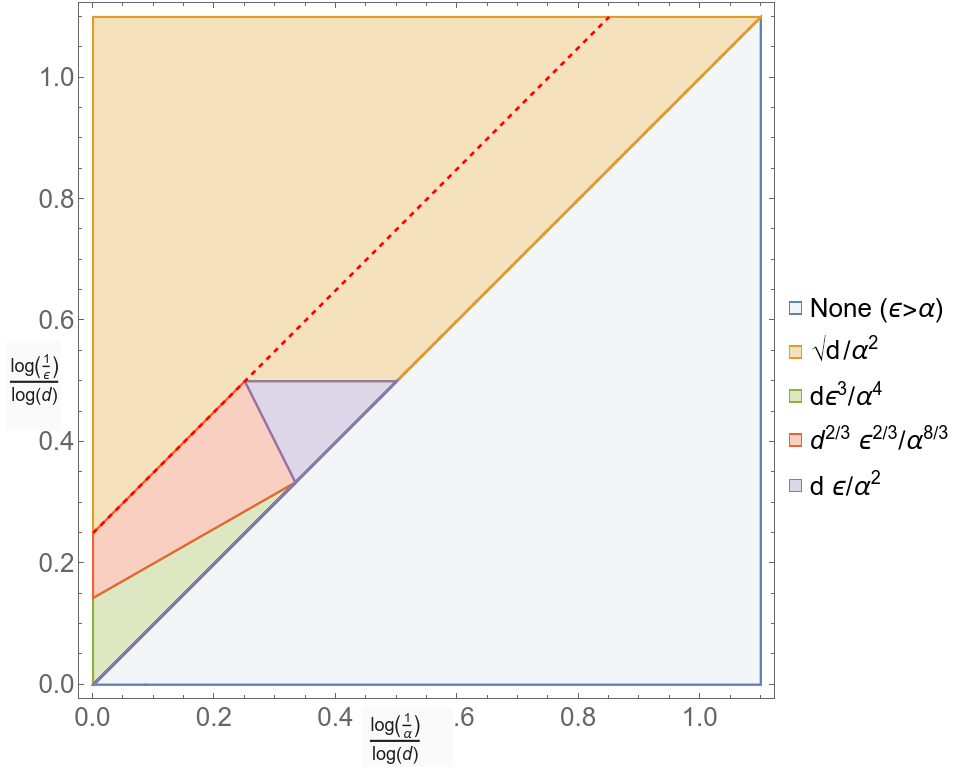}}{\caption{\small{}The various phases of the sample complexity of robust mean testing in the oblivious contamination model, as stated in~\cref{theo:oblivious:informal}: each area of this plot corresponds to which term of the sample complexity dominates, as a function of $\dims,\cor,\dst$. The separation between adaptive and oblivious contamination occurs at the red dashed line (to the right, the oblivious sample complexity is strictly smaller). The lower half corresponds to $\dst < \cor$, where testing is information-theoretically impossible.}
}
\end{figure}
\vspace{-1em} 
To see that \cref{theo:oblivious:informal} implies a strict separation between the oblivious and adaptive models, we recall:
\begin{theorem}[\cite{narayanan2022privatetesting}, see also \cref{thm:adaptive-sample-complexity-main}]\todonote{Alternatively, "against adaptive adversaries" saves some vertical space}\label{theo:adaptive:informal}
    In the \emph{adaptive} contamination model, the optimal sample complexity of $\cor$-robust mean testing is 
    \begin{equation}
        \label{eq:main:adaptive:ub}
        \tilde{\Theta}\Paren{\max\Paren{\frac{\sqrt{\dims}}{\dst^2}, \frac{\dims\cor^2}{\dst^4}}}
    \end{equation}
\end{theorem}
\noindent The sample complexity \eqref{eq:main:oblivious:ub} is strictly smaller than \eqref{eq:main:adaptive:ub} for a range of vanishing $\cor,\dst$, \eg with $\cor = \bigOmega{\frac{\dst}{\dims^{1/4}}}$.

For completeness, in \cref{sec:strong-sample-complexity} we show explicitly how to obtain \cref{theo:adaptive:informal} by combining Narayanan's result on differentially-private mean testing with known robust-privacy equivalence results (as in e.g. \cite{georgiev2022privatetorobust, hopkins2023robusttoprivate, asi2023robusttoprivate}). We further conjecture that a similar separation holds between the oblivious and Huber contamination models; to establish such a separation, it would be enough to prove a (non-efficient) $\tilde{O}(\max({\sqrt{\dims}}/{\dst^2}, {\dims\cor^3}/{\dst^4}))$ sample complexity upper bound in the latter, which in light of~\cref{theo:oblivious:informal} would be nearly tight. We leave this as an interesting open problem.

A subtle difference between our strong and oblivious contamination models concerns which ``good'' samples are \emph{removed} by the adversary.
In the strong model, the adversary chooses adaptively which of the good samples to remove, whereas the oblivious adversary can only choose good samples to remove at random.
Thus, the oblivious adversary could be equivalently defined as merely \emph{adding} samples and doing no removals at all.
One might ask whether the separation in sample complexities we establish between adaptive and oblivious adversaries actually arises from the ability of the adaptive adversary to remove samples, rather than from adaptivity itself.\footnote{For instance, one could consider an oblivious adversary which is allowed to replace the good distribution $\cD$ with $\cD$ conditioned on any event of probability $1-\eps$, thus obliviously ``removing'' part of $\cD$. We thank Guy Blanc for pointing this out.}
We show in \cref{sec:strong-sample-complexity} that the lower bound of \cref{theo:adaptive:informal} actually holds even against adaptive adversaries that may only \emph{add} data points, meaning that the sample complexity separation between adaptive and oblivious adversaries really is caused by the difference in addaptivity for the \emph{added} samples.
This extension to additive-only adaptive adversaries also readily follows from results proven in~\cite{narayanan2022privatetesting}.
\smallskip

Turning now to efficient algorithms, we provide the first polynomial-time algorithm which nearly matches the optimal sample complexity in the adaptive model.
Prior to our work, the best polynomial-time approach was to learn the mean using $O({\dims}/{\dst^2})$ samples, or to apply a polynomial-time algorithm of Narayanan~\cite{narayanan2022privatetesting} which works only when $\cor \leq \alpha \cdot \dims^{-1/4}$.
\begin{theorem}[Adaptively-robust efficient mean testing (Informal; see \cref{thm:poly-time-main})]
\label{thm:poly-time-intro}
    In the {adaptive} contamination model, there is a \emph{quadratic-time} algorithm for $\cor$-robust mean testing with sample complexity $\tilde{O}(\max( \tfrac{\sqrt{\dims}}{\dst^2}, \tfrac{\dims \cor^2}{\dst^4}))$, as long as $\alpha \ge O(\eps \sqrt{\log (1/\eps)})$.
\end{theorem}
This computationally efficient analogue of \cref{theo:adaptive:informal} raises the question of whether a similar analogue of \cref{theo:oblivious:informal} is possible.
(The tester described in \cref{theo:oblivious:informal} relies on a computationally inefficient ``filtering step''; see~\cref{ssec:techniques}). 
Our next result shows strong evidence that this is not possible, and that the separation between adaptive and oblivious contamination models vanishes when restricting oneself to computationally efficient algorithms.
\begin{theorem}[Computational lower bound (Informal; see~\cref{thm:lb:lowdegree})]
    \label{theo:sos-lb:informal}
    In the oblivious contamination model, any $\cor$-robust \emph{low-degree} mean testing algorithm in the Huber contamination model has sample complexity
    \begin{equation}
        \bigOmega{\max\Paren{\frac{\sqrt{\dims}}{\dst^2}, \frac{\dims\cor^2}{\dst^4}}} \, .
    \end{equation}
\end{theorem}
\cref{theo:sos-lb:informal} complements a reduction from planted clique \cite{brennan2020reducibility} which suggests that $n^{\Omega(\log n)}$ time is required to beat $\tfrac{\dims \cor^2}{\dst^4}$ samples, even in the Huber model.
The quantitative version of our result (\cref{thm:lb:lowdegree}) suggests something stronger (albeit for a restricted class of algorithms, rather than via reduction) -- namely, that $\exp(\ns^{\Omega(1)})$ time is needed to use $(\tfrac{\dims \cor^2}{\dst^4})^{1 - \Omega(1)}$ samples, even in the Huber model.
We hope that our results, by uncovering a richer landscape in robust statistics than previously known and showing that the choice of contamination setting is much less innocuous than commonly believed, will spark interest in revisiting these modelling assumptions for various other tasks.

\subsection{Related Work}
\label{sec:related}
\noindent\textbf{Gaussian Mean Testing.}
Gaussian mean testing is known in statistics as the Gaussian sequence model~\cite{ermakov1991minimax,baraud2002non,ingster2003nonparametric}; the understanding that it is possible to use fewer samples than dimensions appears relatively recent \cite{srivastava2008test}.
A recent influential work, \cite{DiakonikolasKS17}, records the sample-optimal mean tester and the ``folklore'' $\Omega(\sqrt{\dims}/\dst^2)$ lower bound, and initiates the study of the complexity of \emph{robust} mean testing.
More recent work focuses on variants such as mean testing under a sparsity assumption \cite{GeorgeC22}, testing with unknown covariance \cite{CanonneCKLW21,DiakonikolasKP23},
 testing subject to differential privacy \cite{Canonne0MUZ20,narayanan2022privatetesting}, robustly testing the covariance \cite{DiakonikolasK21}, or (distributed) testing giving partial observations from each sample~\cite{pmlr-v125-acharya20b,SzaboVZ22}.

\medskip\noindent\textbf{(Algorithmic) Robust Statistics.}
Algorithmic robust statistics, especially in high dimensions, has experienced a recent renaissance following a range of algorithmic breakthroughs; see the book \cite{DiakonikolasK:ARS}.
Robust mean \emph{estimation} has played a fundamental role; the quest for efficient algorithms for robust mean estimation led to the invention of the \emph{filter} technique~\cite{diakonikolas2019robust}.

\medskip\noindent\textbf{Connection to (Differential) Privacy.} A recent line of work~\cite{georgiev2022privatetorobust,hopkins2023robusttoprivate,asi2023robusttoprivate} established a (two-way) correspondence between adversarially robust and differentially private algorithms for a range of tasks, a connection we use to obtain~\cref{theo:adaptive:informal}. Importantly, this correspondence applies to \emph{adaptive} adversaries, and does not, to the best of our knowledge, differentiate between oblivious and adaptive adversaries.

\medskip\noindent\textbf{Noise Models in Statistics and Learning.}
Many developments in computational learning theory have been guided by the mission to design algorithms which work in an array of noise models \cite{balcan2020noise}.
For instance, the statistical query model was invented to capture a class of PAC learning algorithms which tolerate \emph{random classification noise} \cite{kearns1998efficient}.
A full survey is out of scope, but some highlights include \emph{nasty noise}, which is essentially the adaptive contamination model we consider here \cite{bshouty2002pac,diakonikolas2018learning}, and Massart noise, which has led to exciting recent algorithmic advances~\cite{diakonikolas2019distribution,diakonikolas2022cryptographic,nasser2022optimal}.
While \emph{computational} separations are known between these noise models in classification settings (e.g., random classification noise is much easier to handle algorithmically than adversarial label noise), separations in sample complexity seem unlikely, because empirical risk minimization handles even the nastiest noise models.

Two works in particular study questions related to ours.
First, \cite{blanc2022power} shows some equivalences between adaptive and oblivious adversaries up to polynomial factors in sample complexity, for restricted classes of algorithms (SQ) or adversaries (additive).
\cite{diakonikolas2018robustly,DiakonikolasKS17} together show a computational separation between what error $\dst$ is achievable for robustly learning a high-dimensional Gaussian when the adversary can only add samples versus when they can add and remove samples. We emphasize that while previous work showed evidence for a computational gap, we believe ours is the first demonstration of an (unconditional) information-theoretic separation in a natural robust statistics setting.

\subsection{Overview of Techniques}
\label{ssec:techniques}

\subsubsection{Exploiting Obliviousness to Robustly Test with Fewer Samples}

\paragraph{Our Approach.}
We focus first on our main technical contribution, the mean tester from \cref{theo:oblivious:informal}.
To get an improved testing algorithm for oblivious contaminations (compared to adaptive contaminations), we need to exploit that the adversary must commit to the contaminated points before the remaining datapoints are drawn.
A consequence is that the correlation between the sums of good points ($G$) and bad points ($B$) is comparable to independent random vectors of comparable norm:
\[
\ip{\frac{\sum_{i \in B} X_i}{\normtwo{\sum_{i \in B} X_i}}}{\frac{\sum_{i \in G} X_i}{\normtwo{\sum_{i \in G} X_i}}} \approx \frac{\pm 1}{\sqrt{\dims}} \, .
\]
By contrast, an adaptive adversary can make this correlation as large as $1$.

Hence,
the only way the adversary can have a substantial effect on $\big\|\sum_{i \in [\ns]} X_i\big\|_2$ is by making $\normtwo{\sum_{i \in B} X_i}$ larger than it would be for a set of $\cor \ns$ good samples.
Building on this idea, we can design a tester using $\tilde{\Theta}\Paren{\max\Paren{\frac{\sqrt{\dims}}{\dst^2}, \frac{\dims\cor^3}{\dst^4}, \min \Paren{\frac{\dims^{2/3}\cor^{2/3}}{\dst^{8/3}}, \frac{\dims \cor}{\dst^2}}}}$ samples under (roughly) the additional assumption that the sum of every subset of the adversary's vectors has about the same norm it would if the samples were uncorrupted.

The second challenge is to remove this additional assumption.
The standard approach in robust statistics to make bad samples ``look like'' good ones according to some tests (e.g. norms of sums of subsets of points) is to remove samples in subsets which violate those tests; this is often called ``filtering''.
This risks removing about $\cor \ns$ good samples as well, but in many settings this isn't an issue.

However, \emph{removing any good samples after looking at all the samples potentially breaks obliviousness by introducing dependencies between good and bad samples!}
We develop a novel \emph{obliviousness-preserving} filtering technique.
We (iteratively) split the samples into two subsets, $U,V$.
Looking only at $U$, we devise a rule for which samples to keep and which to remove (keeping those contained in a certain intersection of halfspaces); then we apply this rule to $V$ and show that it preserves obliviousness while ensuring that $V$ now satisfies the  assumption about sums of subsets of corrupted vectors. 
We turn now to a more detailed overview.

\medskip\noindent\textbf{Background: Narayanan's Robust Tester.} To understand quantitatively how we can exploit obliviousness of the adversary, we first review a robust mean tester which uses $\tilde{O}(\max( {\sqrt{\dims}}/{\dst^2}, {d \cor^2}/{\dst^4}))$ samples in the strong contamination model, as long as $\cor \ll \dst$ (all of which is information-theoretically necessary).\footnote{A similar tester can be extracted from \cite{narayanan2022privatetesting}.
While Narayanan's paper focuses on differentially private mean testing, the tester can be shown to be robust by virtue of its privacy guarantees; see \cref{sec:strong-sample-complexity}. The tester we describe here is simpler than Narayanan's original tester, in part because we need only robustness, not privacy.}
Our polynomial-time algorithm is also an adaptation of the following robust tester.

As in many robust statistics settings, the overall scheme relies on finding a ``good enough'' subset of $(1-\cor)\ns$ samples $S \subseteq [\ns]$, to then apply a non-robust algorithm on $S$ -- in this case, the simple tester based on $\normtwo{\sum_{i \in S} X_i}^2$.
For $X_1,\ldots,X_{\ns} \in\R^{\dims}$ which are clear from context and $T \subseteq [\ns]$, let $\Sum(T) = \sum_{i \in T} X_i$.
\begin{definition}[Good Enough Subset (Informal)]
\label{def:intro-good-enough}
    For $X_1,\ldots,X_\ns \in \R^\dims$, we say $S \subseteq [\ns]$, $|S| = (1-\cor)\ns$ is \emph{good enough} if, for every $T \subseteq S$ with $|T| \leq \cor \ns$,
\[
    \normtwo{ \Sum(T) }^2 \leq |T|\dims + \tilde{O}(\cor^{1.5} \ns^{1.5} \sqrt{\dims} + \cor^2 \ns^2) \text{ and } \abs{\ip{\Sum(S \setminus T)} {\Sum(T) }} \leq \tilde{O}( \cor \ns^{1.5} \sqrt{\dims} + \cor^2 \ns^2 ) \,.
    \]
\end{definition}\noindent
The choice of parameters in the definition guarantees that any subset of size $(1-\cor)\ns$ of $\ns$ independent samples from $\cN(0,I)$ or $\cN(\mu,I)$, for small-enough $\mu$, is good enough with high probability.
To see why this holds intuitively, observe that if $S$ consists of good samples only, then
$\abs{\ip{\Sum(S \setminus T)}{\Sum(T)}}$ is roughly distributed as $\cN(0, \cor \ns^2 \dims)$, and we need a union bound over $\approx \ns^{\cor \ns}$ choices of $T$.

\begin{definition}[Narayanan's tester]Given $\ns$ $\cor$-contaminated samples, Narayanan's tester finds any good enough subset $S$ and outputs $\mathbf{H}_0$ if $\normtwo{\Sum(S)}^2 - (1-\cor)\ns \dims \ll \dst^2 \ns^2$ and $\mathbf{H}_1$ otherwise. \vspace{-.75em}
\end{definition}
\subparagraph{Analysis Sketch.} Let $X_1,\ldots,X_\ns$ be an $\cor$-contaminated draw from either $\cN(0,I)$ or $\cN(\mu,I)$ for some $\normtwo{\mu} = \dst$.
Let $G \subseteq [\ns]$ be the uncorrupted samples.
(For simplicity, in this overview we assume the adversary has only added samples; removed samples can be handled without much more difficulty.)\cmargin{Removed: ``also for simplicity we will assume $\dst \leq 1$'', isn't it everywhere?}
Let $S \subseteq [\ns]$ be any good enough subset; we want to show $\normtwo{\Sum(S)}^2 - (1-\cor)\dims \geq \Omega(\dst^2 \ns^2)$ in the alternative case, and $\normtwo{\Sum(S)}^2 - (1-\cor)\dims \ll \dst^2 \ns^2$ in the null.
First,
\[
  \E \normtwo{\Sum(G)}^2 - (1-\cor) \dims = \begin{cases}
  \E \sum_{i \neq j \in G} \ip{X_i}{X_j} \approx \dst^2 \ns^2 & \text{ in the alternative case}\\
  0 & \text{ in the null case}
  \end{cases}
\]
and standard concentration arguments show that this holds with high probability so long as $\ns \gg \sqrt{\dims}/\dst^2$.
So we just have to show that $|\normtwo{\Sum(S)}^2 - \normtwo{\Sum(G)}^2| \ll \dst^2 \ns^2$.
This is doable using the following lemma.

\begin{lemma}[Main Lemma for Narayanan's Tester]
  For any two good-enough subsets $S,S'$ of $X_1,\ldots,X_{\ns} \in \R^{\dims}$, $\abs{\normtwo{\Sum(S)}^2 - \normtwo{\Sum(S')}^2} \ll \dst^2 \ns^2$, so long as $\ns \gg \dims \cor^2 / \dst^4$.
\end{lemma}
\begin{proof}
  We divide $S$ into $S \cap S'$ and $S \setminus S'$ and $S'$ into $S' \cap S$ and $S' \setminus S$, so we have
  \begin{align*}
  \normtwo{\Sum(S)}^2 - \normtwo{\Sum(S')}^2 & = \normtwo{\Sum(S \cap S')}^2 + 2 \ip{\Sum(S \cap S')}{\Sum(S \setminus S')} + \normtwo{\Sum(S \setminus S')}^2 \\
  & - \normtwo{\Sum(S' \cap S)}^2 - 2 \ip{\Sum(S' \cap S)}{\Sum(S' \setminus S)} - \normtwo{\Sum(S' \setminus S)}^2 \, .
  \end{align*}
  Now, $\normtwo{\Sum(S \cap S')}^2 - \normtwo{\Sum(S' \cap S)}^2 = 0$, and since $|S \setminus S'| = |S' \setminus S|$, also $|\normtwo{\Sum(S \setminus S')}^2 - \normtwo{\Sum(S' \setminus S)}^2| \leq \tilde{O}(\cor^{1.5} \ns^{1.5} \sqrt{\dims} + \cor^2 \ns^2)$, using good-enough-ness.
  By using good-enough-ness again, both $\abs{\ip{\Sum(S \cap S')} {\Sum(S \setminus S')}}$ and $\abs{\ip{\Sum(S' \cap S)}{\Sum(S' \setminus S)}}$ are at most $\tilde{O}(\cor \ns^{1.5} \sqrt{\dims} + \cor^2 \ns^2)$.
  Since $\cor \ll \dst$, we have $\cor^2 \ns^2 \ll \dst^2 \ns^2$, and since $\ns \gg \dims \cor^2 / \dst^4$, we have $\cor \ns^{1.5} \sqrt{\dims} \ll \dst^2 \ns^2$. 
\end{proof}

\noindent This completes the analysis of Narayanan's tester.
We record two important observations:
\begin{enumerate}
    \item The reason that the tester requires $\dims \cor^2 / \dst^4$ samples lies in the term $\ip{\Sum(S \cap S')}{\Sum(S \setminus S')}$.
    Let's think of $S' = G$, the good samples, and $S$ as some good-enough subset which contains around $\cor \ns$ corrupted samples, $S \setminus G$.
    The adaptive adversary could choose the samples in $S \setminus G$ to make $\Sum(S \setminus G)$ too (anti)-correlated with $\Sum(S \cap G)$.
    There is a limit to how large he can make the (anti)correlation before $S$ is no longer ``good enough'' -- namely, he can make $\ip{\Sum(S \cap G)}{\Sum(S \setminus G)}$ as large as the largest inner product of the form $\ip{\Sum(G \setminus T)}{\Sum(T)}$ for $T \subseteq G$ with $|T| = \cor \ns$, which is around $\cor \ns^{1.5} \sqrt{\dims}$ by standard concentration.

    \item Narayanan's tester requires finding a good-enough subset of $(1-\cor)\ns$ samples; \emph{prima facie} this requires exponential-time brute-force search, but we describe a polynomial-time variant of his approach later.
\end{enumerate}

\medskip\noindent\textbf{Using Only $\dims \cor / \dst^2$ Samples if the Adversary is Oblivious and Not ``Too Big''.}
Narayanan's tester is information-theoretically optimal (up to log factors) against adaptive adversaries.
As our first taste of improved testing against an oblivious adversary, consider the following toy setup.
Suppose the adversary is not only oblivious but also promises us that the $\cor \ns \dims$ bad samples $B$ will satisfy $\normtwo{\Sum(B)}^2 \leq O(\cor \ns \dims)$; roughly, this constraints the adversary to add $\cor n$ vectors of norm $\sqrt{\dims}$ which are approximately pairwise orthogonal.
(If the adversary adds any vector of norm much larger, we can remove it before proceeding.)
We will show how to test using $\sqrt{\dims}/\dst^2 + \dims \cor / \dst^2$ samples, improving on Narayanan's tester for $\cor \gg \dst^2$.

We revisit the simple tester using just $\normtwo{\Sum([n])}^2$.
Dividing $[n]$ into good and corrupted samples $G, B$,
\[
\normtwo{\Sum([n])}^2  - \ns \dims = \Paren{\normtwo{\Sum(G)}^2 - (1-\cor) \ns \dims} + 2 \ip{\Sum(G)}{\Sum(B)} + \normtwo{\Sum(B)}^2 - \cor \ns \dims \, .
\]
As usual, $\normtwo{\Sum(G)}^2 - (1-\cor) \ns \dims \geq \Omega(\dst^2 \ns^2)$ in the alternative case and $\ll \dst^2 \ns^2$ in the null; we want to show the remaining terms are $\ll \dst^2 \ns^2$ in magnitude.
Trivially, $|\normtwo{\Sum(B)}^2 - \cor \ns \dims| \leq O(\cor \ns \dims) \ll \dst^2 \ns^2$ when $\dims \cor / \dst^2 \ll \ns$, using our promise on $\normtwo{\Sum(B)}^2$.

Now let's look at the term where we make the improvement over Narayanan's tester: $\ip{\Sum(G)}{\Sum(B)}$; we are looking to use obliviousness to beat the bound $\cor \ns^{1.5} \sqrt{\dims}$.
We fix $\Sum(B)$ and then sample the random vector $\Sum(G)$, which is distributed either as $\cN(0,(1-\cor) \ns I)$ or $\cN((1-\cor) \ns \mu, (1-\cor) \ns I)$, meaning
\[
\ip{\Sum(G)}{\Sum(B)} \sim \begin{cases}
  \cN \Paren{ (1-\cor) \ns \ip{\mu}{\Sum(B)}, (1-\cor)\ns \normtwo{\Sum(B)}^2}  & \text{ in the alternative case}\\
  \cN \Paren{ 0, (1-\cor)\ns \normtwo{\Sum(B)}^2} & \text{ in the null case}
\end{cases} \, .
\]
So, $\abs{\ip{\Sum(G)}{\Sum(B)}} \leq O(\ns \dst \cdot \sqrt{\cor \ns \dims} + \ns \sqrt{\cor \dims}) \ll \dst^2 \ns^2$, as $\normtwo{\Sum(B)}^2 \leq O(\cor \ns \dims)$ and $\ns \gg \dims \cor / \dst^2$.

\noindent From this simple reasoning, we draw the following important conclusion:
\begin{quote}
    \begin{center}
        \emph{If the adversary is oblivious and is constrained to add samples $B$ which aren't ``too big'', then we can test using fewer samples than against an adaptive adversary.}
    \end{center}
\end{quote}

This leads us to two key questions, whose answers form the main technical ingredients in our oblivious tester.
Can we take an obliviously-corrupted dataset and remove samples in some way to ensure that in the resulting \emph{filtered} dataset, the adversary has added samples $B$ which aren't ``too big'', but do so in a way which doesn't introduce dependencies between good and bad samples which would break the obliviousness we're relying on?
And, what is the right definition for ``too big'' -- could a more refined definition lead to a tester using fewer than $\dims \cor / \dst^2$ samples? 

\medskip\noindent\textbf{Friendly Oblivious Adversaries and The Sum+Variance Tester.}
We will tackle the above questions in reverse order.
We introduce a key definition:
\begin{definition}[Informal, see~\cref{ass:friendly}] A \emph{friendly} oblivious adversary introduces $\{X_i\}_{i \in B}$ such that
\begin{enumerate}
        \item \label{itm:intro-friendliness-1} For disjoint $S,T \subseteq B$ with $|S|, |T| \leq \cor \ns$, $\abs{\ip{\Sum(S)}{\Sum(T)}} \leq \tilde{O}(\sqrt{|S| \cdot |T|} \cdot (\sqrt{\cor \ns \dims} + \cor \ns ))$.
        \item \label{itm:intro-friendliness-2} For distinct $i,j \in B$, $\abs{\ip{X_i}{X_j}} \leq \tilde{O}(\sqrt{\dims})$, and for every $i \in B$, $\normtwo{X_i}^2 = \dims \pm \tilde{O}(\sqrt{\dims})$.
    \end{enumerate}
The parameters are chosen so that every pair of subsets $S,T$ of \emph{good} samples would satisfy these conditions.
\end{definition}
To clarify why friendliness refines the ``not too big'' condition $\normtwo{\Sum(B)}^2 \leq O(\cor \ns \dims)$ from above, observe that subject to friendliness, for any $S \subseteq B$,
\[
\normtwo{\Sum(S)}^2 = |S| \cdot (\dims \pm \tilde{O}(\sqrt{\dims})) + O \Paren{ \E_{S_1,S_2} \ip{\Sum(S_1)}{\Sum(S_2)}} = \dims |S| + \tilde{O}(|S| \sqrt{\cor \ns \dims} + |S| \sqrt{\dims})
\]
where $S_1,S_2$ is a random partition of $S$.
In particular, $\normtwo{\Sum(B)}^2 = \cor \ns \dims \pm o(\dst^2 \ns^2)$ whenever $\ns \gg \dims \cor^3 / \dst^4$.

Now we can introduce our robust mean tester which uses $\tfrac{\sqrt{\dims}}{\dst^2} + \tfrac{\dims \cor^3}{\dst^4} + \tfrac{\dims^{2/3} \cor^{2/3}}{\dst^2}$ samples (up to log factors) in the presence of a friendly oblivious adversary.

\medskip\noindent\textbf{The Sum+Variance Tester} (Algorithm~\ref{alg:oblivious-tester}): Given $X_1,\ldots,X_{\ns} \in \R^{\dims}$, if $\normtwo{\Sum([\ns])}^2 - \ns \dims \geq \Omega(\dst^2 \ns^2)$, or if
    \[
      \frac 1 {\ns} \sum_{i \in [\ns]} \Paren{\frac{\ip{X_i}{\Sum([\ns])} - \dims }{\normtwo{\Sum([\ns])}}}^2 \geq 1 + \Omega \Paren{ \frac{\dst^4 \ns}{\cor \dims}} \, ,
    \]
    return $\mathbf{H}_1$, otherwise return $\mathbf{H}_0$.

\vspace{-0.5em}

\subparagraph{Analysis Sketch.}
For starters, we need to make sure that in the null case, $\normtwo{\Sum([\ns])}^2 - \ns \dims \ll \dst^2 \ns^2$.
Splitting $S$ into good samples $G$ and corrupted samples $B$, we know $\normtwo{\Sum(G)}^2 = (1-\cor) \ns \dims \pm O(\ns \sqrt{\dims})$ and $|\ip{\Sum(G)}{\Sum(B)}| \leq O(\ns \sqrt{\cor \dims})$ using standard concentration tools and obliviousness, and $\normtwo{\Sum(B)}^2 = \cor \ns \dims + \tilde{O}(\cor^{1.5} \ns^{1.5} \sqrt{\dims} + \cor \ns \sqrt{\dims})$ by friendliness.
All together,
\begin{align*}
\normtwo{\Sum([\ns])}^2 - nd & = \normtwo{\Sum(G)}^2 + 2 \ip{\Sum(G)}{\Sum(B)} + \normtwo{\Sum(B)}^2 - nd = \tilde{O}(\ns \sqrt{\dims} +  \cor^{1.5} \ns^{1.5} \sqrt{\dims})
\end{align*} 
which is at most $\dst^2 \ns^2$ exactly when $\ns \gg \tfrac{\sqrt{\dims}}{\dst^2} + \tfrac{\dims \cor^3}{\dst^4}$.

Ideally, we would show next that in the alternative case $\normtwo{\Sum([\ns])}^2 - \ns \dims \gg \dst^2 \ns^2$, but even a friendly, oblivious adversary can ensure this doesn't happen when $\ns \ll \tfrac{\dims \cor}{ \dst^2}$.
With knowledge of the vector $\mu$, he can introduce samples $\{X_i\}_{i \in B}$ such that $\ip{X_i}{\mu} \approx - \tfrac{\dst^2}{\cor}$, which introduces cancellations with $\E \Sum(G)$ that reduce $\normtwo{\Sum([\ns])}^2$.
Overall, he can ensure $\abs{\normtwo{\Sum([\ns])}^2 - \ns \dims} \ll \dst^2 \ns^2$.

But now we encounter a typical theme in robust statistics: the adversary has had to introduce a small set of $X_i$'s such that $\ip{X_i}{\Sum([\ns])}$ is more negative than typical, thereby increasing the variance among $\{ \ip{X_i}{\Sum([\ns])} \}_{i \in [\ns]}$.
For $i \in B$, we expect $\ip{X_i}{\Sum([\ns])}$ to be $\tfrac{\ns \dst^2}{\cor}$ smaller than usual, so heuristically, 
\[
\frac 1 {\ns} \sum_{i \in B} \Paren{\frac{\ip{X_i}{\Sum([\ns])} - \dims }{\normtwo{\Sum([\ns])}}}^2 \gtrsim \frac 1 {\ns} \cdot \cor \ns \cdot \frac{\dst^4 \ns^2}{\cor^2 \ns \dims} = \frac{\dst^4\ns}{\cor \dims} \, ,
\]
where we used $\normtwo{\Sum([\ns])}^2 \approx \ns \dims$.
Adding the contribution from the samples in $G$ gives us $1 + \Omega(\tfrac{\dst^4 \ns}{\cor \dims})$.
We make this idea rigorous in \cref{sec:obliviously-robust-tester}.

Of course, outputting $\mathbf{H}_1$ when $\tfrac 1 {\ns} \sum_{i \in B} \Paren{\frac{\ip{X_i}{\Sum([\ns])} - \dims }{\normtwo{\Sum([\ns])}}}^2 = 1 + \Omega(\tfrac{\dst^4 \ns}{\cor \dims})$ only makes sense if the adversary cannot make this happen in the null model.
We show (\cref{sec:oblivious-tester-null-variance}) that no friendly oblivious adversary can make $\tfrac 1 {\ns} \sum_{i \in B} \Paren{\frac{\ip{X_i}{\Sum([\ns])} - \dims }{\normtwo{\Sum([\ns])}}}^2 = 1 + \Omega(\tfrac{\dst^4 \ns}{\cor \dims})$ if $n \gg \tfrac{\dims^{2/3} \cor^{2/3}}{\dst^{8/3}}$.\todonote{sam: should add an explanation of this variance bound here. (except it's not gonna fit)}

\medskip\noindent\textbf{Friendliness via Obliviousness-Preserving Filtering.} 
We're still missing a key ingredient: how can we force an oblivious adversary to be friendly?
Ensuring condition~\ref{itm:intro-friendliness-2} of friendliness is straightforward.
If we see any $|\normtwo{X_i}^2 -\dims| \gg \sqrt{\dims}$, that $X_i$ must have been introduced by the adversary and can be safely removed, and similarly if any pair $i,j$ has $|\ip{X_i}{X_j}| \gg \sqrt{\dims}$ then (by obliviousness) both $X_i,X_j$ must be corrupted samples and can be removed.
(We are using $\gg$ to hide logarithmic factors.)

But what about condition \ref{itm:intro-friendliness-1}?
A natural idea is to preprocess $X_1,\ldots,X_{\ns}$ by removing any subsets $S,T$ of size at most $\cor \ns$ which violate condition \ref{itm:intro-friendliness-1}.
If we had a subset $S$ which grossly violated \ref{itm:intro-friendliness-1} in the sense that $\normtwo{\Sum(S)}^2 \geq 100 \cor \ns \dims$, we could conclude that $S$ contains at least $99\%$ bad samples.
This might seem good enough -- indeed, a common paradigm in robust statistics is \emph{filtering}, removing samples in way which removes at least as many bad samples as good ones, since any such procedure can ultimately remove at most $\cor \ns$ good samples. 
\emph{However, removing any good samples after looking at all the samples, including the corrupted ones, creates dependencies between good and bad samples, thus breaking obliviousness!}

\vspace{-0.75em}

\subparagraph{Sample-Splitting to Preserve Obliviousness.}
We introduce an \emph{obliviousness-preserving} filter. We: 
\begin{enumerate}
    \item Randomly split $X_1,\ldots,X_{\ns}$ into $U$ and $V$.
    \item Using only $U$, identify a set of unit vectors $v_1,\ldots,v_\ell \in \R^{\dims}$. \label{itm:intro-find-directions}
    \item For all $j \leq \ell$, remove from $V$ any $X_i$ such that $|\ip{X_i}{v_j}| \gg \sqrt{\log \ns}$, then return $V$.
\end{enumerate}
The idea is that the returned $V$ will (with high probability) be a set of samples corrupted by a friendly oblivious adversary.
The threshold $\sqrt{\log \ns}$ is chosen so that with high probability no good sample is removed from $V$.
This means that with high probability the scheme preserves obliviousness, since we could have gotten the same outcome by drawing the good samples in $V$ only after performing filtering.\footnote{In reality we will perform several rounds of obliviousness-preserving filtering, splitting $V$ again into $U',V'$ and so on; as rounds progress we ensure friendliness for pairs of subsets $S,T$ of increasing size.
We will ignore this detail in our technical overview.}

The challenge is ensuring friendliness, which of course rests on the implementation of step~\ref{itm:intro-find-directions}.
In this step, the basic idea is to find a family of subsets $T_1,\ldots,T_\ell \subseteq U$ such that, for each $i\in[\ell]$,
\begin{itemize}
    \item $|T_i| \ll \cor \ns / (\log \ns)^{O(1)}$ (here $\ll$ hides constants; the $(\log \ns)^{O(1)}$ is crucial, as explained below), and
    \item if we choose $v_i = \Sum(T_i)/\normtwo{\Sum(T_i)}$ and remove from $U$ any $X_j$ such that $\iprod{X_j, v_i} \gg \sqrt{\log \ns}$, then $U$ satisfies condition~\ref{itm:intro-friendliness-1} of friendliness.
    If this happens, we'll say that $T_1,\ldots,T_m$ ``cleans'' $U$.
\end{itemize}

We need to establish two things: first, that such a family $T_1,\ldots,T_\ell$ which cleans $U$ exists, and second, with high probability over the random split $U,V$, any $T_1,\ldots,T_\ell \subseteq \{X_1,\ldots,X_{\ns}\}$ which cleans $U$ also cleans $V$.
However, these are in tension.
For the first, we would like to be able to choose the sets $T_1,\ldots,T_\ell$ as large as possible, as this gives more flexibility in the choice of filtering directions and hence makes it easier to clean $U$.
But, for the second, we need tight control over how many different choices of $T_1,\ldots,T_\ell$ the cleaning algorithm could make, because we will need to make a union bound over all such choices; the smaller the sets $T_1,\ldots,T_\ell$ have to be, the fewer choices there are.

\vspace{-0.75em}

\subparagraph{Compression and Small Witnesses.}
The key idea to balance these concerns is to show that if $S_1,S_2$ violate $\theta$-friendliness condition~\ref{itm:intro-friendliness-1}, then we can compress $S_1$ to a smaller set $S_1'$ such that removing all $X_i \in S_2$ with $\Big\langle X_i,\tfrac{\Sum(S_1')}{\normtwo{\Sum(S_1')}}\Big\rangle$ makes progress in cleaning $U$, which means we can add $S_1'$ to our list of $T_i$s.
The following lemma shows this, as long as $S_1 \cup S_2$ already satisfy $\lambda$-friendliness for some $\lambda \gg \theta$ -- we will be able to ensure that they already do via induction.

\begin{lemma}[Small Witness Lemma, Basic Version of \cref{lem:filter-compression}]
\label{lem:small-witness-intro}
  Let $S_1,S_2 \subseteq \R^{\dims}$ have $|S_1|,|S_2| = \cor \ns$ and $\ip{\Sum(S_1)}{\Sum(S_2)} \geq \cor \ns \cdot \sqrt{\theta d}$.
  Suppose $S_1 \cup S_2$ is $\lambda$-friendly, for some $\lambda \gg \theta$, and that there is some parameter $C > 0$ such that $|\ip{X}{X'}| \leq \theta \sqrt{\dims}/C$ and $\|X_i\|^2 = d \pm \theta \sqrt{d} / C$ for all $X,X' \in S_1 \cup S_2$.
  Then there is $S_1' \subseteq S_1$ with $|S_1'| \leq \cor \ns /C$ and $\Omega(\cor \ns)$ vectors $X \in S_2$ such that
  $\Big\langle X_i,\tfrac{\Sum(S_1')}{\normtwo{\Sum(S_1')}}\Big\rangle \geq \Omega\Big(\sqrt{\tfrac{\theta}{C \cor \ns}}\Big)$.\end{lemma}
\vspace{-0.5em}
In \cref{lem:small-witness-intro}, we think of $\theta \approx \cor \ns (\log \ns)^{O(1)}$, so that $\ip{\Sum(S_1)}{\Sum(S_2)} \geq \cor \ns \sqrt{\theta d}$ is a violation of friendliness, and $C \approx (\log \ns)^{O(1)}$ so that $S_1'$ is significantly smaller than $S_1$. 
Proving \cref{lem:small-witness-intro} is outside the scope of this overview, but the strategy is to first show that a large number of vectors in $S_2$ are correlated with $\Sum(S_1)$ (\cref{claim:inner-product-uniformity}), and then show this is preserved when we replace $S_1$ with a random subset $S_1' \subset S_1$.
\cref{lem:small-witness-intro} shows that adding $S_1'$ to the list of $T_i$'s will result in removing $\Omega(\cor \ns)$ vectors; this can only happen $O(1)$ times before all bad samples would be removed, so that we can think of $\ell = O(1)$.

\vspace{-0.75em}

\subparagraph{Small Filters Generalize from $U$ to $V$.}
Lastly, we need to establish that, if we find a short list of small $T_1,\ldots,T_\ell$ which cleans $U$, then with high probability it also cleans $V$.
Consider the set $\mathcal{T}$ of all possible $(T_1,\ldots,T_\ell) \in \binom{\ns}{\cor \ns / (\log \ns)^{O(1)}}^\ell$; note that $|\mathcal{T}| \leq 2^{\cor \ns / (\log \ns)^{O(1)}}$ because $\ell = O(1)$.

Fixing some $(T_1,\ldots,T_\ell) \in \mathcal{T}$, our goal is to show that with probability at least $1-2^{-\Omega(\cor \ns)}$ over the random split $U,V$, if $T_1,\ldots,T_\ell$ cleans $U$ then it cleans $[\ns]$; then we can take a union bound over all of $\mathcal{T}$.
By contrapositive, it is enough to show that, if after removing all $X_i$ from $X_1,\ldots,X_\ns$ such that $\ip{\Sum(T_j)}{X_i} \gg \sqrt{\log n}$, some subsets $S_1,S_2 \subseteq [\ns]$ remain which violate $\lambda$-friendliness, then with probability $1 - 2^{-\Omega(\cor \ns)}$ the random set $U$ also contains some $S_1',S_2'$ which violate $\theta$-friendliness, for some $\theta$ not too much less than $\lambda$.
(This distinction between $\theta,\lambda$ is the origin of the two different friendliness levels in the small witness lemma.)

For the latter, standard concentration arguments show that, with probability $1-2^{-\Omega(\cor \ns)}$, the offending sets $S_1,S_2$ get split evenly between $U$ and $V$, and this in turn is enough to show that some subsets of $U \cap S_1, U \cap S_2$ also violate friendliness.

\vspace{-0.75em}

\subsubsection{Lower Bounds}
\noindent\textbf{Information-Theoretic Lower Bound for Obliviously-Robust Testing.}
Among our lower bounds, the greatest conceptual innovation lies in our proof that robust mean testing with an oblivious adversary requires $\tilde{\Omega}\Paren{\min\Paren{\frac{\dims^{2/3}\cor^{2/3}}{\dst^{8/3}}, \frac{\dims \cor}{\dst^2}}}$ samples.
The remaining terms in the lower bound, $\tfrac{\sqrt{\dims}}{\alpha^2}$ and $\tfrac{d \cor^3}{\dst^4}$, come respectively from the complexity of non-robust mean testing and from a simpler argument using a Huber adversary, respectively. (The latter we describe below.)

To prove the lower bound, we will describe a distribution over mean vectors $\mu$ and adversarial vectors $\{X_i\}_{i \in B}$ such that the joint distribution of $\{X_i\}_{i \in B}$ together with $(1-\cor)n$ samples from $\cN(\mu,I)$ is close in total variation to $\cN(0,I)^{\otimes n}$.
The key trick in designing this distribution is to \emph{correlate}, but \emph{not perfectly align}, $\Sum(B)$ with $-\mu$.
Concretely, we:
\begin{enumerate}
    \item Draw $X_i \sim \cN(0,I)$ for $i \in B$.
    \item Draw $\mu = - \beta \Sum(B) - z$, where $\beta = \beta(\ns,\dims,\cor,\dst) > 0$ is a suitable constant and $z \sim \cN(0, \tfrac{\dst^2}{\dims} I)$.
\end{enumerate}
We show via direct calculation in \cref{sec:oblivious-lb} that the $\chi^2$ divergence, and hence total variation distance, between these two distributions on sets of $\ns$ samples is $o(1)$ so long as $\ns \ll \tilde{\Omega}\Paren{\min\Paren{\frac{\dims^{2/3}\cor^{2/3}}{\dst^{8/3}}, \frac{\dims \cor}{\dst^2}}}$.
The trick above of sampling the corrupted samples $\{X_i\}_{i \in B}$ \emph{before} drawing $\mu$ keeps these calculations tractable.

\medskip\noindent\textbf{Information-Theoretic Lower Bound for Huber-Robust Testing.}
Our final information-theoretic lower bound shows that $\Omega(\dims \cor^3 / \dst^4)$ samples are needed in the presence of a Huber adversary.
Here we borrow the lower-bound instance from \cite{DiakonikolasKS17} -- the adversary just adds samples from $\cN(-\beta \cdot \mu,I)$ for some well-chosen $\beta > 0$.
We tighten the analysis of this instance from \cite{DiakonikolasKS17} by using a \emph{conditional} second moment (a.k.a. conditional $\chi^2$ divergence) approach. (\cite{DiakonikolasKS17} use a vanilla $\chi^2$-divergence analysis of their lower bound instance; this method can prove at best a ${\dims \cor^4}/{\dst^4}$ lower bound, which they obtain.) 

\medskip\noindent\textbf{Low-Degree Lower Bound for Huber-Robust Testing.}
Finally, we show a \emph{low-degree} lower bound in the Huber model (essentially equivalent to an SQ lower bound \cite{brennan2020statistical}) using the same instance from~\cite{DiakonikolasKS17}; this is a direct computation using now-standard techniques from \cite{kunisky2022notes}.

\subsubsection{A Quadratic-Time Tester}
Now we turn to our quadratic-time algorithm for robust mean testing against adaptive adversaries using $\tfrac{\sqrt{\dims}}{\dst^2} + \tfrac{d \cor^2}{\dst^4}$ (up to logarithmic factors) samples, matching Narayanan's tester.
Up to logarithmic factors, our bound matches our low-degree lower bound mentioned above.
Together, these bounds give strong evidence that computationally bounded algorithms must pay a factor of $\tfrac{\dims \cor^2}{\dst^4}$ in the sample complexity, and therefore cannot witness the improved rates described elsewhere in this paper, for any model of contamination.
Recall that Narayanan's tester requires finding a good-enough subset (\cref{def:intro-good-enough}).
Since good-enough-ness involves all subsets of $\cor \ns$ samples, even checking whether some $S \subseteq [\ns]$ is good enough seems to require $\ns^{\cor \ns}$ time.

Borrowing a technique from the robust \emph{estimation}, we show that, at least for the good samples $G \subseteq [\ns]$, there's an efficiently-computable witness to their good-enough-ness.
This witness is the top eigenvalue of the covariance matrix $\E_{i \sim G} (X_i - \E_{j \sim G} X_j)(X_i - \E_{j \sim G} X_j)^\top$, together with a uniform upper bound on the magnitude of the row-sums of the Gram matrix of $\{X_i \, : \, i \in G\}$.

For illustration here, consider the null case and imagine that $\ns \leq \dims$.
Then it turns out to be nicer to consider the Gram matrix $M \in \R^{(1-\cor)\ns \times (1-\cor)\ns}$ with entries $M_{ij} = \ip{X_i}{X_j}$; up to zeros it has the same eigenvalues as the covariance.
Since $X_i \sim \cN(0,I)$ for $i \in G$, we have $M = d \cdot I \pm O(\sqrt{nd})$.
If $1_T$ is the $0/1$ indicator vector for $T \subseteq G$ with $|T| \leq \cor \ns$, then $1_T^\top M 1_T$ certifies the first part of good-enough-ness:
\[
  \normtwo{\Sum(T)}^2 = 1_T^\top M 1_T = d \cdot \normtwo{1_T}^2 \pm O(\sqrt{\ns \dims} \normtwo{1_T}^2) = |T| \dims + O(\cor \ns^{1.5} \sqrt{\dims}) \, .
\]

For the second part, note that $\ip{\Sum(G \setminus T)}{\sum(T)} \approx \sum_{i \in T} \sum_{j \neq i} M_{ij}$ is roughly the row-sums of the (off-diagonals of the) matrix $M$ for $i \in T$.
Each row sum is at most $\tilde{O}(\sqrt{nd})$, so the sum is $\tilde{O}(\cor \ns^{1.5} \sqrt{\dims})$.

These arguments (at least in the case $n \leq d$; $n > d$ is not very different) show that it is enough to find $S \subseteq [\ns]$ with $|S| = (1-\cor)\ns$ and whose Gram matrix has eigenvalues $d \pm O(\sqrt{nd})$ and off-diagonal row-sums at most $\tilde{O}(\cor \ns^{1.5} \sqrt{\dims})$.
In \cref{sec:poly-time} we design a filtering algorithm which does this by starting with $[n]$ and iteratively removing samples $X_i$ with large projection onto too-large or small eigenvectors of the Gram matrix, or whose row-sum is too large, until all the row-sums and eigenvalues are as we desire.

\section{Preliminaries and Notation}
\subsection{Basic Definitions}

Given two distributions $\cD_1, \cD_2$, we recall the definitions of total variation distance and $\chi^2$-divergence.

\begin{definition}[Total Variation Distance]
    Given two probability distributions $\cD_1, \cD_2$ over a measurable space $(\Omega, \mathcal{F})$, the \emph{total variation distance} between $\cD_1, \cD_2$, denoted $\dtv(\cD_1, \cD_2)$, is $\sup_{A \in \mathcal{F}} |\cD_1(A)-\cD_2(A)|$.
\end{definition}

\begin{definition}[$\chi^2$-divergence]
    Given two distributions $\cD_1, \cD_2$ over a measurable space $(\Omega, \mathcal{F})$ with well-defined probability density functions $p_1, p_2$, the \emph{$\chi^2$-divergence} between $\cD_1$ and $\cD_2$, denoted $\dchi(\cD_1 \| \cD_2)$, is given by $\BE_{X \sim \cD_2} \Paren{\frac{p_1(X)}{p_2(X)}-1}^2$. (Note that this is not symmetric in $\cD_1,\cD_2$.)
\end{definition}
\noindent We recall the following standard relation between total variation distance and $\chi^2$-divergence.
\begin{fact} \label{prop:chi-tv}
    For any distributions $\cD_1, \cD_2$ with probability density functions, $\dtv(\cD_1, \cD_2)^2 \leq \frac{1}{4}\dchi(\cD_1 \| \cD_2)$.
\end{fact}

\subsection{Useful Probabilistic Tools and Inequalities}
In this subsection, we recall several basic but useful concentration inequalities and moment bounds which we will rely on.
\paragraph{Gaussian Concentration.} First, we note a well-known proposition regarding univariate Gaussians.
\begin{fact} \label{prop:mgf_squared_gaussian}
    For $X \sim \mathcal{N}(0, 1)$ and $a, b \in \R$ such that $b < \frac{1}{2}$, we have that
\[\E\left[e^{a X + b X^2}\right] = \frac{\exp\left(\frac{a^2}{2-4b}\right)}{\sqrt{1-2b}}.\]
    In the special case $a = 0$, this becomes
$\E\left[e^{b X^2}\right] = \frac{1}{\sqrt{1-2b}}.$
\end{fact}
The following provides a generalization of~\cref{prop:mgf_squared_gaussian} to multivariate Gaussians: we include a proof for completeness.\cmargin{Do we need to?}\todonote{Do you know of a reference?}
\begin{proposition}\label{prop:shifted-Gaussian-integral}
Let $z \in \R^d$ be drawn from the Gaussian $\cN(0, \delta^2 I)$.  Then for parameters $a \geq 0$ and $s \in \R^d$, we have
\[
\E_z\left[ \exp\left( -\frac{a}{2} \normtwo{z}^2 - \langle s , z\rangle \right) \right] =  \frac{1}{\Paren{a\delta^2 + 1}^{d/2}} \exp\left( \frac{\normtwo{s}^2}{2(a + 1/\delta^2)}\right) 
\]
\end{proposition}
\begin{proof}
We have
\[
\E_z\left[ \exp\left( -\frac{a}{2} \normtwo{z}^2 - \langle s , z\rangle \right) \right] = \int \frac{1}{\delta^d (2 \pi)^{d/2}} \exp\left( -\frac{a}{2} \normtwo{z}^2 - \langle s , z\rangle - \frac{\normtwo{z}^2}{2\delta^2}\right) dz \,.
\]
To compute the integral, we can complete the square in the exponential and write it as 
\[
-\frac{1}{2}\normtwo{ \sqrt{a + 1/\delta^2} z - \frac{1}{\sqrt{a + 1/\delta^2}} s } + \frac{\normtwo{s}^2}{2 (a + 1/\delta^2)}
\]
so
\[
\E_z\left[ \exp\left( -\frac{a}{2} \normtwo{z}^2 - \langle s , z\rangle \right) \right] = \exp\left(\frac{\normtwo{s}^2}{2 (a + 1/\delta^2)} \right) \left( \frac{1}{\sqrt{a + 1/\delta^2}}\right)^d \frac{1}{\delta^d} = \frac{1}{\sqrt{a\delta^2 + 1}^d} \exp\left( \frac{\normtwo{s}^2}{2(a + 1/\delta^2)}\right)
\]
as desired.
\end{proof}

\begin{proposition} \label{prop:exponential-expectation-determinant}
    Let $z \sim \cN(0, I_n)$ be an $n$-dimensional Gaussian vector. Then, for any symmetric matrix $M$ with all eigenvalues strictly greater than $-1$,
\[\BE[e^{-\frac{1}{2} \cdot z^\top M z}] = \det(I+M)^{-1/2}.\]
\end{proposition}

\begin{proof}
    First, suppose that $M$ is diagonal, with eigenvalues $\lambda_1, \lambda_2, \dots, \lambda_n$. Then, $-\frac{1}{2} z^\top M z = -\frac{1}{2} \cdot \sum \lambda_i z_i^2$. Since each coordinate of $z_i$ is independent, 
\[\BE[e^{-\frac{1}{2} z^\top M z}] = \BE\left[\prod_{i=1}^n e^{-\frac{1}{2} \lambda_i z_i^2}\right] = \prod_{i=1}^n \BE\left[e^{-\frac{1}{2} \lambda_i z_i^2}\right] = \prod_{i=1}^n \frac{1}{\sqrt{1+\lambda_i}} = \det(I+M)^{-1/2} \, ,\]
    using \Cref{prop:mgf_squared_gaussian} with $a = 0$.
    Finally, by rotational symmetry of Gaussians, the claim holds for all symmetric $M$.
\end{proof}

From \Cref{prop:exponential-expectation-determinant}, we have the following immediate corollary.

\begin{corollary} \label{cor:mgf_bivariate_gaussian}
    Let $X, Y$ be a 2-dimensional multivariate Gaussian with mean $\textbf{0}$ and covariance matrix $\Sigma = \left(\begin{matrix} a & b \\ b & c \end{matrix}\right)$. Then, if $a+c < \frac{1}{2}$, $\E[e^{X^2+Y^2}] = \frac{1}{\sqrt{(1-2a)(1-2c)-4b^2}}$
\end{corollary}

\begin{proof}
    We can write $e^{X^2+Y^2} = e^{x^\top \Sigma x} = e^{-\frac{1}{2} x^\top (-2\Sigma) x}$ for $x \sim \cN(0, I)$. Since $\Sigma$ is PSD and has trace less than $\frac{1}{2},$ its eigenvalues are both less than $\frac{1}{2},$ so $-2\Sigma$ has all eigenvalues strictly more than $-1$. Hence, we can apply \Cref{prop:exponential-expectation-determinant}, noting that $\det(1-2\Sigma) = (1-2a)(1-2c)-4b^2,$ which completes the proof.
\end{proof}

Next, we note some basic facts about the norm and inner products of Gaussians.

\begin{fact}\label{fact:bounded-norm}
Consider $n$ points $X_1, \dots , X_n \in \R^d$ drawn from a Gaussian $\cN(\mu, I)$ where $\norm{\mu} \leq O(1)$.
Then, 
with probability $1 - \delta$, we have for all $i \in [n]$
\[d - 10 \left(\sqrt{\log(n/\delta) d} + \log(n/\delta)\right) \leq \norm{X_i}^2 \leq d + 10 \left(\sqrt{\log(n/\delta) d} + \log(n/\delta)\right) \,. \]
\end{fact}

In light of Fact~\ref{fact:bounded-norm}, it suffices to consider when all of the points, including the contaminations (in any of the models) have $\norm{X_i}^2 = d \pm O(\sqrt{d} \cdot \poly\log(d, n))$,
since we can simply remove all points whose norm is too large or too small: with high probability these points are all contaminated.

\begin{fact}\label{fact:concentration-inner-prod}
    Let $z_1, z_2$ be Gaussians $\cN(0, I)$ in $\R^d$. Then, for any $C \le O(\sqrt{d}),$ $\BP(|\langle z_1, z_2 \rangle| \ge C \sqrt{d}) \le 2 e^{-\Omega(C^2)}$.
\end{fact}

\begin{proof}
    First, with probability at least $e^{-\Omega(d)},$ $\|z_2\|_2 \le 2 \sqrt{d}$, by \Cref{prop:chi-square-concentration}. Then, conditioned on the norm of $z_2$, $\langle z_1, z_2 \rangle$ has distribution $\cN(0, 1) \cdot \|z_2\|_2$, which is at most $C \cdot \|z_2\|_2$ with probability at least $2e^{-C^2/2}$. Hence, $\BP(|\langle z_1, z_2 \rangle| \ge C \sqrt{d}) \le 2e^{-C^2/8} + e^{-\Omega(d)}$.
\end{proof}

We will also make use of the \emph{Hanson-Wright inequality}.

\begin{lemma}[Hanson--Wright] \label{lem:hanson-wright}
    Given an $n \times n$ matrix $A \in \R^{n \times n}$ and an $n$-dimensional Gaussian vector $Z \sim \cN(0, I)$, for any $t \ge 0$,
\[\BP\left(\left| Z^\top A Z - \BE[Z^\top A Z]\right| \ge t\right) \le 2 \exp\mleft(-c \min\left(\frac{t^2}{\|A\|_F^2}, \frac{t}{\|A\|_{op}}\right)\mright),\]
    for some absolute constant $c > 0$. This implies that, with very high probability, $\left| Z^\top A Z - \BE[Z^\top A Z]\right| \le \tilde{O}(\|A\|_F)$.
\end{lemma}

The Hanson-wright inequality with $A = I$, the $n \times n$-dimensional identity vector, immediately implies the following.
\begin{proposition} \label{prop:chi-square-concentration}
    Let $z_1, \dots, z_n$ be \iid Gaussians. Then, for any $t \ge 0$,
\[
\BP\left(\left|\sum_{i=1}^n z_i^2 - n \right| \ge t\right) \le 2 \exp\mleft(-c \cdot \min\left(\frac{t^2}{n}, t\right)\mright).
\]
\end{proposition}

We note one more result about Gaussian samples, which follows from well-known facts about sufficient statistics. The following result says that if we know the mean $\bar{X}$ of some Gaussian samples $X_i$ drawn as $\cN(\mu, I)$, the posterior distribution of the deviations $X_i-\bar{X}$ does not depend on the mean $\mu$.

\begin{proposition}[{\cite[Corollary~18]{narayanan2022privatetesting}}] \label{prop:sufficient_statistic} 
    For any $\mu \in \R^d$, let $X_1, \dots, X_n$ be distributed \iid as $\mathcal{N}(\mu, I)$, and let $\bar{X} = \frac{1}{n}\Paren{X_1+\cdots+X_n}$. Then, for $Z_1, \dots, Z_n \overset{i.i.d.}{\sim} \mathcal{N}(0, I)$, independent of $(X_1, \dots, X_n)$, and $\bar{Z} = \frac{1}{n}\Paren{Z_1+\cdots+Z_n}$, we have that $X_1, \dots, X_n$ has the same distribution as $\bar{X}+Z_1-\bar{Z}, \dots, \bar{X}+Z_n-\bar{Z}$.
\end{proposition}

\paragraph{Hypergeometric distributions.} Next, we will require some bounds on Hypergeometric distributions. First, we define a Hypergeometric distribution.
\begin{definition}
    For $n\in\N$ and $0\leq k_1,k_2 \leq n$, a Hypergeometric distribution $\hgeom(n, k_1, k_2)$ is the distribution of the random variable $Y$ generated as follows. Fix a set $[n]$ of size $n$, and let $S, T$ be independent random subsets of $[n]$ of size $k_1, k_2$ respectively. Then, output $Y = |S \cap T|$.
\end{definition}

It is well-known that $\hgeom(n, k_1, k_2)$ has expectation $\frac{k_1 \cdot k_2}{n}.$
We will also use the following Bernstein-type inequality for Hypergeometric distributions.

\begin{lemma}[{\cite[Corollary~1]{greene2017hypergeometric}}]\label{lem:hypergeometric-bernstein}
    Suppose that $k_1, k_2 \le \frac{1}{2} \cdot n$, and $X \sim \hgeom(n, k_1, k_2)$. Then, for all $\lambda > 0$,
\[\BP\left(\sqrt{k_1} \cdot \Big(\frac{X}{k_1} - \frac{k_2}{n}\Big) > \lambda\right) \le \exp\left(-\frac{\lambda^2/2}{(k_2/n) + \lambda/(3 \sqrt{k_1})}\right).\]
\end{lemma}
The following is a direct corollary of \Cref{lem:hypergeometric-bernstein}.
\begin{corollary} \label{cor:hypergeometric}
    Suppose that $k_1 = k_2 = \eps n$, and $X \sim \hgeom(n, k_1, k_2)$. Then, for all $t > 0$,
\[\BP\left(\frac{X}{n} - \eps^2 > t\right) \le \exp\left(-\min\left(\frac{t^2 \cdot n}{4 \eps^2}, \frac{t \cdot n}{4}\right)\right).\]
\end{corollary}

We also will utilize the following subgaussian concentration bound for Hypergeometric distributions.
\begin{proposition}[\cite{skalahypergeometric}] \label{prop:hypergeometric-subgaussian} If $X \sim \hgeom(n, k_1, k_2)$, then for any $t \ge 0$, $\BP\left[X \ge \BE[X] + t \cdot k_1\right] \le e^{-2 t^2 \cdot k_1},$ and $\BP\left[X \ge \BE[X] - t \cdot k_1\right] \le e^{-2 t^2 \cdot k_1}.$
\end{proposition}

\subsection{Simplification of Alternative Hypothesis} \label{subsec:wlog}

We recall that we wish to distinguish between the null hypothesis where $\mu = 0$ and the alternative hypothesis where $\|\mu\|_2 \ge \alpha$. In this subsection, we briefly explain why it suffices to consider a slightly weaker alternative hypothesis of $\alpha \le \|\mu\|_2 \le 2 \alpha$. This reduction is very similar to one used in \cite[Proposition 23]{narayanan2022privatetesting}. We will only describe the reduction for oblivious robust testing, as we will not (directly) need the reduction in the adaptive case.

\begin{proposition} \label{prop:wlog-bounded}
    Let $0 < \alpha \le O(1)$. Suppose $\cA$ is an algorithm that can distinguish between $n$ $\eps$-obliviously contaminated samples from $\cN(0, I)$ and $n$ $\eps$-obliviously contaminated samples from $\cN(\mu, I)$ where $\|\mu\|_2 \in [\alpha, 2 \alpha]$, with probability at least $0.9$. Then, there exists an algorithm $\cA'$ that can distinguish between $n \cdot \poly\!\log(n, d, \frac{1}{\alpha})$ $\frac{\eps}{\poly\!\log(n, d, \frac{1}{\alpha})}$-obliviously contaminated samples from $\cN(0, I)$ and $n \cdot \poly\!\log(n/d)$ $\frac{\eps}{\poly\!\log(n, d, \frac{1}{\alpha})}$-obliviously contaminated samples from $\cN(\mu, I)$, where $\|\mu\|_2 \ge \alpha$, with probability at least~$0.9.$
\end{proposition}

\begin{proof}
    Suppose our dataset of $n \cdot \poly\log(n, d, \frac{1}{\alpha})$ points is called $X$, which we split into groups $X^{(r, t)}$, where $1 \le r \le R = O(\log nd)$ and $1 \le t \le T = O(\log \frac{n}{\alpha})$, and where $|X^{(r, t)}| = n$. Also, let $X^{(t)} = \bigcup_r X^{(r, t)}$.
    We can consider conditioning on the location and value of each corrupted point, and then consider drawing the uncorrupted points.
    Then, if $X$ is $\frac{\eps}{O(\log^2 (nd/\alpha))}$-obliviously corrupted, each $X^{(r, t)}$ is $\eps$-obliviously corrupted.
    Also, conditioned on the indices and values of the corrupted points, the uncorrupted points in each $X^{(r, t)}$ are independent. (For the rest of the proof, we will think of the corrupted indices/values as fixed.)

    First, we show an amplification result that on each $X^{(t)}$, we can distinguish between $\mu = 0$ and $\|\mu\|_2 \in [\alpha, 2 \alpha]$, with failure probability at most $\frac{1}{nd}$ (instead of failure probability $0.1$).
    For each $X^{(r, t)}$, because the corruptions are oblivious to the data, the probability of $\cA$ outputting the right answer on the group $X^{(r, t)}$ is at least $0.9$, and is independent across groups (after the above conditioning). So by a Chernoff bound, $\cA$ will output the right answer on at least $0.8 \cdot R$ groups with probability at least $1-\frac{1}{nd}$, for any fixed $t$. 
So, the algorithm should simply output the majority across all $r$.

    Next, the same result holds if the alternative hypothesis is $\|\mu\|_2 \in [\alpha \cdot 2^{t-1}, \alpha \cdot 2^t]$ for any $t \ge 1$. To see why, replace each $X_i \in X^{(r, t)}$ with $X'_i := (X_i + \sqrt{2^{2t}-1} \cdot Z_i)/2^t$, where $Z_j \sim \mathcal{N}(0, I)$ is independent for each $X_i$. If $X_i \sim \mathcal{N}(\mu, I)$, then $X'_i \sim \mathcal{N}(\mu/2^t, I)$. Moreover, $\{X_i'\}$ is still $\eps$-obliviously corrupted, because $\{Z_i\}$ is chosen as i.i.d. Gaussians independent of the samples.

    The algorithm $\cA'$ thus works as follows. For each $1 \le t \le O(\log(n/d))$, we test on $X^{(t)}$ whether $\mu = 0$ or $\|\mu\|_2 \in [\alpha \cdot 2^{t-1}, \alpha \cdot 2^t]$, with failure probability at most $\frac{1}{nd}.$ $\cA'$ rejects if any of these tests on $X^{(t)}$ rejects for some $1 \le t \le T$.
Under the null hypothesis, with at least $0.99$ probability, no test will reject. However, if $\|\mu\|_2 \in [\alpha, 10 \sqrt{d}]$, then there exists some $0 \le t \le O(\log(d/\alpha))$ such that $\|\mu\|_2 \in [\alpha \cdot 2^t, \alpha \cdot 2^{t+1}]$, so the test on $X^{(t)}$ will reject. Finally, we use an additional $O(1)$ randomly chosen points to robustly test whether $\|\mu\|_2 \ge 10 \sqrt{d}$, with $0.99$ success probability.
\end{proof}

\subsection{Notation} \label{subsec:notation}
We record here several notational conventions.
\begin{itemize}
    \item In what follows, $\dst>0$ is the distance parameter, $\cor\in[0,1]$ is the corruption rate, $\dims$ denotes the dimension, and $\ns$ is the number of samples. In the remainder of the paper, we assume $\alpha \le O(1)$. In addition, one must have $\cor \leq \dst$. 
    \item We use $\tilde{O}$, $\tilde{\Omega}$, $\tilde{\Theta}$ to hide polylogarithmic factors in the argument. 
    \item Given a distribution $\cD$, we use $p_\cD(\cdot)$ to represent the corresponding PDF (whenever it is well-defined).
    \item Throughout this paper, for a set of vectors $S$, we will use the shorthand $\Sum(S)$ to denote the sum of the vectors in $S$, \ie $\Sum(S)\eqdef \sum_{x\in S} x$.
\end{itemize}

Throughout the remainder of the paper, we will assume $\alpha \geq 0.99^d$.  We will also assume the desired failure probability $\delta \geq 0.99^d$.  Thus, we can also assume that the number of samples $n \leq 100d\log(1/\delta)/\alpha^2 \leq (1.1)^d$ since that would suffice to learn the distribution to accuracy $0.1\alpha$.\footnote{Recall that we have a non-robust testing lower bound of \smash{$\bigOmega{\sqrt{d}/\alpha^2}$} and an efficient robust learning upper bound of \smash{$\bigO{d/\alpha^2}$}. In the regime $\alpha \leq  0.99^d$, however, $d/\alpha^2 = O(\sqrt{d} \log(1/\alpha) /\alpha^2) = \tilde{O}(d/\alpha^2)$, so the trivial upper and lower bounds match up to logarithmic factors. For the failure probability $\delta$, note that we can amplify any constant success probability in both the oblivious and Huber contamination models by running multiple trials, at the cost of a multiplicative $O(\log(1/\delta))$ factor.}

\section{Reducing to ``Friendly" Oblivious Contaminations} \label{sec:strong-sample-splitting}

The first key step in our oblivious upper bound is arguing that we can reduce to when the contaminated points are reasonably behaved.  Formally, we want to argue that it suffices to consider a friendly oblivious contamination defined as follows.

\begin{restatable}{definition}{friendly}[(Friendly) Oblivious Contamination Model] \label{defn:friendly-intro}
  We say $X_1,\ldots,X_n$ are obliviously $\eps$-contaminated samples from a distribution $\cD$ if they are drawn as follows: first $Y_1,\ldots,Y_{\e n}$ are chosen adversarially, then $Y_{\e n + 1},\ldots, Y_{n} \sim \cD$ i.i.d., and finally $Y_1,\ldots,Y_n$ are randomly permuted to produce $X_1,\ldots,X_n$.

  In the \emph{friendly} oblivious contamination model, we additionally have the following assumption about the data:
\end{restatable}
\begin{assumption} 
\label{ass:friendly}
A dataset $X_1,\ldots,X_n \in \R^d$ is $\kappa$-\emph{friendly} if the  following all hold:
\begin{enumerate}
    \item For any disjoint subsets $S, T \subset [n]$ of sizes $k_1, k_2 \le \eps \cdot n$,
\[\left|\left\langle \sum_{i \in S} X_i, \sum_{i\in T} X_i \right\rangle\right| \le \kappa \cdot (\sqrt{k_1 k_2} \cdot \max(\sqrt{\eps n d}, \eps n)).\] 
    \item For every distinct $i \neq j \in [n]$, $|\langle X_i, X_j \rangle| \le \kappa \cdot \sqrt{d}$. 
    \item For every $i \in [n]$, $\|X_i\|_2^2 = d \pm \kappa \sqrt{d}$. 
\end{enumerate}
\end{assumption}
In this definition, one should think of $\kappa = \poly(\log(n), \log(d))$.

Note that we need to make the reduction to friendly oblivious contamination while preserving the ``obliviousness" of the contaminated points.  Getting the first condition is the main difficulty (the latter two are relatively straight-forward in light of Fact~\ref{fact:bounded-norm} and Claim~\ref{claim:inner-product} below)  as natural algorithms for filtering/removing points don't preserve this ``obliviousness" and thus cannot be used.  Nevertheless in this section, we show how to filter an arbitrary oblivious contamination on a dataset to a friendly oblivious contamination while preserving obliviousness. We will prove the following theorem.

\begin{theorem}[Dealing with $\widetilde{O}(1)$-Friendly Contamination Suffices]\label{thm:reduce-to-friendly}
Assume there exists an algorithm for robust mean testing in $\R^d$ under $\kappa$-friendly oblivious $\eps$-contamination that uses $n = f(d, \alpha , \eps)$ samples and succeeds with probability $p > 2/3$ where $\kappa = (10 \log (nd))^{2000}$.  Assume $n \leq (1.1)^d$.  Then there exists an algorithm for robust mean testing in $\R^d$ under (arbitrary) oblivious $\eps/2$-contaimination that succeeds with probability $p - 0.01$ and uses $n \poly(\log (nd))$ samples.
\end{theorem}

\subsection{Structure of Obliviously Contaminated Samples}

We begin by proving a few basic structural properties that hold with high probability for an obliviously contaminated dataset.  First, we show that the inner product between any two points that are not both contaminated must be small.

\begin{claim}\label{claim:inner-product}
Consider a set $S = \{X_1, \dots , X_n \}$ of $n$ points in $\R^d$ that are drawn from $N(\mu, I)$ and then $\eps$-contaminated in the oblivious contamination model.  Let $R \subset S$\todonotedone{Shyam: oops I've always been using $B$ as the bad points and $R$ as the sum ... Hopefully its ok if I made it clear my notation in the sections. \\ OK so i changed it a bit to be in mathbf font, and gave careful definitions, hopefully that way nobody will confuse it with anything else at least.} be the subset of contaminated points.  Also assume $\norm{\mu} \leq 1$ and $n \leq (1.1)^d$.  Then for any $0.99^d < \delta < 1$, with probability $1 - \delta$, we have that for all $X_i \in S\backslash R, X_j \in S$ with $i \neq j$,
\[
\frac{|\langle X_i, X_j \rangle |}{\norm{X_i} \norm{X_j}} \leq \frac{10\log(n/\delta)}{\sqrt{d}}
\]
\end{claim}
\begin{proof}
Since the contamination is oblivious, we can imagine fixing the index $j$ first and then drawing $X_i$.  We can write $X_i = \mu + v$ where $v \sim N(0, I)$.  We have with probability $1 - \delta/(2n^2)$
\[
|\langle X_i , X_j/\norm{X_j} \rangle| = |\langle \mu , X_j/\norm{X_j} \rangle  + \langle v  , X_j/\norm{X_j} \rangle  | \leq | 1 + 5\sqrt{\log(n/\delta)}|
\]
where in the last step we simply noted that $\langle v  , X_j/\norm{X_j} \rangle $ is distributed as a standard Gaussian and the desired inequality follows from standard tail bounds.  Also by Fact~\ref{fact:bounded-norm}, $\norm{X_i} \geq \sqrt{d/2}$ with probability at least $1 - \delta/(2n)$ and combining this with the above gives
\[
\frac{|\langle X_i, X_j \rangle |}{\norm{X_i} \norm{X_j}} \leq \frac{10\log(n/\delta)}{\sqrt{d}} \,.
\]
Union bounding the failure probability over all $i,j$ we are done.
\end{proof}

We also have the following bound on the number of uncontaminated points with large projection onto any direction determined by a small subset of datapoints.

\begin{claim}\label{claim:projection-counts}
Consider a set $S = \{X_1, \dots , X_n \}$ of $n$ points in $\R^d$ that are drawn from $N(\mu, I)$ and then $\eps$-contaminated in the oblivious contamination model.  Let $R \subset S$ denote the contaminated points.  Also assume $\norm{\mu} \leq 1$ and $n \leq (1.1)^d$.  Then for any $0.99^d < \delta < 1$, with probability $1 - \delta$, we have the following property:  for any subset $T \subset S$, 
\[
\left\lvert \left\{  X_i \in S \backslash R \; \bigg| \; |\langle X_i , \Sum(T)/\norm{\Sum(T)}\rangle | \geq 10\sqrt{\log(n/\delta)} \right\} \right\rvert \leq 2|T| \,.
\]
\end{claim}
\begin{proof}
We consider a fixed set $T$ and then union bound over all possible choices of $T$.  For a fixed set $T$, we can imagine fixing the points $X_i \in T$ first and then drawing the remaining points $X_i \in S\backslash (R \cup T) \sim N(\mu, I)$ afterwards.  It suffices to upper bound the probability that more than $|T|$ of these points satisfy 
\[
|\langle X_i , \Sum(T)/\norm{\Sum(T)}\rangle | \geq 10\sqrt{\log(n/\delta)} \,.
\]
This probability can be upper bounded by
\[
(\delta/n)^{10 |T|} \cdot n^{|T|} \leq (\delta/n)^{9|T|} 
\]
and then union bounding over all possible choices of $T$ gives the desired statement. 
\end{proof}

\subsection{Oblivious Filtering via Sample Splitting}

Recall that our approach to prove Theorem~\ref{thm:reduce-to-friendly} will be to ``obliviously" filter the dataset, removing some of the contaminated points, so that the remaining data is friendly.  In light of Fact~\ref{fact:bounded-norm} and Claim~\ref{claim:inner-product}, it is not difficult to enforce the latter two conditions for friendliness since we can simply remove points whose norm is too large or too small and also remove pairs of points whose inner product is too large.  The main difficulty lies in enforcing the first condition and this is our focus for the remainder of this section.

It will be convenient to make the following definition.

\begin{definition}
Let $S$ be a set of vectors in $\R^d$.  For parameters $\lambda, m, k$, we say that $S$ is $(\lambda, m,k)$-balanced if for all pairs of disjoint subsets $S_1 , S_2 \subset S$ with $|S_1| ,|S_2|  \leq m$ and $|S_1| |S_2| \leq k$, we have
\[
\left \lvert \langle \Sum(S_1), \Sum(S_2) \rangle \right\rvert \leq  \sqrt{\lambda|S_1||S_2|d}
\]
\end{definition}

Roughly, it will suffice to ensure that our dataset is balanced for $\lambda \sim \kappa^2 \eps n $,  $m \sim \eps n $ and $k \sim m^2$ \footnote{For most of this section, we will work in the regime $\eps n \poly(\log (nd)) < d$.  We will show a reduction when we finally prove Theorem~\ref{thm:reduce-to-friendly} that allows us to reduce to this case.}.  We will do this by iteratively doubling $k$ i.e. going from $(\lambda , m ,k/2)$-balanced to $(\lambda ,m , k)$-balanced.  At a high-level the way we do this while maintaining obliviousness of the contaminations is as follows. We randomly split the dataset into two parts $A,B$ and only look at $A$ to construct some filter that ``cleans" $A$ i.e. makes it $(\lambda ,m , k)$-balanced.  We then argue that with high probability, this filter must clean $B$ and we simply apply it to $B$ and iterate on the remaining data (throwing away $A$).  Crucially, this sample splitting preserves the obliviousness of the contaminations because the filters are constructed independently of the uncontaminated data since we can view the uncontaminated points in $B$ as being drawn after running our algorithm on $A$.  See Algorithm~\ref{alg:single-iteration} and Algorithm~\ref{alg:full-sample-splitting}  for more specific details.

We first need to prove a few basic properties.  If a set of vectors $S$ is $(\lambda ,m , k/2)$-balanced and not $(\lambda , m , k)$-balanced, then there must be some disjoint subsets $S_1, S_2 \subset S$ with $k/2 \leq |S_1||S_2| \leq k$ that witness this i.e.
\[
|\langle \Sum(S_1), \Sum(S_2) \rangle| \geq \sqrt{\lambda  |S_1||S_2|d} \,.
\]
The above statement says that on average, vectors in $S_2$ have large inner product with $\Sum(S_1)$.  In the next claim, we prove that this is actually the case for a large subset of $S_2$.

\begin{claim}\label{claim:inner-product-uniformity}
Let $S_1, S_2$ be two disjoint sets of vectors in $\R^d$. Let $k,m$ be some parameters such that $|S_1|, |S_2| \leq m$.  Assume that $S_1 \cup S_2$ is $(\lambda, m, k/2)$-balanced.  Then if $|S_1| \cdot |S_2| \leq k$ and 
\[
\langle \Sum(S_1) , \Sum(S_2) \rangle \geq  \sqrt{\theta |S_1| |S_2| d}
\]
for some parameter $\theta \leq \lambda$, then there is a subset $S_2' \subset S_2$ with $|S_2'| \geq \frac{\theta |S_2|}{8\lambda}$ such that for all $v \in S_2'$, 
\[
 \langle v, \Sum(S_1) \rangle   \geq \frac{1}{4} \sqrt{\frac{\theta |S_1|d}{|S_2|}} 
\]
\end{claim}
\begin{proof}
Let $T$ be the set of all vectors $v \in S_2$ such that 
\[
\langle v , \Sum(S_1) \rangle \geq 2 \sqrt{\frac{\lambda^2|S_1|}{\theta |S_2|}d} \,.
\]
If $T$ has size larger than $\theta |S_2|/(4 \lambda)$, then by taking $T'$ to be a random subset of $T$ of size $\theta |S_2|/(4 \lambda)$, we would get 
\[
\BE\left[\langle \Sum(T') , \Sum(S_1) \rangle\right] \geq \frac{\theta |S_2|}{4\lambda} \cdot 2 \sqrt{\frac{\lambda^2|S_1|}{\theta |S_2|}d} = \frac{1}{2} \sqrt{\theta |S_1||S_2|d} \geq \sqrt{\lambda |T'||S_1|d} 
\]
Hence, the above deterministically happens for some $T' \subset T$ of size $\theta |S_2|/(4 \lambda)$,
which contradicts the assumption that $S_1 \cup S_2$ is $(\lambda, m , k/2)$-balanced.  Thus, we must actually have $|T| \le \theta |S_2|/(4 \lambda)$, and
\[
\langle \Sum(T), \Sum(S_1) \rangle \leq \sqrt{\lambda |T| |S_1|d} \leq \frac{1}{2} \sqrt{\theta |S_1||S_2|d} \,.
\]
In particular, this means that
\[
\langle \Sum(S_2 \backslash T),  \Sum(S_1) \rangle \geq \frac{1}{2} \sqrt{\theta |S_1||S_2|d} \,.
\]
Next, let $R$ be the set of all vectors $v \in S_2$ such that
\[
\langle v, \Sum(S_1) \rangle \leq \frac{1}{4} \sqrt{\frac{\theta |S_1|d}{|S_2|}} \,.
\]
We have that 
\[
\langle \Sum(S_2 \backslash (T \cup R)), \Sum(S_1) \rangle \geq \frac{1}{4} \sqrt{\theta |S_1||S_2|d} \,.
\]
Thus, by the construction of $R,T$, we conclude that the number of vectors $v \in S_2$ such that 
\[
\langle v, \Sum(S_1) \rangle \geq \frac{1}{4} \sqrt{\frac{\theta |S_1|d}{|S_2|}}
\]
is at least 
\[
\frac{\langle \Sum(S_2 \backslash (T \cup R)), \Sum(S_1) \rangle }{2 \sqrt{\frac{\lambda^2|S_1|}{\theta |S_2|}d}} \geq \frac{\theta |S_2|}{8\lambda} \,. \qedhere
\]
\end{proof}

With the above equipped, we can now show that if $S_1, S_2 \subset S$ are two sets of samples that violate $(\lambda ,m , k )$-balancedness, then if we split $S$ into two parts $A,B$, with all but exponentially small (in $\min(|S_1|,|S_2|)$) probability, both parts $A,B$ will witness a violation for slightly smaller values of $\lambda,m,k$.  

\begin{lemma}\label{lem:violating-sets-split}
Let $S_1, S_2$ be two disjoint sets of vectors in $\R^d$. Let $k,m$ be some parameters such that $|S_1|, |S_2| \leq m$.  Assume that $S_1 \cup S_2$ is $(\lambda, m, k/2)$-balanced.  Also assume that $|S_1| \cdot |S_2| \leq k$ and 
\[
\langle \Sum(S_1) , \Sum(S_2) \rangle  \geq  \sqrt{\lambda |S_1| |S_2| d} \,.
\]
Consider splitting $S_1,S_2$ each into two sets $S_{1,A}, S_{1,B}$ and $S_{2,A}, S_{2,B}$ respectively where each element is assigned to the first part independently with probability $p$.  Then with probability $1 - 2^{- \frac{\min(|S_1|, |S_2|)p^3(1- p)^3}{10^{10}}}$,  there are subsets $S_{1,A}', S_{2,A}', S_{1,B}', S_{2,B}'$ with $S_{1,A}' \subset S_{1,A}$, $ S_{2,A}' \subset S_{2,A}$, $S_{1,B}' \subset S_{1,B}$, $S_{2,B}'  \subset S_{2,B}$ such that 
\begin{equation*}
\begin{split}
&|S_{1,A}'|   = \frac{p^2|S_1|}{10^6} \\ &|S_{1,B}'|= \frac{(1 - p)^2|S_1|}{10^6}\\
&  |S_{2,A}'| = \frac{p^2|S_2|}{10^6} \\ &|S_{2,B}'| = \frac{(1 - p)^2|S_2|}{10^6} \\
&\langle \Sum(S_{1,A}') , \Sum(S_{2,A}') \rangle \geq \frac{p^4\sqrt{\lambda |S_1||S_2| d}}{10^{13}}
\\ 
& \langle \Sum(S_{1,B}') , \Sum(S_{2,B}') \rangle \geq \frac{(1 - p)^4\sqrt{\lambda |S_1||S_2| d}}{10^{13}}
\end{split}
\end{equation*}
\end{lemma}

The proof of Lemma~\ref{lem:violating-sets-split} relies on the claim below, which characterises what happens when we split one of the sets, say $S_2$ into two parts.

\begin{claim}\label{claim:split-one-side}
Let $S_1, S_2$ be two disjoint sets of vectors in $\R^d$. Let $k,m$ be some parameters such that $|S_1|, |S_2| \leq m$.  Assume that $S_1 \cup S_2$ is $(\lambda, m, k/2)$-balanced.  Also assume that $|S_1| \cdot |S_2| \leq k$ and 
\[
 \langle \Sum(S_1) , \Sum(S_2) \rangle  \geq  \sqrt{\theta |S_1| |S_2| d}
\]
for some parameter $\theta \leq \lambda$.  Now consider splitting $S_2$ into two pieces $A,B$ where each element is independently assigned to $A$ with probability $p$ (and assigned to $B$ otherwise).  Then with probability $ 1 - 2^{-\frac{p^2(1-  p)^2 \theta|S_2|}{10^2\lambda}}$, there exist subsets $A' \subset A$ and $B' \subset B$ such that 
\[
\begin{split}
|A'| &=  \frac{p \theta |S_2|}{20\lambda}  \\
|B'| &= \frac{(1 - p)\theta |S_2|}{20\lambda}  \\
\langle \Sum(A'), \Sum(S_1)\rangle &\geq \frac{p \theta}{10^2\lambda} \sqrt{\theta |S_1||S_2|d} \\ 
\langle \Sum(B'), \Sum(S_1) \rangle &\geq \frac{(1 - p)\theta}{10^2 \lambda} \sqrt{\theta |S_1||S_2|d}
\end{split}
\]
\end{claim}
\begin{proof}
First, construct the set $S_2'$ according to Claim~\ref{claim:inner-product-uniformity}.  By Hoeffding's inequality, with probability $1 - 2^{-\frac{p^2(1-  p)^2 \theta|S_2|}{10^2\lambda}}$ , we have $|S_2' \cap A| \geq p\theta |S_2|/(20 \lambda)$ and $|S_2' \cap B| \geq (1 - p)\theta |S_2|/(20\lambda)$.  Let $A',B'$ be arbitrary subsets of $S_2' \cap A, S_2' \cap B$ with sizes  $p\theta |S_2|/(20 \lambda)$ and $(1 - p)\theta |S_2|/(20\lambda)$ respectively.  Then by the properties of $S_2'$ guaranteed by Claim~\ref{claim:inner-product-uniformity} we have
\[
\langle \Sum(A'), \Sum(S_1)\rangle \geq |A'| \cdot \frac{1}{4} \sqrt{\frac{\theta |S_1|d}{|S_2|}}  \geq  \frac{p \theta}{10^2\lambda} \sqrt{\theta |S_1||S_2|d}
\]
and similar for $\langle \Sum(B'), \Sum(S_1) \rangle$, completing the proof.
\end{proof}

We can now prove Lemma~\ref{lem:violating-sets-split} by applying Claim~\ref{claim:split-one-side} twice.

\begin{proof}[Proof of Lemma~\ref{lem:violating-sets-split}]
First consider when $S_2$ is split into $S_{2,A}$ and $S_{2,B}$ and apply Claim~\ref{claim:split-one-side} with $\theta = \lambda$.  This gives us sets $S_{2,A}^{(1)}$ and $S_{2,B}^{(1)}$ with 
\[
\begin{split}
|S_{2,A}^{(1)}| &= \frac{p|S_2|}{20} \\
|S_{2,B}^{(1)}| &= \frac{(1 - p)|S_2|}{20} \\
 \langle \Sum(S_1), \Sum(S_{2,A}^{(1)})\rangle  &\geq \frac{p}{10^2}\sqrt{\lambda |S_1||S_2| d} \\
\langle \Sum(S_1), \Sum(S_{2,B}^{(1)}) \rangle  &\geq \frac{(1 - p)}{10^2}\sqrt{\lambda |S_1||S_2| d} \,.
\end{split}
\]
Now we can apply Claim~\ref{claim:split-one-side} again when splitting $S_1$ with $\theta = p\lambda/10^4$ to get $S_{1,A}'$ with
\[
\begin{split}
|S_{1,A}'| &= \frac{p^2 |S_1|}{10^6} \\
\langle \Sum(S_{1,A}'), \Sum(S_{2,A}^{(1)}) \rangle &\geq \frac{p^3}{10^8} \sqrt{\lambda|S_1||S_2|d} \,.
\end{split}
\]
Now we can take $S_{2,A}'$ to be random subset of $S_{2,A}^{(1)}$ of size $p^2|S_2|/10^6$ and we have in expectation that
\[
\langle \Sum(S_{1,A}'), \Sum(S_{2,A}') \rangle \geq \frac{p^4}{10^{13}} \sqrt{\lambda |S_1||S_2|d} 
\]
so in particular it holds for some choice of $S_{2,A}'$.  We can construct $S_{2,B}'$ similarly.  The overall failure probability over all applications of Claim~\ref{claim:split-one-side} is at most $2^{-\frac{p^3(1 - p)^3 \min(|S_1|,|S_2|)}{10^{10}}}$ and this completes the proof.
\end{proof}

In light of Lemma~\ref{lem:violating-sets-split}, we know that when we split the set of samples $S$ into two parts $A,B$, any pair of subsets $S_1, S_2$ that violates $(\lambda ,m ,k)$-balancedness creates a violation in both parts with (approximately) $\exp(-\min(|S_1|,|S_2|))$ failure probability.  Now, we roughly proceed as follows. If the set of all possible filters considered by our algorithm has size less than $\exp(\min(|S_1|,|S_2|))$, then we can union bound and conclude that actually any filter that cleans $A$ to be $(\lambda',m',k')$-balanced (for some slightly smaller $\lambda' , m', k'$) must actually clean $S$ to be $(\lambda ,m , k)$-balanced.  Then it suffices to argue that there exists a filter in this set that actually cleans $A$.  The full argument will be slightly more involved as we have to deal with different possibilities for $\min(|S_1|,|S_2|)$ separately.

We first need a few more basic observations.

\begin{definition}
Let $S$ be a set of vectors in $\R^d$.  We say that $S$ is $\rho$-bounded if for all $v \in S$, $d - \sqrt{ \rho d} \leq \norm{v}^2 \leq d + \sqrt{\rho d}$ and for all distinct $u,v \in S$, $-\sqrt{\rho d} \leq \langle u, v \rangle \leq  \sqrt{ \rho d}$.
\end{definition}

\begin{claim}\label{claim:simple-norm-bound}
Let $S \in \R^d$ be a set of vectors that is $(\lambda, m,k)$-balanced and $\lambda$-bounded.  Then for all subsets $T \subset S$ with $|T| \leq \min(m, \sqrt{k})$, $\norm{\Sum(T)}^2    \leq |T|d + 2|T| \sqrt{\lambda d}$.    
\end{claim}
\begin{proof}
We can write
\[
\norm{\Sum(T)}^2 \leq |T|(d +  \sqrt{\lambda d}) +  \sum_{u,v \in T, u \neq v} \langle u ,v \rangle \,.
\]
Now consider a random partition of $T$ into two sets $T_1, T_2$ where each element is assigned uniformly at random.  Then
\[
 \sum_{u,v \in T, u \neq v} \langle u ,v \rangle  = 2\E[ \langle \Sum(T_1), \Sum(T_2) \rangle] \leq |T| \sqrt{\lambda  d} 
\]
where we used the assumption of $(\lambda, m , k)$-balancedness.   Thus,
\[
\norm{\Sum(T)}^2  \leq |T|d + 2|T| \sqrt{\lambda d}
\]
and we are done.
\end{proof}

\begin{claim}\label{claim:small-sets-bound}
Let $S \subset \R^d$ be a set of vectors that is $\lambda/m$-bounded.  Then it is $(\lambda , m , m)$-balanced.
\end{claim}
\begin{proof}
 Consider disjoint subsets $S_1, S_2 \subset S$.  Then by the assumption of $\lambda/m$-bounded,
 \[
| \langle \Sum(S_1), \Sum(S_2) \rangle | \leq |S_1||S_2| \sqrt{\lambda d/m} \,.
\]   
If $|S_1| |S_2| \leq m$ then the above is at most $\sqrt{\lambda |S_1||S_2| d}$, completing the proof.
\end{proof}

Recall that one key point in the earlier sketch is that our algorithm can only enumerate over a (reasonably) small set of filters.  Here we first show that if $S_1,S_2$ violate balancedness, then there exists a direction determined by a small subset $S_1'$ with $|S_1'| \sim |S_1||S_2|/m$ such that filtering along this direction removes a large ($\sim |S_2|$) number of points.  We can then aggregate multiple filtering directions for different choices of $S_1'$ to construct our full filter.  Note that bounding the sizes of the individual sets $S_1'$ is the key for bounding the total number of possible filters being considered.

\begin{lemma}\label{lem:filter-compression}
Let $k,m, \theta, \lambda, C$ be some parameters.  Let $S_1, S_2$ be two disjoint sets of vectors in $\R^d$ with $|S_1|, |S_2| \leq m$, $10 C m \leq |S_1| \cdot |S_2| \leq k$, and
\[
\langle \Sum(S_1) , \Sum(S_2) \rangle \geq \sqrt{\theta |S_1| |S_2| d}\,.
\]
Also, assume that $S_1 \cup S_2$ is $(\lambda, m , k)$-balanced and $\theta/(10^5 C m)$-bounded, where $\theta \le \lambda \le d$ and $C \geq 1$.  Then, there exists a subset $S_1' \subset S_1$ with $|S_1'| \leq |S_1||S_2|/(Cm)$ such that there are at least $\frac{\theta |S_2|}{80 \lambda}$ vectors $v \in S_2$ such that 
\[
\bigg\langle v , \frac{\Sum(S_1')}{\norm{\Sum(S_1')}}\bigg\rangle \geq \frac{\sqrt{\theta}}{16\sqrt{Cm}}
\]
\end{lemma}
\begin{proof}
First we apply Claim~\ref{claim:inner-product-uniformity} to $S_2$ to get a subset $S_2'$ with the stated properties.  Now consider any vector $v \in S_2'$.  Consider drawing a random subset $S_1' \subset S_1$ of size $|S_1'| =  |S_1||S_2|/(Cm)$ (note this is well defined because $|S_1||S_2|  \geq 10Cm$ and $|S_1'| \leq |S_1|$).  First, we compute
\[
\E_{S_1'}[ \langle \Sum(S_1') , v \rangle  ] = \frac{|S_1'|}{|S_1|} \langle v, \Sum(S_1) \rangle \geq \frac{|S_1'|}{4} \sqrt{\frac{\theta d}{|S_1||S_2|}} 
\]
Next, we can compute the second moment
\[
\begin{split}
\E_{S_1'}[ \langle \Sum(S_1') , v \rangle^2  ] &= \sum_{u \in S_1} \frac{|S_1'|}{|S_1|} \langle u, v \rangle^2 + \sum_{u,u' \in S_1, u \neq u'} \frac{|S_1'|(|S_1'| - 1)}{|S_1| (|S_1| - 1)} \langle u, v \rangle \langle u' , v \rangle \\ &\leq \sum_{u \in S_1} \frac{|S_1'|}{|S_1|} \langle u, v \rangle^2 + \frac{|S_1'|(|S_1'| - 1)}{|S_1| (|S_1| - 1)} \sum_{u,u' \in S_1}  \langle u, v \rangle \langle u' , v \rangle \\ &\leq |S_1'| \frac{\theta d}{10^5 C m} + \frac{|S_1'|(|S_1'| - 1)}{|S_1| (|S_1| - 1)}\left(  \langle v, \Sum(S_1) \rangle \right)^2 \\& \leq |S_1'| \frac{\theta d}{10^5 C m} + (\E_{S_1'}[ \langle \Sum(S_1') , v \rangle ])^2 .
\end{split}
\]
Thus, the variance is at most $|S_1'|\theta d/(10^5 C m)$.  Now since $|S_1'| =  |S_1| |S_2|/m$ and $C \ge 1$, we have
\[
\E_{S_1'}[ \langle \Sum(S_1') , v \rangle  ] \geq 5 \sqrt{\textsf{Var}_{S_1'}( \langle \Sum(S_1') , v \rangle)}
\]
and thus with probability at least $0.1$, 
\[
\langle v, \Sum(S_1') \rangle \geq 0.5  \E_{S_1'}[ \langle \Sum(S_1') , v \rangle  ] \geq \frac{|S_1'|}{8} \sqrt{\frac{\theta d}{|S_1||S_2|}}  \,.
\]
Next note that by the constraints on the parameters, $|S_1'| \leq \min(m, \sqrt{k})  $ and thus by Claim~\ref{claim:simple-norm-bound}, 
\[
\norm{\Sum(S_1')} \leq  \sqrt{|S_1'|d + 2|S_1'| \sqrt{\lambda  d}} \leq 2\sqrt{|S_1'|d}
\]
which implies that with $0.1$ probability over the randomness of the choice of $S_1'$
\[
\bigg\langle v, \frac{\Sum(S_1')}{\norm{\Sum(S_1')}} \bigg\rangle \geq \frac{1}{16}\sqrt{\frac{S_1' \theta }{|S_1||S_2|}} \geq \frac{\sqrt{\theta}}{16\sqrt{Cm}} \,.
\]
This holds for all $v \in S_2'$ where $S_2'$ was constructed at the beginning of this proof according to Claim~\ref{claim:inner-product-uniformity}.  By linearity of expectation, this means that there is some choice of $S_1'$ such that there are at least 
\[
0.1|S_2'| \geq \frac{\theta |S_2|}{80 \lambda}
\]
vectors $v \in S_2$ such that 
\[
\bigg\langle v, \frac{\Sum(S_1')}{\norm{\Sum(S_1')}} \bigg\rangle  \geq \frac{\sqrt{\theta}}{16\sqrt{Cm}}
\]
as desired.
\end{proof}

We now describe a single iteration of our algorithm (see Algorithm~\ref{alg:single-iteration}) where we take as input a parameter $s$ and the goal is to eliminate all pairs of subsets $S_1,S_2$ with $s/2 \leq \min(|S_1|, |S_2|) \leq s$ that violate $(\lambda ,m , k)$-balancedness.  Repeating this algorithm over logarithmically many scales for $s$ and then logarithmically many scales for $k$ will give our full algorithm (see Algorithm~\ref{alg:full-sample-splitting}).

\begin{definition}\label{def:filter}
Given a collection of (finite) sets of vectors in $\R^d$, say $F_1, \dots , F_\ell$, and a parameter $\gamma \geq 0$, we define $\Filter_{\gamma}(F_1, \dots , F_\ell) \subset \R^d$ to consist of all vectors $v \in \R^d$ such that 
\[
\max_{i \in [l]} | \langle v, \Sum(F_i)/\norm{\Sum(F_i)} \rangle | \geq \gamma  \,.
\]
When we apply $\Filter_{\gamma}(F_1, \dots , F_\ell)$ to a set $S \subset \R^d$, we remove from $S$ all points that are in $\Filter_{\gamma}(F_1, \dots , F_\ell)$.
\end{definition}

\begin{algorithm}
\begin{algorithmic}
\Require Finite set of samples $S \subset \R^d$
\Require Parameters $\lambda, m , k, s, p$
\State Partition $S$ into two sets $A,B$ where each element is independently assigned to $A$ with probability $p$
\State Set $\tau = \frac{sp^6 }{10^{11} \log |S|}$
\State Set $\gamma = \sqrt{\frac{\lambda p^{50}}{10^{100}m}}$
\State Set $F_1 ,F_2, \dots , F_k = \emptyset$
\For{All collections of subsets $T_1, \dots T_\ell \subset A$ with $|T_1| + |T_2| + \dots + |T_\ell| \leq \tau$}
\State Set $\textsf{check} = \textbf{True}$
\For{ All disjoint pairs $S_1, S_2 \subset A \backslash \Filter_{\gamma}(T_1, \dots , T_\ell)$ with $|S_1| = p^2 s/(2 \cdot 10^6) , |S_2| = p^2 k /(2 \cdot 10^6 s) $}
\If{$| \langle \Sum(S_1), \Sum(S_2) \rangle | \geq \frac{p^4\sqrt{\lambda |S_1||S_2| d}}{10^{14}}$}
\State Set $\textsf{check} = \textbf{False}$
\EndIf
\EndFor
\If{$\textsf{check}$} 
\State Set $F_1 = T_1 , \dots , F_\ell = T_\ell$
\State \textbf{Break}
\EndIf
\EndFor
\State Set $B' = B \backslash \Filter_{\gamma}(F_1, \dots , F_\ell)$
\State\Return $B'$
\end{algorithmic}
\caption{Single Filtering Iteration}\label{alg:single-iteration}
\end{algorithm}
\begin{lemma}[Analysis of Algorithm~\ref{alg:single-iteration}]\label{lem:filter-one-iteration}
Assume that the set $S$ is $(\lambda , m , k/2)$-balanced and $\lambda p^{50}/(10^{100} m) $-bounded.
\todonotedone{Shyam: I believe $m$ needs to be even smaller. Specifically, to apply Lemma 4.12, we need $10 C m \le |S_{1, A}'| \cdot |S_{1, B}'|$.}
Assume the parameters satisfy $\lambda \leq d , m \leq k p^{20}/10^{50}, p \leq 1/2$.  Also, assume that there is a subset $R \subset S$ with $|R| \leq (p^{50}m)/(10^{100} \log |S|)$ such that for any subset $T \subset S$, we have  
\[
\left\lvert \left\{ v \in S \backslash R \; \bigg| \; |\langle v , \Sum(T)/\norm{\Sum(T)}\rangle | \geq \sqrt{\frac{\lambda p^{50}}{10^{100}m}}  \right\} \right\rvert \leq 2|T| \,.
\]

Then with probability $1 - 2^{-sp^6/10^{11}}$, the set $B'$ output by Algorithm~\ref{alg:single-iteration} has the property that for any disjoint sets $S_1, S_2 \subset B'$ with $s/2 \leq |S_1| \leq s, |S_1| \leq |S_2| \leq m$ and $k/2 \leq |S_1| |S_2| \leq k$,
\[
|\langle \Sum(S_1), \Sum(S_2) \rangle | \leq \sqrt{\lambda |S_1||S_2|d} \,.
\]
\end{lemma}
\begin{proof}
Throughout this proof we set $\gamma = \sqrt{\frac{\lambda p^{50}}{10^{100}m}}$ just as in Algorithm~\ref{alg:single-iteration}.  We now introduce some terminology.  
We say that $S\backslash \Filter_{\gamma}(T_1, \dots , T_\ell)$ is unclean if there exist disjoint $S_1, S_2 \subset S\backslash \Filter_{\gamma}(T_1, \dots , T_\ell)$ such that $s/2 \leq |S_1| \leq s, |S_1| \leq |S_2| \leq m$ and $k/2 \leq |S_1||S_2| \leq k$ and 
\[
|\langle \Sum(S_1), \Sum(S_2) \rangle | \geq \sqrt{\lambda |S_1||S_2|d}
\]
and otherwise we say that $S\backslash \Filter_{\gamma}(T_1, \dots , T_\ell)$  is clean. Similarly, if $A\backslash \Filter_{\gamma}(T_1, \dots , T_\ell)$ contains two disjoint sets $S_{1,A}', S_{2,A}'$ such that 
\[
\begin{split}
|S_{1,A}'| &= \frac{p^2 s}{2 \cdot 10^6} \\
|S_{2,A}'| &= \frac{p^2 k}{2\cdot 10^6 s} \\
| \langle \Sum(S_{1,A}') , \Sum(S_{2,A}') \rangle | &\geq \frac{p^4\sqrt{\lambda |S_1||S_2| d}}{10^{14}}
\end{split}
\]
then we say $A\backslash \Filter_{\gamma}(T_1, \dots , T_\ell)$ is unclean and otherwise we say that it is clean.

There are at most $|S|^{\tau}$ distinct filters considered in Algorithm~\ref{alg:single-iteration}.  For each of these filters $\{ T_1, \dots , T_\ell \}$, we apply Lemma~\ref{lem:violating-sets-split} to $S \backslash \Filter_{\gamma}(T_1, \dots , T_\ell)$.  If it is unclean, then, with probability at least $ 1 - 2^{-s p^6/10^{10}}$, $A$ is unclean. 
This is because if $S_1, S_2$ witness $S \backslash \Filter_{\gamma}(T_1, \dots , T_\ell)$ being unclean, then since $s/2 \leq  |S_1| \leq s$ and $k/(2s) \leq |S_2| \leq 2k/s$, we can choose random subsets $S_{1,A}', S_{2,A}'$ of the appropriate size from the sets guaranteed by Lemma~\ref{lem:violating-sets-split}. (Note: we can apply Lemma~\ref{lem:violating-sets-split} because $|S_1| |S_2| \le k$ and $|S_1| \le |S_2| \le m$ by definition of unclean, and since $S \supset S_1 \cup S_2$ is assumed to be $(\lambda, m, k/2)$-balanced.) Now we can union bound this over all $|S|^\tau$ distinct filters and since by the definition of $\tau$,
\[
|S|^{\tau} \leq 2^{\frac{sp^6}{5 \cdot 10^{10}}}
\]
and thus with probability $1 - 2^{-sp^6/10^{11}}$, we have the following property: 

For any $\{T_1, \dots , T_\ell \}$ if $A\backslash \Filter_{\gamma}(T_1, \dots , T_\ell)$ is clean then $S\backslash \Filter_{\gamma}(T_1, \dots , T_\ell)$ is clean.  If Algorithm~\ref{alg:single-iteration} chooses $F_1, \dots , F_\ell$ such that  $S\backslash \Filter_{\gamma}(F_1, \dots , F_\ell)$ is clean then we are done.  Thus, it remains to show that there actually exists a filter $F_1, \dots , F_\ell$ that cleans $A$. 

We construct such a filter iteratively.  Start with an empty filter.  Now if we are not done, then there must exist a pair $S_{1,A}', S_{2,A}'$ that witnesses $A$ being unclean.  We will apply Lemma~\ref{lem:filter-compression} on this pair (with $\theta \leftarrow \lambda p^4/{10^{20}}, C = 10^{30}/p^6$, $k \leftarrow k/2$).  First we verify that the conditions of Lemma~\ref{lem:filter-compression} are met.  We have 
\[
|S_{1,A}'| |S_{2,A}'| = \frac{p^4k}{ 4 \cdot 10^{12}} \geq \frac{10^{37} m}{p^{16}} \geq 10Cm
\]
and also clearly $|S_{1,A}'| |S_{2,A}'|  \leq k/2$. Also, $|S'_{1, A}|, |S'_{2, A}| \le m$ since even $|S_1|, |S_2| \le m$. Recall that we have
\[
| \langle \Sum(S_{1,A}') , \Sum(S_{2,A}') \rangle | \geq \frac{p^4\sqrt{\lambda |S_1||S_2| d}}{10^{14}} \geq \sqrt{\theta |S_{1,A}'||S_{2,A}'|d}
\]
and also $S_{1,A}' \cup S_{2,A}'$ is $(\lambda ,m , k/2)$-balanced by assumption.  Finally, 
\[
\frac{\lambda p^{50}}{10^{100}m} \leq \frac{\theta}{10^5 C m },
\]
so the boundedness condition is satisfied, and clearly $\theta \leq \lambda \leq d$ and $C \geq 1$.  \todonotedone{Shyam: This was very hard for me to verify, mainly the fact that the conditions of Lemma 4.12 are met. You should write a couple sentences explaining why we can apply it. (Also see my above comment about $m$ needing to be smaller.) Also, I assume you are applying the lemma where you use $k/2$ instead of $k$, right?}  Thus, Lemma~\ref{lem:filter-compression} tells us that we can find a subset $F_1$ with $|F_1| \leq p^6k/(10^{40}m)$ such that 
\[
\left\lvert \left\{ v \in A \bigg| |\langle v , \Sum(F_1)/\norm{\Sum(F_1)}  \rangle | \geq \sqrt{\frac{\lambda p^{50}}{10^{100}m}}  \right\} \right\rvert \geq \frac{p^6 k}{10^{30} s} \,.
\] 
Thus, by our assumption on $R$, the number of elements in the above set that are in $R$ is at least $\frac{p^6 \cdot k}{10^{30} s} - 2 |F_1|$, which is at least $\frac{p^6 k}{2 \cdot 10^{30} s}$, since $s/2 \le |S_1| \le m$.  In particular, we added at most $p^6k/(10^{40}m)$ elements to our filter and eliminated at least $\frac{p^6 k}{2 \cdot 10^{30} s}$ elements of $R$.  Now we can iterate the above argument on $A \backslash \Filter_{\gamma}(F_1)$.  Overall, repeating this process, the total number of elements that we will add to our filter is at most
\[
\frac{p^6k}{10^{40}m}\left( \frac{|R|}{\frac{p^6 k}{2 \cdot 10^{30} s}} + 1\right) \leq \tau \,.
\]
This completes the proof.
\end{proof}

\begin{algorithm}
\begin{algorithmic}
\Require Finite set of samples $S \subset \R^d$
\Require Parameters $\lambda , m, \delta$
\State Set $k = 10^{200}m \log^{100} (|S|m/\delta) $
\State Set $S_{\textsf{filt}} = S$
\While{$k \leq m^2$}
\State Set $s =10^{199} \log^{100}(|S|m/\delta)$
\While{$s \leq m$}
\State Run Algorithm~\ref{alg:single-iteration} on $S$ with parameters $\lambda , m ,k , s, p = 1/(5\log^2 m) $
\State Set $S_{\textsf{filt}} \leftarrow B'$  where $B'$ is the output of Algorithm~\ref{alg:single-iteration}
\State $s \leftarrow 2s$
\EndWhile
\State $k \leftarrow 2k$
\EndWhile
\State\Return $S_{\textsf{filt}}$
\end{algorithmic}
\caption{Full Sample Splitting}\label{alg:full-sample-splitting}
\end{algorithm}

\begin{lemma}[Analysis of Algorithm~\ref{alg:full-sample-splitting}]\label{lem:full-sample-splitting}
Let $S \subset \R^d$  be a finite set of vectors and $\lambda , m ,\delta$ be some parameters with $\lambda \leq d$.  Assume that $S$ is $\gamma^2$-bounded where $\gamma = \sqrt{\frac{\lambda}{10^{200} m\log^{100} (|S|m/\delta) }}$.  Also assume that there is a subset $R \subset S$ with $|R| \leq \frac{m}{10^{200} \log^{100} (|S|m/\delta)}$ such that for all subsets $T \subset S$, \[
\left\lvert \left\{ v \in S \backslash R \; \bigg| \; |\langle v , \Sum(T)/\norm{\Sum(T)}\rangle | \geq \gamma  \right\} \right\rvert \leq 2|T| \,.
\]
Then if we run Algorithm~\ref{alg:full-sample-splitting} on $S$, with probability $1 - \delta$, the output $S_{\textsf{filt}}$ will be $(\lambda, m ,m^2)$-balanced.
\end{lemma}
\begin{proof}
First by Claim~\ref{claim:small-sets-bound}, we have that $S$ is $(\lambda , m' , m')$-balanced with $m' = m \cdot 10^{200}\log^{100}(|S|m/\delta)$ (and thus also $(\lambda , m , m')$-balanced).  Now we prove that after every execution of the outer while loop (for a fixed value of $k$, before doubling $k$) in Algorithm~\ref{alg:full-sample-splitting}, the set $S_{\textsf{filt}}$ will be $(\lambda , m , k)$-balanced.  We do this by induction, where the base case follows from the preceding statement.  Now after doubling $k$, we know that $S_{\textsf{filt}}$  is $(\lambda , m, k/2)$-balanced.  Next we apply Lemma~\ref{lem:filter-one-iteration} for each execution of the inner while loop in Algorithm~\ref{alg:full-sample-splitting}.  Note that this is valid because $\lambda \leq d$, $k$ is initialized sufficiently large, and our upper bound on $|R|$ is sufficiently small.  Also $s$ is initialized sufficiently large so we can union bound the failure probability over all iterations and deduce that with probability $1 - \delta$, the conclusion of Lemma~\ref{lem:filter-one-iteration} every time we run Algorithm~\ref{alg:single-iteration}.  If after the completion of the inner while loop, the set $S_{\textsf{filt}}$ is not $(\lambda , m , k)$-balanced then there must be some disjoint $|S_1|,|S_2|$ with $|S_1|, |S_2| \leq m,$ $k/2 \le |S_1||S_2| \leq k$, and $|\langle \Sum(S_1), \Sum(S_2) \rangle | \geq \sqrt{\lambda |S_1||S_2|d}$. WLOG $|S_1| \leq |S_2|$.  By the inductive hypothesis, we must have $|S_1||S_2| \geq k/2$ and since $|S_2| \leq m $,
\[
|S_1| = \frac{|S_1| \cdot |S_2|}{|S_2|} \geq \frac{k/2}{m} \geq 5 \cdot 10^{199} \log^{100} (|S|m/\delta)
\]
and thus there was some value of $s$ for which we executed the inner while loop and $s/2 \leq  |S_1| \leq s$. However, applying the guarantee of Lemma~\ref{lem:filter-one-iteration} for this execution of Algorithm~\ref{alg:single-iteration} implies that such $S_1, S_2$ cannot exist and this is a contradiction.  Thus, actually  $S_{\textsf{filt}}$ must be $(\lambda , m , k)$-balanced and this completes the induction.  Since we keep increasing $k$ up to $m^2$, at the end we know that $S_{\textsf{filt}}$ is $(\lambda ,m , m^2)$-balanced and we are done.
\end{proof}

Now we can use Lemma~\ref{lem:full-sample-splitting} to prove Theorem~\ref{thm:reduce-to-friendly}.

\begin{proof}[Proof of Theorem~\ref{thm:reduce-to-friendly}]
Consider starting with a set $S$ of $n (10\log (nd))^{10}$ obliviously $\eps$-contaminated samples.  First, we remove all points $X_i \in S$ with $\norm{X_i}^2 \geq d +  \sqrt{(\log (nd))^{100}d}$ or $\norm{X_i}^2 \leq d -  \sqrt{(\log (nd))^{100} d}$.  Next, for all pairs of distinct points $X_i, X_j$ with $| \langle X_i, X_j \rangle | \geq \sqrt{(\log (nd))^{200} d}$, we remove both of them.  By Fact~\ref{fact:bounded-norm} and Claim~\ref{claim:inner-product}, with probability $0.999$, this only removes contaminated points.  Furthermore, the remaining dataset is equivalent to an obliviously $\eps$-contaminated one (since it is equivalent to first remove the subset of contaminated points that violate the previous conditions and then draw the uncontaminated points).

\paragraph{Case 1: $\eps n \lesssim d$} 
We first consider the case where $(10 \log(nd))^{1000} \eps n \leq  d$.  We run Algorithm~\ref{alg:full-sample-splitting} with $\delta = 0.001$ and 
\[
\begin{split}
 m &= \eps n (10 \log(nd))^{200}\\
 \lambda &= \eps n (10 \log(nd))^{1000} \,.
\end{split}
\]
Recall that Algorithm~\ref{alg:full-sample-splitting} runs $O(\log^2 m)$ iterations of Algorithm~\ref{alg:single-iteration}.  For all executions of Algorithm~\ref{alg:single-iteration}, we have $\gamma \geq (10 \log(nd))^{300}$.  Also note that we can imagine drawing the uncontaminated points in $B$ after drawing the points in $A$.  On the other hand, the filters are constructed only from $A$.  Thus, with $1 - 1/(nd)^{100}$ probability, none of the uncontaminated points are removed by the filters.  We can union bound this failure probability over all executions of Algorithm~\ref{alg:single-iteration} to get that with probability $0.999$, no uncontaminated points are removed by any filters throughout the execution of Algorithm~\ref{alg:full-sample-splitting}.  By the construction of $S_{\textsf{filt}}$, we conclude that with $0.999$ probability $|S_{\textsf{filt}}| \geq n$ and the number of contaminated points in $S_{\textsf{filt}}$ is at most $\eps n$.  Also, none of the filters constructed throughout Algorithm~\ref{alg:full-sample-splitting} depend on the points in   $S_{\textsf{filt}}$ so it is equivalent to an obliviously $\eps$-contaminated dataset (since it is equivalent to simply apply these filters to the contaminated points before drawing the rest of the dataset).  It remains to argue that with high probability, $S_{\textsf{filt}}$ is $\kappa$-friendly and then we can apply the tester that we assumed works under $\kappa$-friendly oblivious $\eps$-contamination to complete the proof.

Note that $\lambda = \eps n (10 \log(nd))^{1000} \leq  d$ by assumption and after the initial filtering step (where we filter by norm and pairwise inner product), we know that the dataset is $(\log (nd))^{200}$-bounded.  Also, by Claim~\ref{claim:projection-counts}, there is a set $R \subset S_{\textsf{filt}}$ with $|R| \leq \eps n$ (consisting of exactly the contaminated points) such that  for all subsets $T \subset S_{\textsf{filt}}$, 
\[
\left\lvert \left\{ v \in S \backslash R \; \bigg| \; |\langle v , \Sum(T)/\norm{\Sum(T)}\rangle | \geq 100 \log n \right\} \right\rvert \leq 2|T| \,.
\]
Thus, we can apply Lemma~\ref{lem:full-sample-splitting} to get that $S_{\textsf{filt}}$ is $(\lambda ,m , m^2)$-balanced.  This then implies that $S_{\textsf{filt}}$ is equivalent to a $\kappa$-friendly obliviously $\eps$-contaminated dataset and we are done in this case.

\paragraph{Case 2: $\eps n \gtrsim d$} 
Now it remains to consider the case where $(10 \log(nd))^{1000} \eps n \geq  d$.  We can increase the dimension by adding dummy coordinates to all of the points.  We can draw these coordinates independently from $N(0,1)$ and pad the dimension to $d' = (10 \log(nd))^{1000} \eps n$.  We will use $S'$ to denote the padded dataset and $X_i'$ to denote points in $S'$.  Recall that we filtered by norm and inner product at the beginning.  Since all of the additional coordinates are i.i.d. standard Gaussians, with $0.999$ probability, we have that after the padding, for all $X_i' \in S'$
\[
d' - \sqrt{(\log(nd'))^{100} d'} \leq \norm{X_i'}^2 \leq d' + \sqrt{(\log(nd'))^{100} d'} 
\]
and for all distinct $X_i', X_j' \in S'$,
\[
|\langle X_i', X_j' \rangle| \leq \sqrt{(\log nd'))^{200}d'} \,.
\]
Now we can run Algorithm~\ref{alg:full-sample-splitting} as in the previous case on the padded points.  By the same argument, we end up with an obliviously $\eps$-contaminated dataset $S_{\textsf{filt}}'$ such that $S_{\textsf{filt}}'$ is $(\log(nd'))^{200}$-bounded and $(\eps n (\log (nd'))^{1000},\eps n (\log (nd'))^{200}, \eps^2 n^2 (\log (nd'))^{400} )$-balanced.  WLOG say $S_{\textsf{filt}}' = \{X_1', \dots , X_n' \}$. Now we take $S_{\textsf{filt}}'$  and remove the padding to get $S_{\textsf{filt}} = \{X_1, \dots , X_n \}$.  Let $\Pi_{\textsf{pad}}$ denote the operator that projects onto the padded coordinates.  For a set $A \subset [n]$, $\sum_{i \in A} \Pi_{\textsf{pad}} X_i'$ is just a vector in $\R^{d' - d}$ distributed according to $N(0, |A|I)$.  Union bounding over all choices of disjoint sets $A,B \subset [n]$ with $|A|,|B| \leq \eps n$ using \Cref{fact:concentration-inner-prod}, with probability $0.999$ over the randomness in the padded coordinates,
\[
\left\lvert \left\langle \sum_{i \in A} \Pi_{\textsf{pad}}X_i', \sum_{i \in B} \Pi_{\textsf{pad}} X_i' \right\rangle \right\rvert \leq \sqrt{|A||B| \eps n  d'}(\log( nd'))^{100} \leq \sqrt{|A||B| }\eps n (10\log (nd))^{1000} \,.
\]
Combining the above and the balancedness of $S_{\textsf{filt}}'$, we get that after removing the padding for any disjoint sets $A,B \subset [n]$, 
\[
\left\lvert \left\langle \sum_{i \in A} X_i, \sum_{i \in B} X_i \right\rangle \right\rvert \leq \sqrt{|A||B| }\eps n (10\log (nd))^{1000} + \sqrt{\eps n (\log (nd'))^{1000}|A||B|d'} \leq \kappa \sqrt{|A||B|}\eps n   \,.
\]
Thus, $S_{\textsf{filt}}$ is $\kappa$-friendly (recall the latter two conditions follow from the filtering by norm and inner product that we did at the beginning) and we are done.
\end{proof}

\section{Mean Testing Robustly Against Oblivious Adversaries}
\label{sec:obliviously-robust-tester}

In this section, we prove our main technical result, the upper bound against oblivious adversaries:
\begin{theorem}
    \label{thm:oblivious-ub}
    Suppose that $n \ge \tilde{O}\mleft(\frac{\sqrt{d}}{\alpha^2} + \frac{d \eps^3}{\alpha^4} + \min\mleft(\frac{d^{2/3} \eps^{2/3}}{\alpha^{8/3}}, \frac{d \eps}{\alpha^2}\mright)\mright)$, and that $1 \ge \alpha \ge \eps \cdot \log(n d)^{O(1)}$. Then, there exists an $\eps$-robust mean tester using $n$ samples.
\end{theorem}

\subsection{Setup and Algorithm}

From \Cref{sec:strong-sample-splitting}, we may assume we are dealing with the friendly oblivious contamination model. We restate the definition for convenience.

\friendly*
\begin{assumption} \label{assumption}
A dataset $X_1,\ldots,X_n \in \R^d$ is $\kappa$-\emph{friendly} if the  following all hold:
\begin{enumerate}
    \item For any disjoint subsets $S, T \subset [n]$ of sizes $k_1, k_2 \le \eps \cdot n$,
\[\left|\left\langle \sum_{i \in S} X_i, \sum_{i\in T} X_i \right\rangle\right| \le \kappa \cdot (\sqrt{k_1 k_2} \cdot \max(\sqrt{\eps n d}, \eps n)).\] \label{assumption:cross-sum}
    \item For every distinct $i \neq j \in [n]$, $|\langle X_i, X_j \rangle| \le \kappa \cdot \sqrt{d}$. \label{assumption:inner-product}
    \item For every $i \in [n]$, $\|X_i\|_2^2 = d \pm \kappa \sqrt{d}$. \label{assumption:norm}
\end{enumerate}
\end{assumption}
\noindent We think of $\kappa$ as a sufficiently large $\log(nd)^{O(1)}$ term.

\medskip

Given this promise on the data, the algorithm, roughly speaking, checks the mean and the variance in the direction of the sum of the points. If both look reasonable for a set of samples from the null distribution, we accept, otherwise, we reject. Formally, we use the algorithm described in \autoref{alg:oblivious-tester}.

\begin{algorithm}[H]

\caption{Robust mean tester for obliviously-corrupted data satisfying Assumption~\ref{assumption}. Input: $X_1,\ldots,X_n \in \R^d$, $\alpha,\eps > 0$.}
\label{alg:oblivious-tester}

    \begin{algorithmic}[1]
\State Let $\mathbf{S} := \sum_{i \in [n]} X_i$.
\If{$\left | \|\mathbf{S}\|_2^2 - nd\right| > 0.01 \alpha^2 n^2$}
\State \textbf{return} \textbf{REJECT}.
\ElsIf{$\frac{1}{n} \sum_{i \in [n]} \left(\frac{\langle X_i, \mathbf{S} \rangle - d}{\|\mathbf{S}\|_2}\right)^2 \ge 1 + 0.025 \frac{\alpha^4}{\eps} \cdot \frac{n}{d}$ \textbf{and} $n \le O(\kappa^5) \cdot \left(\frac{\sqrt{d}}{\alpha^2} + \frac{d \eps}{\alpha^2}\right)$}
\State \textbf{return} \textbf{REJECT}.
\Else
\State \textbf{return} \textbf{ACCEPT}.
\EndIf
\end{algorithmic}

\end{algorithm}

We have the following results which lead to our main theorem.

\begin{lemma} \label{lem:mean-var-main}
    Suppose that $X_1, \dots, X_n$ are drawn from the friendly $\eps$-oblivious contamination model. Moreover, assume that $d \ge \eps \cdot n$ and $n \le \frac{d}{\alpha^2}$. Then, \autoref{alg:oblivious-tester} can solve robust mean testing in $n = O(\kappa^5) \cdot \left(\frac{\sqrt{d}}{\alpha^2} + \frac{d \eps^3}{\alpha^4} + \frac{d^{2/3} \eps^{2/3}}{\alpha^{8/3}}\right)$ samples, whenever $\alpha \ge \kappa^5 \cdot \eps$.
\end{lemma}

\begin{lemma} \label{lem:mean-var-main-2}
    Suppose that $X_1, \dots, X_n$ are drawn from the friendly $\eps$-oblivious contamination model. Then, \autoref{alg:oblivious-tester} can solve robust mean testing in $n = O(\kappa^5) \cdot \left(\frac{\sqrt{d}}{\alpha^2} + \frac{d \eps}{\alpha^2}\right)$ samples, whenever $\alpha \ge \kappa^5 \cdot \eps$.
\end{lemma}

We can think of ``solving robust mean testing'' (as stated in Lemmas \ref{lem:mean-var-main} and \ref{lem:mean-var-main-2}) as follows.
If given $n$ $\eps$-obliviously contaminated samples from $\cN(0, I)$, with high probability we either return ACCEPT or the samples were not $\kappa$-friendly. Likewise, if given $n$ $\eps$-obliviously contaminated samples from $\cN(0, I)$, with high probability we either return REJECT or the samples were not $\kappa$-friendly. In this section, we always use the phrase \textbf{\emph{with high probability}} to mean the failure probability is at most $\frac{1}{\poly(n, d)}$.

\medskip

Given Lemmas \ref{lem:mean-var-main} and \ref{lem:mean-var-main-2}, we explain how to prove our main theorem, Theorem~\ref{thm:oblivious-ub}.

\begin{proof}[Proof of \Cref{thm:oblivious-ub}]
    First, we assume that the data was drawn from the $\kappa$-friendly $\eps$-oblivious contamination model.
    Suppose that $n \ge \kappa^6 \cdot \left(\frac{\sqrt{d}}{\alpha^2} + \frac{d \eps^3}{\alpha^4} + \min\left(\frac{d^{2/3} \eps^{2/3}}{\alpha^{8/3}}, \frac{d \eps}{\alpha^2}\right)\right)$ and $\alpha \ge \kappa^5 \cdot \eps$. Then, if $n \ge \frac{d}{\alpha^2},$ we have $n \ge \kappa^{10} \cdot \frac{d \eps^2}{\alpha^4},$ so we can use~\Cref{thm:adaptive-sample-complexity-main}. Alternatively, if $n \le \frac{d}{\alpha^2}$ and $d \ge \eps \cdot n$, we can use either Lemma \ref{lem:mean-var-main} or \ref{lem:mean-var-main-2}. Finally, if $d \le \eps \cdot n,$ but $n \ge \kappa^6 \cdot \frac{\sqrt{d}}{\alpha^2},$ then $\kappa^5 \cdot \frac{d \eps}{\alpha^2} = \kappa^5 \cdot \frac{d}{\eps} \cdot \left(\frac{\eps}{\alpha}\right)^2 \le \kappa^{-5} \cdot \frac{d}{\eps}$. Since $n \ge \kappa^{-1} \cdot \frac{d}{\eps},$ this means that $O(\kappa^5) \cdot \frac{d \eps}{\alpha^2} \le O(\kappa^{-4}) \cdot n$. Therefore, $n \ge O(\kappa^5) \cdot \frac{\sqrt{d}}{\alpha^2} + O(\kappa^5) \cdot \frac{d \eps}{\alpha^2},$ so we can apply \Cref{lem:mean-var-main-2}.

    By \Cref{thm:reduce-to-friendly}, we may remove the assumption about the data being friendly. This completes the proof.
\end{proof}

The rest of the section is primarily devoted to \Cref{lem:mean-var-main}, but we prove \Cref{lem:mean-var-main-2} in \Cref{subsec:last-case-oblivious-lb}.
By \Cref{prop:wlog-bounded}, we may assume that the alternative hypothesis is $\|\mu\|_2 \in [\alpha, 2 \alpha]$. In fact, for simplicity we will pretend the alternative is $\|\mu\|_2 = \alpha$. Indeed, if $\|\mu\|_2 = \alpha' \in [\alpha, 2 \alpha]$, then our proof will show that either $\left|\|\mathbf{S}\|_2^2-nd \right| > 0.01 (\alpha')^2 n^2 \ge 0.01 \alpha^2 n^2$ or $\frac{1}{n} \sum_{i \in [n]} \left(\frac{\langle X_i, \mathbf{S} \rangle - d}{\|\mathbf{S}\|_2}\right)^2 \ge 1 + 0.025 \frac{(\alpha')^4}{\eps} \cdot \frac{n}{d} \ge 1 + 0.025 \frac{\alpha^4}{\eps} \cdot \frac{n}{d}$.

In the rest of this section, we use $\mathbf{S}$ to represent $\Sum([n]) = \sum_{i \in [n]} X_i$. We also will split the data into good (uncorrupted) points $G$ and bad (corrupted) points $B$. We will always use $\mathbf{R}$ to denote $\Sum(B) = \sum_{i \in B} X_i$. If the good samples are drawn as $\cN(\mu, I)$, we always use $\mathbf{T}$ to denote $\sum_{i \in G} (X_i-\mu)$, and $\mathbf{Q}$ to denote $|G| \cdot \mu$. Note that $\mathbf{S} = \mathbf{Q}+\mathbf{R}+\mathbf{T}$. Also, note that in the null case, $\mathbf{Q} = 0$ and $\mathbf{T} = \sum_{i \in G} X_i$, which means $\mathbf{S} = \mathbf{R}+\mathbf{T}$.

\subsection{Consequences of \Cref{assumption}} \label{subsec:consequences-of-assumption}

In this section, we prove a series of propositions that will be useful in bounding the mean and variance.
In all of the following, let $X_1,\ldots,X_n \in \R^d$ be any vectors satisfying \Cref{assumption}.

First, we have the following bound on the norm of any sum of at most $\eps n$ points.

\begin{proposition} \label{prop:norm-sum}
    Let $X_1,\ldots,X_n \in \R^d$ satisfy Assumption~\ref{assumption}. Then, for any subset $B'$ of size $k \le \eps n$, $\left\|\sum_{i \in B'} X_i\right\|_2^2 = kd \pm O(\kappa) \cdot k \cdot (\sqrt{\eps n d} + \eps n)$.  
\end{proposition}

\begin{proof}
    Let $B'_1,B'_2$ be a random partition of $B'$ into sets of equal size.
    For any distinct $i,j \in B'$, let $p = \Pr(i \in B'_1, j \in B'_2) = \Omega(1)$.
    Then 
    \[
    \left | \sum_{i \neq j \in B'} \iprod{X_i,X_j} \right | =  \left | \frac 1 p \cdot \E_{B'_1,B'_2} \bigg\langle \sum_{i \in B'_1} X_i, \sum_{j \in B'_2} X_j \bigg\rangle \right | \leq \kappa \cdot k \cdot O(\sqrt{\eps n d} + \eps n),
    \]
    by \Cref{assumption:cross-sum} of \Cref{assumption}.
    Finally, $\sum_{i \in B'} \|X_i\|_2^2 = k d \pm k \cdot \kappa \sqrt{d}$. Overall, this means $\left\|\sum_{i \in B'} X_i\right\|_2^2 = kd \pm \kappa \cdot k \cdot O(\sqrt{\eps n d} + \eps n)$.
\end{proof}

Next, we have the following proposition.

\begin{proposition} \label{prop:null_bad_R_concentration}
    Let $X_1,\ldots,X_n$ satisfy \Cref{assumption}, with $\eps n \le d$. 
    Let $B \subset [n]$ be any subset of size $\eps n$, and let $\mathbf{R} = \sum_{i \in B} X_i$. Let $B' \subset B$ be a subset of size $k \le \eps \cdot n$. Then, $\sum_{i \in B'} \langle X_i, \mathbf{R} \rangle = kd \pm O(\kappa \cdot \eps n \sqrt{kd})$.
\end{proposition}

\begin{proof}
    We have 
\[\sum_{i \in B'} \langle X_i, \mathbf{R} \rangle = \left\langle \sum_{i \in B'} X_i, \sum_{i \in B} X_i \right\rangle = \left\|\sum_{i \in B'} X_i \right\|_2^2 + \left\langle \sum_{i \in B'} X_i, \sum_{i \in B \backslash B'} X_i \right\rangle.\]
    By \Cref{prop:norm-sum}, we know that $\|\sum_{i \in B'} X_i \|_2^2 = k d \pm O(\kappa \cdot k \sqrt{\eps n d})$. By \Cref{assumption:cross-sum} of \Cref{assumption}, we know that $\left|\langle \sum_{i \in B'} X_i, \sum_{i \in B \backslash B'} X_i \rangle\right| \le \kappa \cdot \sqrt{k \cdot \eps n} \cdot \sqrt{\eps n d} = \kappa \cdot \eps n \cdot \sqrt{k d}$. Since $k \le \eps n$, $k \sqrt{\eps n d} \le \eps n \sqrt{k d}$, so overall, we have that $\sum_{i \in B'} \langle X_i, \mathbf{R} \rangle = kd \pm O(\kappa \cdot \eps n \sqrt{k d})$.
\end{proof}

As a result, we have the following.

\begin{proposition} \label{prop:null-bad-R-variance}
    Let $X_1,\ldots,X_n$ satisfy \Cref{assumption}, with $\eps n \le d$. Then, for any subset $B \subset [n]$ of size $\eps n$, we have that $\sum_{i \in B} (\langle X_i, \mathbf{R} \rangle - d)^2 \le O(\kappa^3 \cdot \eps^2 n^2 d),$ where $\mathbf{R} := \sum_{i \in B} X_i$.
\end{proposition}

\begin{proof}
    For each $i \in B$, define $y_i := \langle X_i, \mathbf{R} \rangle - d$. Consider the $k$th largest $y_i$. It must be at most $O(\kappa \cdot \eps n \sqrt{d/k})$, or else the sum of the $k$ largest $y_i$ would exceed $O(\kappa \cdot \eps n \sqrt{kd})$, contradicting Proposition \ref{prop:null_bad_R_concentration}. Likewise, the $k$th smallest $y_i$ must be greater than or equal to $-O(\kappa \eps n \sqrt{d/k})$, so the $k$th largest $|y_i|$ is at most $O(\kappa \eps n \sqrt{d/k})$.

    This means that
\[\sum_{i \in B} (\langle X_i, \mathbf{R} \rangle - d)^2 = \sum_{i \in B} |y_i|^2 \le \sum_{k=1}^{\eps n} O\left(\kappa \cdot \eps n \sqrt{\frac{d}{k}}\right)^2 = \sum_{k=1}^{\eps n} O\left(\frac{\kappa^2 \eps^2 n^2 d}{k}\right) = O(\kappa^3 \eps^2 n^2 d). \qedhere\]
\end{proof}

We also have the following bound.

\begin{proposition} \label{prop:bad-cross-terms}
    Let $X_1,\ldots,X_n$ satisfy \Cref{assumption}, with $\eps n \le d$. Let $B \subset [n]$ have size $\eps n$, and let $\mathbf{R} = \sum_{i \in B} 
    X_i$. Then, $\left\|\sum_{i \in B} (\langle X_i, \mathbf{R} \rangle-d) X_i\right\|_2 \le O(\kappa^2 \cdot \eps n d)$.
\end{proposition}

\begin{proof}
    Write $y_i := \langle X_i, \mathbf{R} \rangle-d$. We can then write 
\begin{align*}
    \sum_{i \in B} (\langle X_i, \mathbf{R} \rangle-d) X_i &= \sum_{i \in B} y_i X_i \\
    &= \sum_{i \in B \, : \, y_i \geq 0} \int_0^{y_i} X_i \, dt - \sum_{i \in B \, : \, y_i < 0} \int_0^{|y_i|} X_i \, dt \\
    &= \underbrace{\int_0^\infty \left(\sum_{i \in B: y_i \ge t} X_i\right) dt}_{A_{+}} - \underbrace{\int_0^\infty \left(\sum_{i \in B: y_i \le -t} X_i\right) dt}_{A_{-}}.
\end{align*}
    For simplicity, we will just bound $\|A_+\|_2,$ as the argument for bounding $\|A_-\|_2$ is identical. As in \Cref{prop:null-bad-R-variance}, for any real number $t \ge 0$, the number of indices $i \in B$ such that $\langle X_i, \mathbf{R} \rangle - d \ge t$ is some $k(t)$ for $t \le O(\kappa \cdot \eps n \sqrt{d/k(t)})$, so $k(t) \le O\left(\frac{\kappa^2 \eps^2 n^2 d}{t^2}\right)$ for all $t$. In addition, $k(t) \le \eps n$ always, because we are only counting indices in $B$.

    For any $t \ge 0$, define $A_+(t) := \left\|\sum_{i \in B: y_i \ge t} X_i\right\|_2$.
    By \Cref{prop:norm-sum} (which we can apply since $k(t) \le \eps n$),
\[\|A_+(t)\|_2 = O\left(\sqrt{k(t) \cdot d} + \sqrt{\kappa \cdot k(t)} (\eps n d)^{1/4} \right) = O\left(\frac{\kappa \eps n d}{t} + \frac{\kappa^{3/2} (\eps n)^{5/4} d^{3/4}}{t}\right) \le O\left(\frac{\kappa^{1.5} \eps n d}{t}\right),\]
    since we are assuming $d \ge \eps n$.

    In addition, for any $t \ge 0$, as  $k(t) \le \eps n$, we have $\|A_+(t)\|_2 \le O(\sqrt{\eps n d} + \sqrt{\kappa} (\eps n)^{3/4} d^{1/4}) = O(\sqrt{\kappa} \cdot \sqrt{\eps n d})$. Also, note that $A_+(t) = 0$ for $t \ge O(\kappa \eps n \sqrt{d})$, as this will imply $k(t) < 1$, so $k(t) = 0$.

    Overall, we have that
\begin{align*}
    \|A_+\|_2 &\le \int_0^\infty \|A_+(t)\|_2 dt \\
    &\le \int_0^{\kappa \eps n} O(\sqrt{\kappa \cdot \eps n d}) dt + \int_{\kappa \eps n}^{\kappa \cdot \eps n \sqrt{d}} O\left(\frac{\kappa^{1.5} \eps n d}{t}\right) dt.
\end{align*}
    The first integral is trivially bounded by $O(\kappa^{1.5} \cdot \eps n \sqrt{\eps n d}) \le O(\kappa^{1.5} \cdot \eps n d)$, since $\eps n \le d$.
    The second integral equals
\[\kappa^{1.5} \eps n d \cdot \log \frac{\kappa \eps n \sqrt{d}}{\kappa \eps n} \le \kappa^2 \cdot \eps n d.\]
    An identical calculation for $A_-$, combined with the triangle inequality, completes the proof.
\end{proof}

Finally, we will bound the Frobenius norm and operator norm of $\sum X_i X_i^\top$, over any subset of $\eps n$ points.

\begin{proposition} \label{prop:frobenius_bound}
    Suppose $B \subset [n]$ has size $\eps n$. Then, under \Cref{assumption}, and if $\eps n \le d$, we have that $\|\sum_{i \in B} X_i X_i^\top\|_F\le O(\kappa \cdot d \sqrt{\eps n})$.
\end{proposition}

\begin{proof}
    Note that 
\[\bigg\|\sum_{i \in B} X_i X_i^\top\bigg\|_F^2 = \sum_{i, j \in B} \Tr(X_i X_i^\top X_j X_j^\top) = \sum_{i, j \in B} \langle X_i, X_j \rangle^2.\]
    By \Cref{assumption:norm} of \Cref{assumption}, we know that for $i = j$, $\langle X_i, X_j \rangle^2 = \|X_i\|_2^4 \le \kappa^2 d^2,$ and by \Cref{assumption:inner-product} of \Cref{assumption}, for each $i \neq j$, $\langle X_i, X_j \rangle^2 \le \kappa^2 \cdot d$. So, because $|B| = \eps \cdot n \le d$,
\[\bigg\|\sum_{i \in B} X_i X_i^\top\bigg\|_F^2 \le O(\kappa^2 \cdot \eps n \cdot d^2 + \kappa^2 \cdot (\eps n)^2 \cdot d) = O(\kappa^2 \cdot \eps n \cdot d^2).\]
    We take the square root, and the result follows.
\end{proof}

\begin{proposition} \label{prop:operator_bound}
    Suppose $B \subset [n]$ has size $\eps n$. Then, under \Cref{assumption}, and if $\eps n \le d$, we have that $\|\sum_{i \in B} X_i X_i^\top\|_{op} \le O(\kappa^2 \cdot d)$.
\end{proposition}

\begin{proof}
    Choose a unit vector $w$. Since $\sum_{i \in B} X_i X_i^\top$ is PSD, it suffices to show that $w^\top \left(\sum_{i \in B} X_i X_i^\top\right) w = \sum_{i \in B} \langle X_i, w \rangle^2$ is at most $O(\kappa^2 \cdot d).$

    First, we consider the $k$th largest value of $\langle X_i, w \rangle$. For any subset $B' \subset B$ of $k \le \eps n$ elements, we have $\left\|\sum_{i \in B'} X_i\right\|_2 \le \sqrt{k d + \kappa \cdot k \sqrt{\eps n d} + \kappa \cdot k \cdot \eps n} \le O(\sqrt{\kappa \cdot kd})$, by \Cref{prop:norm-sum} and since $\eps n \le d$. Therefore, $\sum_{i \in B'} \langle X_i, w \rangle = \left\langle \sum_{i \in B'} X_i, w \right\rangle \le O(\sqrt{\kappa \cdot kd}).$
    This means that the $k$th largest value of $\langle X_i, w \rangle$ is at most $O(\sqrt{\kappa \cdot d/k})$, and the $k$th smallest value of $\langle X_i, w \rangle$ is at least $-O(\sqrt{\kappa \cdot d/k})$, which means the $k$th largest value of $\langle X_i, w \rangle^2$ is at most $O(\kappa \cdot d/k)$. Adding this over $1 \le k \le \eps \cdot n$, we have that $\sum_{i \in B} \langle X_i, w \rangle^2 \le \sum_{k=1}^{\eps n} O\left(\kappa \cdot \frac{d}{k}\right) = O(\kappa^2 \cdot d).$
\end{proof}

\subsection{The Null Case: Mean}

In this subsection, we verify that, under the null hypothesis and \Cref{assumption}, with high probability \autoref{alg:oblivious-tester} not reject on the first step, assuming sufficiently many samples. In this subsection, we do not assume that $\eps n \le d$.

Define $v = \left(\sum_{i \in [n]} X_i\right)/\|\sum_{i \in n} X_i \|_2$ to be the unit vector representing the \emph{direction} of the sum of all points. Also, define $z_i := \langle X_i, v \rangle$ for all $i \le n.$ Recall that $G \subset [n]$ represents the set of good (uncorrupted) data points, and $B = [n] \backslash G$ represents the set of bad (corrupted) data points. Note that $|G| = (1-\eps) n$ and $|B| = \eps n$. We now prove the main lemma for this subsection.

\begin{lemma} \label{lem:null-mean}
    Assume the null hypothesis, meaning that each $X_i$ for $i \in G$ is drawn i.i.d. as $\cN(0, I)$. Also, assume $n \ge \kappa^5 \cdot \left(\frac{\sqrt{d}}{\alpha^2} + \frac{d \eps^3}{\alpha^4}\right)$ and $\alpha \ge \kappa^5 \cdot \eps$.
    Then, under \Cref{assumption}, with high probability, $\left\|\sum_{i \in [n]} X_i\right\|_2^2 = n d \pm 0.01 \alpha^2 n^2.$
\end{lemma}

\begin{proof}
    We write
\[\left\|\sum_{i \in [n]} X_i\right\|_2^2 = \underbrace{\left\|\sum_{i \in G} X_i\right\|_2^2}_{a} 
 + \underbrace{\left\|\sum_{i \in B} X_i\right\|_2^2}_{b} + 2 \cdot \underbrace{\left\langle \sum_{i \in G} X_i, \sum_{i \in B} X_i \right\rangle}_{c}.\]
    Using standard concentration, we can write $a = (1-\eps) nd \pm O(\kappa n \sqrt{d})$ with high probability, and using \Cref{prop:norm-sum}, we can write $b = \eps n d \pm \kappa \cdot \eps n \cdot O(\sqrt{\eps n d}+\eps n)$. Finally, we know that the samples $X_i$ for $i \in G$ are drawn independently, from the samples in $B$, which means that with very high probability, $|c| \le O\left(\frac{\kappa}{\sqrt{d}}\right) \cdot \|\sum_{i \in G} X_i\|_2 \cdot \|\sum_{i \in B} X_i\|_2$. We can bound this as $O\left(\frac{\kappa}{\sqrt{d}} \cdot \sqrt{\kappa nd} \cdot \sqrt{\kappa}(\sqrt{\eps n d}+\eps n)\right) = \kappa^2 (\sqrt{\eps n^2 d} + \eps n^{3/2})$, by \Cref{prop:norm-sum}.

    Overall, we have that 
\[\left\|\sum_{i \in [n]} X_i\right\|_2^2 = n d \pm \kappa^2 \cdot O\left(n \sqrt{d} + (\eps n)^{3/2} \sqrt{d} + \eps^2 n^2 + \sqrt{\eps n^2 d} + \eps n^{3/2} \right).\]
    But note that $\sqrt{\eps n^2 d} < n \sqrt{d}$, so the only relevant error terms are the other ones, $n \sqrt{d}$, $(\eps n)^{3/2} \sqrt{d}$, $\eps^2 n^2$, and $\eps n^{3/2}$. By assuming that $n \ge \kappa^5 \cdot \left(\frac{\sqrt{d}}{\alpha^2}+\frac{d \eps^3}{\alpha^4}\right)$,
we have that the first two error terms $O(\kappa^2 n \sqrt{d})$ and $O(\kappa^2 (\eps n)^{3/2} d)$ are each bounded by $0.001 \alpha^2 n^2$. The third term $O(\kappa^2 \cdot \eps^2 n^2)$ is at most $0.001 \alpha^2 n^2,$ assuming that $\alpha \ge \kappa^5 \cdot \eps$. The final term $O(\kappa^2 \cdot \eps n^{3/2})$ is at most $0.001 \alpha^2 n^2$ assuming that $n \ge \kappa^5 \cdot \frac{\eps^2}{\alpha^4},$ which is true if $n \ge \kappa^5 \cdot \frac{\sqrt{d}}{\alpha^2}$ and $\eps \le \alpha$. Hence, we have that $\left\|\sum_{i \in [n]} X_i\right\|_2^2 = n d \pm 0.01 \alpha^2 n^2.$
\end{proof}

\subsection{The Null Case: Variance}
\label{sec:oblivious-tester-null-variance}

In this subsection, we verify that, under the null hypothesis and \Cref{assumption}, with high probability \autoref{alg:oblivious-tester} does not reject on the second step, assuming sufficiently many samples. Hence, \autoref{alg:oblivious-tester} accepts. In this and the next subsection, we may additionally assume that $\eps n \le d$ and $n \le \frac{d}{\alpha^2}.$ More formally, we make the following assumption in this subsection.

\begin{assumption} \label{assumption:general}
    We make \Cref{assumption}. In addition, we assume that $n \ge \kappa^5 \cdot \left(\frac{\sqrt{d}}{\alpha^2} + \frac{d \eps^3}{\alpha^4} + \frac{d^{2/3} \eps^{2/3}}{\alpha^{8/3}}\right)$ and $\alpha \ge \kappa^5 \cdot \eps$. Finally, we also assume that $\eps n \le d$ and $n \le \frac{d}{\alpha^2}.$
\end{assumption}

We recall that $\mathbf{R} := \sum_{i \in B} X_i$ and $\mathbf{T} = \sum_{i \in G} X_i$. Also, recall that $\mathbf{S} = \sum_{i \in [n]} X_i = \mathbf{R}+\mathbf{T}$, and $v = \mathbf{S}/\|\mathbf{S}\|_2$.
We wish to provide an upper bound for
\[\frac{1}{n} \sum_{i \in [n]} \left(\frac{\langle X_i, \mathbf{S}\rangle - d}{\|\mathbf{S}\|_2}\right)^2 = \frac{1}{n} \sum_{i \in [n]} \left(\frac{\langle X_i, \mathbf{R}\rangle + \langle X_i, \mathbf{T}\rangle - d}{\|\mathbf{S}\|_2}\right)^2.\]

First, we consider the bad terms, i.e., we bound
\[ \frac{1}{n} \sum_{i \in B} \left(\frac{\langle X_i, \mathbf{R}\rangle + \langle X_i, \mathbf{T}\rangle - d}{\|\mathbf{S}\|_2}\right)^2= \frac{1}{n \cdot \|\mathbf{S}\|_2^2} \sum_{i \in B} \left[\left(\langle X_i, \mathbf{R} \rangle - d\right)^2 + 2 \left(\langle X_i, \mathbf{R} \rangle - d\right) \cdot \langle X_i, \mathbf{T} \rangle + \langle X_i, \mathbf{T} \rangle^2 \right]. \]

Using the results from \Cref{subsec:consequences-of-assumption}, we prove the following bound.

\begin{lemma} \label{lem:null-bad-bound}
    Assume the null hypothesis. Then, under \Cref{assumption:general}, $\frac{1}{n} \sum_{i \in B} \left(\frac{\langle X_i, \mathbf{S}\rangle - d}{\|\mathbf{S}\|_2}\right)^2 \le \eps + 0.01 \cdot \frac{\alpha^4}{\eps} \cdot \frac{n}{d}$ with high probability.
\end{lemma}

\begin{proof}
    First, note that 
\begin{equation} \label{eq:null-bad-r}
    \sum_{i \in B} (\langle X_i, \mathbf{R} \rangle - d)^2 = O(\kappa^3 \cdot \eps^2 n^2 d)
\end{equation}
    by \Cref{prop:null-bad-R-variance}. Next, we can write $\sum_{i \in B} (\langle X_i, \mathbf{R} \rangle - d) \langle X_i, \mathbf{T} \rangle = \langle \mathbf{T}, \sum_{i \in B} (\langle X_i, \mathbf{R} \rangle - d) X_i \rangle$. However, by \Cref{prop:bad-cross-terms}, $\|\sum_{i \in B} (\langle X_i, \mathbf{R} \rangle - d) X_i\|_2 \le O(\kappa^2 \cdot \eps n d)$ and $\sum_{i \in B} (\langle X_i, \mathbf{R} \rangle - d) X_i$ is independent of $\mathbf{T}$, which is drawn as $\cN(0, (1-\eps) n I)$. Therefore, with high probability, 
\begin{equation} \label{eq:null-bad-rt}
   \left| \sum_{i \in B} (\langle X_i, \mathbf{R} \rangle - d) \langle X_i, \mathbf{T} \rangle\right| = \left|\Big\langle \mathbf{T}, \sum_{i \in B} (\langle X_i, \mathbf{R} \rangle - d) X_i \Big\rangle\right| \le O(\kappa \sqrt{n}) \cdot O(\kappa^2 \cdot \eps n d)
    = O(\kappa^3 \cdot \eps n^{3/2} d).
\end{equation}

    Finally, we can write $\sum_{i \in B} \langle X_i, \mathbf{T} \rangle^2 = \mathbf{T}^\top \left(\sum_{i \in B} X_i X_i^\top\right) \mathbf{T} = (1-\eps) n \cdot Z^\top \left(\sum_{i \in B} X_i X_i^\top\right) Z$, where $Z$ is a standard Gaussian independent of $\{X_i\}_{i \in B}$. We can then use the Hanson-Wright inequality (\Cref{lem:hanson-wright}) to say that with high probability, $\left|Z^\top \left(\sum_{i \in B} X_i X_i^\top\right) Z - \Tr(\sum_{i \in B} X_i X_i^\top)\right| \le O\left(\kappa \cdot \|\sum_{i \in B} X_i X_i^\top\|_F\right).$ We can write $\Tr(\sum_{i \in B} X_i X_i^\top) = \sum_{i \in B} \|X_i\|_2^2 = \eps n d \pm O(\kappa \cdot \eps n \sqrt{d})$, by \Cref{assumption:norm} of \Cref{assumption}. So, by using \Cref{prop:frobenius_bound}, we have that 
\begin{equation} \label{eq:null-bad-t}
    \sum_{i \in B} \langle X_i, \mathbf{T} \rangle^2 
    = (1-\eps) n \cdot (\eps n d \pm O(\kappa \cdot \eps n \sqrt{d} + \kappa^2 \cdot d \sqrt{\eps n}))
    \le \eps n^2 d + O(\kappa^2 \cdot \eps^{1/2} n^{3/2} d),
\end{equation}
    since we are assuming $\eps n \le d$.
    By combining Equations \eqref{eq:null-bad-r}, \eqref{eq:null-bad-rt}, and \eqref{eq:null-bad-t}, we have that
\begin{align*}
    \sum_{i \in B} (\langle X_i, \mathbf{S} \rangle - d)^2
    &= \sum_{i \in B} \left[\left(\langle X_i, \mathbf{R} \rangle - d\right)^2 + 2 \left(\langle X_i, \mathbf{R} \rangle - d\right) \cdot \langle X_i, \mathbf{T} \rangle + \langle X_i, \mathbf{T} \rangle^2 \right] \\
    &\le \eps n^2 d + \kappa^3 \cdot O(\eps^2 n^2 d + \eps^{1/2} n^{3/2} d).
\end{align*}
    As we are assuming that $n \ge \kappa^5 \cdot \frac{d\eps^3}{\alpha^4}$,
    this implies that $O(\kappa^3 \cdot \eps^2 n^2 d) \le 0.001 \cdot \frac{\alpha^4 n^3}{\eps}$.
    Moreover, we are assuming that $n \ge \kappa^5 \cdot \frac{d^{2/3} \eps^{2/3}}{\alpha^{8/3}} \ge \kappa^5 \cdot \frac{d^{2/3} \eps}{\alpha^{8/3}}$, which implies that $O(\kappa^3 \cdot \eps^{1/2} n^{3/2} d) \le 0.001 \cdot \frac{\alpha^4 n^3}{\eps}.$
In summary,
\begin{equation} \label{eq:null-bad-numerator}
    \sum_{i \in B} (\langle X_i, \mathbf{S} \rangle - d)^2 \le \eps n^2 d + 0.002 \cdot \frac{\alpha^4}{\eps} \cdot n^3.
\end{equation}

    Next, we note that $\|\mathbf{S}\|_2^2 = nd \pm 0.01 \alpha^2 n^2 = nd \cdot \left(1 \pm 0.01 \frac{\alpha^2 n}{d}\right),$ using \Cref{lem:null-mean}.
As a result, as we are assuming that $n \le \frac{d}{\alpha^2},$ 
then the reciprocal of $1 \pm 0.01 \frac{\alpha^2 n}{d}$ is in the range $1 \pm 0.02 \frac{\alpha^2 n}{d}$. Therefore, by \eqref{eq:null-bad-numerator},
\begin{align*}
    \frac{1}{n} \sum_{i \in B} \left(\frac{\langle X_i, \mathbf{S}\rangle - d}{\|\mathbf{S}\|_2}\right)^2 &\le \frac{1}{n^2 d} \cdot \left(1 + 0.02 \frac{\alpha^2 n}{d}\right) \cdot \left(\eps n^2 d + 0.002 \cdot \frac{\alpha^4}{\eps} \cdot n^3\right) \\
    &= \left(1 + 0.02 \frac{\alpha^2 n}{d}\right) \cdot \eps \cdot \left(1 + 0.002 \cdot \frac{\alpha^4}{\eps^2} \cdot \frac{n}{d}\right) \\
    &\le \eps \cdot \left(1 + 0.01 \cdot \frac{\alpha^4}{\eps^2} \cdot \frac{n}{d}\right) \\
    &= \eps + 0.01 \cdot \frac{\alpha^4}{\eps} \cdot \frac{n}{d},
\end{align*}
    where the penultimate line uses the fact that $\frac{\alpha^2 n}{d} < \min\left(1, 0.1 \frac{\alpha^4}{\eps^2} \cdot \frac{n}{d}\right)$ since we are assuming that $n \le \frac{d}{\alpha^2}$ and $\alpha \ge 10 \eps$.
\end{proof}

Next, we deal with the sum over good points.

\begin{lemma} \label{lem:null-good-variance}
    Assume the null hypothesis. Then, under \Cref{assumption:general}, with high probability,
\[\frac{1}{n} \sum_{i \in G} \left(\frac{\langle X_i, \mathbf{S} \rangle - d}{\|\mathbf{S}\|_2}\right)^2 \le 1-\eps + 0.01 \frac{\alpha^4}{\eps} \cdot \frac{n}{d}.\]
\end{lemma}

\begin{proof}
    Recall that $\mathbf{S} = \mathbf{R}+\mathbf{T}$, and suppose $\mathbf{R}, \mathbf{T}$ are fixed. Then, by~\Cref{prop:sufficient_statistic}, the posterior distribution of $\{X_i\}_{i \in G}$ conditioned on $\mathbf{R}$ and $\mathbf{T}$ is $\left\{\frac{\mathbf{T}}{(1-\eps) n} + Y_i-\bar{Y}\right\}$, where $\{Y_i\}_{i \in G}$ are i.i.d. $\cN(0, I)$, independent of $(\mathbf{R}, \mathbf{T})$, and $\bar{Y} = \frac{1}{(1-\eps) n} \sum_{i \in G} Y_i$.
    As a result, the posterior distribution of $\{\langle X_i, \mathbf{S} \rangle - d\}_{i \in G}$ is $\left\{\frac{\langle \mathbf{T}, \mathbf{S} \rangle}{(1-\eps) n} - d + \|\mathbf{S}\|_2 \cdot (z_i-\bar{z}) \right\},$ where $\{z_i\}_{i \in G}$ are i.i.d. univariate $\cN(0, 1)$, and $\bar{z}=\frac{1}{(1-\eps) n} \sum_{i \in G} z_i$.
    Hence, we can rewrite the desired sum over good points as
\[\frac{1}{n} \sum_{i \in G} \left(\frac{\frac{\langle \mathbf{T}, \mathbf{S} \rangle}{(1-\eps n)} - d}{\|\mathbf{S}\|_2} + (z_i-\bar{z})\right)^2.\]

    Now, note that $\langle \mathbf{T}, \mathbf{S} \rangle = \|\mathbf{T}\|_2^2 + \langle \mathbf{T}, \mathbf{R} \rangle$. Since $\mathbf{T} \sim \cN(0, (1-\eps) n I)$ is independent of $\mathbf{R}$, and $\|\mathbf{R}\|_2 \le O(\kappa \sqrt{\eps n d})$ by \Cref{prop:norm-sum} and the assumption that $\eps n \le d$, we have that $|\langle \mathbf{T}, \mathbf{R} \rangle| \le O(\kappa^2 \sqrt{\eps n^2 d})$ with high probability. In addition, $\|\mathbf{T}\|_2^2 = (1-\eps) n d \pm O(\kappa \cdot n \sqrt{d})$ with high probability, as $\mathbf{T}$ is the sum of the uncorrupted samples. In sum, $\langle \mathbf{T}, \mathbf{S} \rangle = (1-\eps) nd \pm O(\kappa^2 \cdot n \sqrt{d})$. Therefore, $\left|\frac{\langle \mathbf{T}, \mathbf{S} \rangle}{(1-\eps) n} - d\right| \le O(\kappa^2 \sqrt{d})$. Since $\|\mathbf{S}\|_2^2 = nd \left(1 \pm 0.01 \frac{\alpha^2 n}{d}\right) = \Theta(nd)$ as we are assuming $n \le \frac{d}{\alpha^2},$ 
this means $\left(\frac{\langle \mathbf{T}, \mathbf{S} \rangle}{(1-\eps) n} - d\right)/\|\mathbf{S}\|_2 = \pm O\left(\frac{\kappa^2}{\sqrt{n}}\right).$
    
Next, defining $\tilde{z}_i := z_i - \bar{z}$, we have $\sum_{i \in G} \tilde{z}_i^2 = \sum_{i \in G} z_i^2 - (1-\eps) n \cdot \bar{z}^2$. Clearly, $\bar{z} \sim \cN(0, \frac{1}{(1-\eps) n})$, so $|\bar{z}| \le \kappa/\sqrt{n}$ with high probability. Thus, $\sum_{i \in G} \tilde{z}_i^2 = \sum_{i \in G} z_i^2 - (1-\eps) n \cdot \bar{z}^2 = (1-\eps) n \pm O(\kappa \sqrt{n} + \kappa^2) = (1-\eps) n \pm O(\kappa^2 \sqrt{n}),$ by \Cref{prop:chi-square-concentration}. Hence, because the average of $\tilde{z}_i$ over $i \in G$ is $0$,
\begin{align*}
    \frac{1}{n} \sum_{i \in G} \left(\frac{\frac{\langle \mathbf{T}, \mathbf{S} \rangle}{(1-\eps n)} - d}{\|\mathbf{S}\|_2} + \tilde{z}_i\right)^2 
    &= \frac{1}{n} \cdot \sum_{i \in G} \tilde{z}_i^2 + (1-\eps) \cdot \left(\frac{\frac{\langle \mathbf{T}, \mathbf{S} \rangle}{(1-\eps n)} - d}{\|\mathbf{S}\|_2}\right)^2 \\
    &= \frac{1}{n} \cdot \sum_{i \in G} \tilde{z}_i^2 + (1-\eps) \cdot O\left(\frac{\kappa^2}{\sqrt{n}}\right)^2\\
    &= (1-\eps) \pm O\left(\frac{\kappa^4}{\sqrt{n}}\right).
\end{align*}
    As long as $n \ge \kappa^5 \cdot \left(\frac{d^{2/3} \eps^{2/3}}{\alpha^{8/3}}\right)$
the error term is at most $0.01 \frac{\alpha^4}{\eps} \cdot \frac{n}{d}$, which completes the proof.
\end{proof}

By combining Lemmas \ref{lem:null-bad-bound} and \ref{lem:null-good-variance}, we have the following.

\begin{lemma} \label{lem:null-variance}
    Assume the null hypothesis.
Then, under \Cref{assumption:general},
    $\frac{1}{n} \sum_{i=1}^n \left(\frac{\langle X_i, \mathbf{S} \rangle - d}{\|\mathbf{S}\|_2}\right)^2 \le 1 + 0.02 \cdot \frac{\alpha^4 n}{\eps d},$ with high probability.
\end{lemma}

\subsection{The Alternative Case: Variance}
\label{sec:oblivious-tester-alt-variance}
Let $\mu = \alpha \cdot v$, where $v$ is a unit vector.
Recall that $\mathbf{Q} = (1-\eps) n \cdot \alpha v$, $\mathbf{R} = \sum_{i \in B} X_i$, and $\mathbf{T} = \sum_{i \in G} (X_i - \alpha v) = (\sum_{i \in G} X_i) - \mathbf{Q}$. Let $\mathbf{S} = \mathbf{Q}+\mathbf{R}+\mathbf{T} = \sum_{i \in [n]} X_i$. We wish to bound
\[\frac{1}{n} \sum_{i \in [n]} \left(\frac{\langle X_i, \mathbf{S} \rangle - d}{\|\mathbf{S}\|_2}\right)^2 = \frac{1}{n} \sum_{i \in [n]} \left(\frac{\langle X_i, \mathbf{Q}+\mathbf{R} \rangle + \langle X_i, \mathbf{T} \rangle - d}{\|\mathbf{S}\|_2}\right)^2.\]

We can again split $[n]$ into bad and good points. For the bad points $B$, our goal is to bound
\[\frac{1}{n} \cdot \sum_{i \in B} \left(\frac{\langle X_i, \mathbf{S}\rangle - d}{\|\mathbf{S}\|_2}\right)^2 = \frac{1}{n \cdot \|\mathbf{S}\|_2^2} \cdot \sum_{i \in B} \left[\langle X_i, \mathbf{T} \rangle^2 + 2 \langle X_i, \mathbf{T} \rangle \cdot (\langle X_i, \mathbf{Q}+\mathbf{R}\rangle - d) + (\langle X_i, \mathbf{Q}+\mathbf{R}\rangle - d)^2\right].\]

Before doing so, we will consider the relationship between the values $\mathbf{R}, \mathbf{Q}, \mathbf{T}$. Note that $\mathbf{T}$ is independent of both $\mathbf{R}$ and $\mathbf{Q}$, whereas $\mathbf{R}$ may depend on $\mathbf{Q}$.

\begin{proposition} \label{prop:q-r-basic}
    Suppose that $X_1, \dots, X_n$ satisfy \Cref{assumption}, that $\eps n \le d$, and \autoref{alg:oblivious-tester} does not reject in Line 3. Then, $\|\mathbf{R}\|_2^2 = \eps n d \pm O(\kappa \cdot (\eps n)^{3/2} \sqrt{d})$, and $\|\mathbf{Q}+\mathbf{R}\|_2^2 = \eps n d \pm 0.01 \alpha^2 n^2 \pm \kappa^2 \cdot O\left(n \sqrt{d} + \alpha n^{3/2}\right)$.
\end{proposition}

\begin{proof}
    By \Cref{prop:norm-sum}, we know that $\|\mathbf{R}\|_2^2 = \|\sum_{i \in B} X_i\|_2^2 = \eps n d \pm O(\kappa \cdot (\eps n)^{3/2} \sqrt{d})$. Next, because \autoref{alg:oblivious-tester} did not reject in Line 3, we have that $\|\mathbf{Q}+\mathbf{R}+\mathbf{T}\|_2^2 = \|\sum_{i \in [n]} X_i\|_2^2 = nd \pm 0.01 \alpha^2 n^2$. However, we can write $\|\mathbf{Q}+\mathbf{R}+\mathbf{T}\|_2^2 = \|\mathbf{Q}+\mathbf{R}\|_2^2 + \|\mathbf{T}\|_2^2 + 2 \langle \mathbf{Q}+\mathbf{R}, \mathbf{T} \rangle.$
    Let $A = \|\mathbf{Q}+\mathbf{R}\|_2.$
    Then, with high probability, $\|\mathbf{T}\|_2^2 = (1-\eps) n d \pm O(\kappa n \sqrt{d})$ and $\langle \mathbf{Q}+\mathbf{R}, \mathbf{T} \rangle = \pm O(\kappa A \sqrt{n})$, since $\mathbf{T} \sim \cN(0, (1-\eps) n I)$ is independent of $\mathbf{Q}+\mathbf{R}$.
    This means $nd \pm 0.01 \alpha^2 n^2 = A^2 + (1-\eps) nd \pm O(\kappa n \sqrt{d}) \pm O(\kappa \cdot A \sqrt{n}),$ so $A^2 = \eps n d \pm 0.01 \alpha^2 n^2 \pm O(\kappa n \sqrt{d}) \pm O(\kappa \cdot A \sqrt{n})$.
    Finally, we know that $A \le \|\mathbf{Q}\|_2+\|\mathbf{R}\|_2 \le \alpha n + \sqrt{\kappa \cdot \eps n d}$, which means 
\[\|\mathbf{Q}+\mathbf{R}\|_2^2 = A^2 = \eps n d \pm 0.01 \alpha^2 n^2 \pm O(\kappa^2) \cdot (n \sqrt{d} + \alpha n^{3/2}). \qedhere\]
\end{proof}

Hence, we have the following corollary.

\begin{proposition} \label{prop:q-r-negative-correlation}
    Suppose that $X_1, \dots, X_n$ satisfy \Cref{assumption}, $\eps \le 0.5$, , and \autoref{alg:oblivious-tester} does not reject in Line 3. Then, if $n \ge \kappa^5 \cdot \left(\frac{\sqrt{d}}{\alpha^2}+\frac{d \eps^3}{\alpha^4}\right)$, $\left|\|\mathbf{R}\|_2^2 - \|\mathbf{Q}+\mathbf{R}\|_2^2\right| \le 0.1 \cdot \|\mathbf{Q}\|_2^2$ with high probability.
\end{proposition}

\begin{proof}
    As a direct consequence of \Cref{prop:q-r-basic},
\begin{equation} \label{eq:R-Q+R}
    \left|\|\mathbf{R}\|_2^2 - \|\mathbf{Q}+\mathbf{R}\|_2^2\right| \le 0.01 \alpha^2 n^2 +  \kappa^2 \cdot O\left((\eps n)^{3/2} \sqrt{d} + n \sqrt{d} + \alpha n^{3/2}\right).
\end{equation}
    Assuming that $n \ge \kappa^5 \cdot \left(\frac{\sqrt{d}}{\alpha^2} + \frac{d \eps^3}{\alpha^4}\right)$, each error term is at most $0.01 \alpha^2 n^2$, so \eqref{eq:R-Q+R} is at most $0.05 \alpha^2 n^2$. Since $\|\mathbf{Q}\|_2 = \alpha (1-\eps) n \ge 0.5 \alpha n,$ this means that $\left|\|\mathbf{R}\|_2^2 - \|\mathbf{Q}+\mathbf{R}\|_2^2\right| \le 0.1 \cdot \|\mathbf{Q}\|_2^2$.
\end{proof}

We now turn to bounding the sum for the bad points $B$.

\begin{lemma} \label{lem:alternative-bad-variance}
    Suppose that \autoref{alg:oblivious-tester} does not reject in Line 3, and $\eps \le 0.1$. Then, with high probability under the alternative hypothesis and \Cref{assumption:general},
\[\sum_{i \in B} \left(\langle X_i, \mathbf{S}\rangle - d\right)^2 \ge \eps n^2 d \left(1 + \frac{0.05 \alpha^4 n}{\eps^2 d}\right).\]
\end{lemma}

\begin{proof}
    We can rewrite the left-hand side of the above expression as 
\[\sum_{i \in B} \left((\langle X_i, \mathbf{Q}+\mathbf{R} \rangle - d)^2 + \langle X_i, \mathbf{T} \rangle^2 + 2 (\langle X_i, \mathbf{Q}+\mathbf{R} \rangle - d) \cdot \langle X_i, \mathbf{T} \rangle\right)\]

    First, we consider $\sum_{i \in B} (\langle X_i, \mathbf{Q}+\mathbf{R} \rangle - d)^2$. 
Since $\sum_{i \in B} X_i = \mathbf{R}$, by Jensen's inequality this is at least $\eps n \cdot \left(\langle \frac{\mathbf{R}}{\eps n}, \mathbf{Q}+\mathbf{R}\rangle - d\right)^2 = \frac{1}{\eps n} \cdot (\langle \mathbf{R}, \mathbf{Q}+\mathbf{R} \rangle - \eps n d)^2$. However, we can write $\langle \mathbf{R}, \mathbf{Q}+\mathbf{R} \rangle = \frac{\|\mathbf{Q}+\mathbf{R}\|_2^2 + \|\mathbf{R}\|_2^2 - \|\mathbf{Q}\|_2^2}{2},$ and since $\|\mathbf{R}\|_2^2 \le \|\mathbf{Q}+\mathbf{R}\|_2^2 + 0.1 \|\mathbf{Q}\|_2^2$ by \Cref{prop:q-r-negative-correlation}, this means that $\langle \mathbf{R}, \mathbf{Q}+\mathbf{R} \rangle \le \|\mathbf{Q}+\mathbf{R}\|_2^2 - 0.45 \|\mathbf{Q}\|_2^2$. 
    By \Cref{prop:q-r-basic}, we have $\|\mathbf{Q}+\mathbf{R}\|_2^2 = \eps n d \pm 0.01 \alpha^2 n^2 \pm O(\kappa^2) \cdot \left(n \sqrt{d} + \alpha n^{3/2}\right)$.
Therefore, since $\|\mathbf{Q}\|_2^2 = \alpha^2 (1-\eps)^2 n^2 \ge 0.8 \alpha^2 n^2$ as $\eps \le 0.1$, this means that 
\[\langle \mathbf{R}, \mathbf{Q}+\mathbf{R} \rangle - \eps n d \le 0.01 \alpha^2 n^2 + O(\kappa^2) \cdot \left(n \sqrt{d} + \alpha n^{3/2}\right) - 0.36 \alpha^2 n^2 = -0.35 \alpha^2 n^2 + O(\kappa^2) \cdot \left(n \sqrt{d} + \alpha n^{3/2}\right).\]
    Therefore,
\begin{align}
    \sum_{i \in B} (\langle X_i, \mathbf{Q}+\mathbf{R} \rangle - d)^2 &\ge \frac{1}{\eps n} \cdot \left(0.35 \alpha^2 n^2 - O(\kappa^2) \cdot (n \sqrt{d} + \alpha n^{3/2})\right)^2 \nonumber \\
    &\ge 0.1 \frac{\alpha^4}{\eps} n^3 - O(\kappa^2) \cdot \left(\frac{\alpha^2}{\eps} \cdot n^2 \sqrt{d} + \frac{\alpha^3}{\eps} n^{5/2}\right) \nonumber \\
    &\ge 0.08 \cdot \frac{\alpha^4}{\eps} n^3. \label{eq:alt-var-bad-Q+R}
\end{align}
    Above, the second inequality follows because $(A-B)^2 \ge A^2-2AB$ for any real $A,B$, and the last inequality follows because the two error terms are each at most $0.01 \frac{\alpha^4}{\eps} n^3$ if $n \ge \kappa^5 \cdot \frac{\sqrt{d}}{\alpha^2}$.

    To bound $\sum_{i \in B} \langle X_i, \mathbf{T} \rangle^2,$ we can write this as $\mathbf{T}^\top \left(\sum_{i \in B} X_i X_i^\top\right) \mathbf{T} = (1-\eps) n \cdot Z^\top \left(\sum_{i \in B} X_i X_i^\top\right) Z$, where $Z \sim \cN(0, I)$ is independent of $\{X_i\}_{i \in B}$. We apply Hanson-Wright (\Cref{lem:hanson-wright}) along with \Cref{prop:frobenius_bound}, to say that with high probability, $\left|Z^\top \left(\sum_{i \in B} X_i X_i^\top\right) Z - \Tr(\sum_{i \in B} X_i X_i^\top)\right| \le O\left(\kappa \cdot \|\sum_{i \in B} X_i X_i^\top\|_F\right) \le O(\kappa^2 \cdot d \sqrt{\eps n}).$ In addition, $\Tr(\sum_{i \in B} X_i X_i^\top) = \sum_{i \in B} \|X_i\|_2^2 = \eps n d \pm \kappa \cdot \eps n \sqrt{d}$. Therefore, since $\eps n \le d$,
\begin{equation} \label{eq:alt-var-bad-t}
    \sum_{i \in B} \langle X_i, \mathbf{T} \rangle^2 = (1-\eps) n \cdot \left(\eps n d \pm O(\kappa^2 \cdot d \sqrt{\eps n} + \kappa \cdot \eps n \sqrt{d})\right) = \eps n^2 d  - \eps^2 n^2 d \pm O\left(\kappa^2 \cdot \eps^{1/2} n^{3/2} d\right).
\end{equation}

    To bound the final term $\sum_{i \in B} (\langle X_i, \mathbf{Q}+\mathbf{R} \rangle - d) \langle X_i, \mathbf{T} \rangle$, we first bound $\left\|\sum_{i \in B} (\langle X_i, \mathbf{Q}+\mathbf{R} \rangle - d) X_i\right\|_2.$ We can use \Cref{prop:bad-cross-terms} to obtain that $\left\| \sum_{i \in B}(\langle X_i, \mathbf{R} \rangle - d) X_i \right\|_2 \le O(\kappa^2 \eps n d) \le O(\kappa^2 \alpha n d),$ as we assumed that $\eps \le \alpha$. Next, to bound $\left\| \sum_{i \in B} \langle X_i, \mathbf{Q} \rangle X_i \right\|_2,$ we can write $ \sum_{i \in B} \langle X_i, \mathbf{Q} \rangle X_i = \left(\sum_{i \in B} X_i X_i^\top\right) \cdot \mathbf{Q},$ which has norm at most $\left\|\sum_{i \in B} X_i X_i^\top\right\|_{op} \cdot \|\mathbf{Q}\|_2 \le O(\kappa^2 \cdot d \cdot \alpha n)$, using \Cref{prop:operator_bound}.
   Therefore, since $\{X_i\}_{i \in B}$ and $\mathbf{Q}$ are independent of $\mathbf{T} \sim \cN(0, (1-\eps) n I)$, with high probability we have that
\begin{align}
    \left|\sum_{i \in B} (\langle X_i, \mathbf{Q}+\mathbf{R} \rangle - d) \cdot \langle X_i, \mathbf{T} \rangle\right| \nonumber
    &= \left|\left\langle \mathbf{T}, \sum_{i \in B} (\langle X_i, \mathbf{R} \rangle - d) X_i + \sum_{i \in B} \langle X_i, \mathbf{Q} \rangle X_i \right\rangle\right| \nonumber \\
    &\le O(\kappa \sqrt{n}) \cdot \left(\left\| \sum_{i \in B} (\langle X_i, \mathbf{R} \rangle - d) X_i \right\|_2 + \left\| \sum_{i \in B} \langle X_i, \mathbf{Q} \rangle X_i \right\|_2\right) \nonumber \\
    &\le O(\kappa^3 \cdot \alpha n^{3/2} d). \label{eq:alt-var-bad-rt}
\end{align}

In summary, by combining Equations \eqref{eq:alt-var-bad-Q+R}, \eqref{eq:alt-var-bad-t}, and \eqref{eq:alt-var-bad-rt}, we have
\begin{align*}
    \sum_{i \in B} (\langle X_i, \mathbf{S} \rangle - d)^2 \ge \frac{0.08 \alpha^4 n^3}{\eps} + \eps n^2 d - O(\kappa^3) \cdot \left(\eps^2 n^2 d + \eps^{1/2} n^{3/2} d + \alpha n^{3/2} d\right).
\end{align*}
    Now, assuming that $n \ge \kappa^5 \cdot \left(\frac{d \eps^3}{\alpha^4} + \frac{d^{2/3} \eps^{2/3}}{\alpha^{8/3}}\right) \ge \kappa^5 \cdot \left(\frac{d \eps^3}{\alpha^4} + \frac{d^{2/3} \eps}{\alpha^{8/3}} + \frac{d^{2/3} \eps^{2/3}}{\alpha^2}\right)$, each of the three error terms is at most $0.01 \frac{\alpha^4}{\eps} n^3$. So overall, 
\[\sum_{i \in B} (\langle X_i, \mathbf{S} \rangle - d)^2 \ge \frac{0.05 \alpha^4 n^3}{\eps} + \eps n^2 d = \eps n^2 d \left(1 + \frac{0.05 \alpha^4 n}{\eps^2 d}\right). \qedhere\]
\end{proof}

\begin{corollary} \label{cor:alternative-bad-variance}
    Suppose that \autoref{alg:oblivious-tester} does not reject in Line 3, and $\eps \le 0.1$. Then, under the alternative hypothesis and \Cref{assumption:general}, with high probability $\frac{1}{n} \cdot \sum_{i \in B} \left(\frac{\langle X_i, \mathbf{S} \rangle - d}{\|\mathbf{S}\|_2}\right)^2 \ge \eps + 0.04 \frac{\alpha^4}{\eps} \cdot \frac{n}{d}$.
\end{corollary}

\begin{proof}

    Since \autoref{alg:oblivious-tester} did not reject on Line 3, this means $\|\mathbf{S}\|_2^2 = nd \pm 0.01 \alpha^2 n^2 = nd \cdot \left(1 \pm 0.01 \frac{\alpha^2 n}{d}\right)$.
    So, $\frac{1}{n \cdot \|\mathbf{S}\|_2^2} = \frac{1}{n^2d} \cdot \left(1 \pm 0.02 \frac{\alpha^2 n}{d}\right),$ because we assumed $n \le \frac{d}{\alpha^2}$. Hence, by \Cref{lem:alternative-bad-variance}, we have
\begin{align*}
    \frac{1}{n} \cdot \sum_{i \in B} \left(\frac{\langle X_i, \mathbf{S} \rangle - d}{\|\mathbf{S}\|_2}\right)^2 &\ge \eps \cdot \left(1 + 0.05 \frac{\alpha^4}{\eps^2} \cdot \frac{n}{d}\right) \cdot \left(1 - 0.02 \alpha^2 \cdot \frac{n}{d}\right) \\
    &\ge \eps \cdot \left(1 + 0.04 \frac{\alpha^4}{\eps^2} \cdot \frac{n}{d}\right) \\
    &= \eps + 0.04 \frac{\alpha^4}{\eps} \cdot \frac{n}{d}. \qedhere
\end{align*}
\end{proof}
Finally, we bound the good samples.

\begin{lemma} \label{lem:alternative-good-variance}
    Assume that $n \ge \kappa^5 \cdot \left(\frac{d^{2/3}\eps^{2/3}}{\alpha^{8/3}}\right).$ Then, with high probability
\[\frac{1}{n} \sum_{i \in G} \left(\frac{\langle X_i, \mathbf{S} \rangle - d}{\|\mathbf{S}\|_2}\right)^2 \ge 1-\eps - \frac{0.01 \alpha^4 n}{\eps d}.\]
\end{lemma}

\begin{proof}
    Let's fix the vectors $\mathbf{Q}, \mathbf{R}, \mathbf{T}$, and consider the posterior distribution of the good samples $\{X_i\}_{i \in G}$. By~\Cref{prop:sufficient_statistic}, we can write $X_i = \frac{\mathbf{Q}+\mathbf{T}}{(1-\eps) n} + Y_i-\bar{Y},$ where $\{Y_i\}_{i \in G}$ are distributed as i.i.d. $\cN(0, I)$ and $\bar{Y} = \frac{1}{(1-\eps) n} \sum_{i \in G} Y_i$. Hence, $\{\langle X_i, \mathbf{S}\rangle\}_{i \in G}$ is distributed as $\frac{\langle \mathbf{Q}+\mathbf{T}, \mathbf{S}\rangle}{(1-\eps) n} + \|\mathbf{S}\|_2 \cdot z_i-\bar{z},$ where $z_i$ are distributed as i.i.d. $\cN(0, 1)$ and $\bar{z}=\frac{1}{(1-\eps) n} \sum_{i \in G} z_i$. 
Hence, defining $\tilde{z}_i = z_i-\bar{z}$, since $\tilde{z}_i$ have mean $0$, we can rewrite our expression as
\[\frac{1}{n} \sum_{i \in G} \left(\frac{\frac{\langle \mathbf{Q}+\mathbf{T}, \mathbf{S} \rangle}{(1-\eps) n} - d}{\|\mathbf{S}\|_2} + \tilde{z}_i\right)^2 \ge \frac{1}{n} \sum_{i \in G} \left(\tilde{z}_i\right)^2 \ge (1-\eps) - \frac{\kappa^2}{\sqrt{n}}.\]
    The final inequality above combines the facts that $\sum_{i \in G} \tilde{z}_i^2 = \sum_{i \in G} z_i^2 - (1-\eps) n \bar{z}^2,$ that $\sum_{i \in G} z_i^2 = (1-\eps) n \pm \kappa \sqrt{n}$ by \Cref{prop:chi-square-concentration}, and that $|\bar{z}| \le \kappa/\sqrt{n}.$
Finally, because we are assuming $n \ge \kappa^5 \cdot \left(\frac{d^{2/3}\eps^{2/3}}{\alpha^{8/3}}\right)$, we have that $\frac{\kappa^2}{\sqrt{n}} \le 0.01 \frac{\alpha^4 n}{\eps d}$. This completes the proof.
\end{proof}

By combining \Cref{cor:alternative-bad-variance} and \Cref{lem:alternative-good-variance}, the following lemma is immediate.

\begin{lemma} \label{lem:alt-variance}
    Suppose that \autoref{alg:oblivious-tester} does not reject in Line 3, and $\eps \le 0.1$. Then, under the alternative hypothesis and \Cref{assumption:general}, with high probability 
\[\frac{1}{n} \sum_{i=1}^n \left(\frac{\langle X_i, \mathbf{S} \rangle - d}{\|\mathbf{S}\|_2}\right)^2 \ge 1 + 0.03 \cdot \frac{\alpha^4 n}{\eps d}.\]
\end{lemma}

As a direct consequence of Lemmas \ref{lem:null-mean}, \ref{lem:null-variance}, and \ref{lem:alt-variance}, \Cref{lem:mean-var-main} is immediate.

\subsection{Proof of \Cref{lem:mean-var-main-2}} \label{subsec:last-case-oblivious-lb}

In this section, we finish the proof of \Cref{thm:oblivious-ub}, by proving \Cref{lem:mean-var-main-2}.

It suffices to prove the following lemma.

\begin{lemma} \label{lem:alt-mean}
    Assume the alternative hypothesis, and that $n \ge \kappa^5 \cdot \left(\frac{\sqrt{d}}{\alpha^2} + \frac{d \eps}{\alpha^2}\right)$ and that $\eps \le 0.1$ and $\alpha \ge \kappa^5 \cdot \eps$. Then, under \Cref{assumption}, with high probability $\left\|\sum_{i \in [n]} X_i\right\|_2^2 \ge nd + 0.1 \alpha^2 n^2$.
\end{lemma}

\begin{proof}
    As usual, we write $\sum_{i \in [n]} X_i = \mathbf{Q}+\mathbf{R}+\mathbf{T}$, so $\big\|\sum_{i \in [n]} X_i\big\|_2^2 = \left\|\mathbf{Q}+\mathbf{R}+\mathbf{T}\right\|_2^2 = \|\mathbf{Q}+\mathbf{R}\|_2^2 + \|\mathbf{T}\|_2^2 + 2 \langle \mathbf{Q}+\mathbf{R}, \mathbf{T} \rangle$.

    Let $A = \|\mathbf{Q}+\mathbf{R}\|_2$. Note that $A \ge \|\mathbf{Q}\|_2-\|\mathbf{R}\|_2 = 0.9 \alpha n - \|\mathbf{R}\|_2$, assuming $\eps \le 0.1$. In addition, by \Cref{prop:norm-sum}, we have that $\|\mathbf{R}\|_2^2 = \|\sum_{i \in B} X_i\|_2^2 \le \eps n d + O(\kappa) \cdot ((\eps n)^{3/2} \sqrt{d} +(\eps n)^2)$. So, $\|\mathbf{R}\|_2 \le O(\sqrt{\eps n d} + \kappa \cdot \eps n)$. Assuming that $n \ge \kappa^5 \cdot \frac{d \eps}{\alpha^2}$ and $\alpha \ge \kappa^5 \cdot \eps$, both $O(\sqrt{\eps n d})$ and $O(\kappa \cdot \eps n)$ are at most $0.1 \alpha n$. Thus, $A \ge 0.7 \alpha n$.

    Since $\mathbf{T} \sim \cN(0, (1-\eps) n I)$ is independent of $\mathbf{Q}, \mathbf{R},$ this means that $\|\mathbf{T}\|_2^2 \ge (1-\eps) n d - \kappa n \sqrt{d}$ and $|\langle \mathbf{Q}+\mathbf{R}, \mathbf{T} \rangle| \le (\kappa \sqrt{n}) \cdot A$ with high probability. Thus, 
\[\|\mathbf{Q}+\mathbf{R}+\mathbf{T}\|_2^2 \ge A^2 + (1-\eps) n d - \kappa n \sqrt{d} - 2 \kappa \sqrt{n} \cdot A \ge (1-\eps) n d + (A-\kappa \sqrt{n})^2 - \kappa^2 n \sqrt{d}.\]
    Since $n \ge \kappa^5 \cdot \frac{\sqrt{d}}{\alpha^2} \ge \frac{\kappa^5}{\alpha^2},$ this means that $\kappa \sqrt{n} \le 0.2 \alpha n$, so $A-\kappa \sqrt{n} \ge 0.5 \alpha n.$ Moreover, $\kappa^2 n \sqrt{d} \le 0.05 \alpha^2 n^2$. Thus,
\[\|\mathbf{Q}+\mathbf{R}+\mathbf{T}\|_2^2 \ge (1-\eps) n d + (0.5 \alpha n)^2 - 0.05 \alpha^2 n^2 = (1-\eps) n d + 0.2 \alpha^2 n^2.\]
    Assuming that $n \ge \kappa^5 \cdot \frac{d \eps}{\alpha^2}$, $\eps n d \le 0.1 \alpha^2 n^2,$ which means this is at least $nd + 0.1 \alpha^2 n^2$.
\end{proof}

By combining Lemmas \ref{lem:null-mean} and \ref{lem:alt-mean}, Lemma \ref{lem:mean-var-main-2} is immediate (since $\frac{d \eps^3}{\alpha^4} < \frac{d \eps}{\alpha^2}$). Note that we never assumed $\eps n \le d$ in either of these lemmas. 
\section{Lower bound in the Huber model} \label{sec:huber-lb}

In this section, we prove that under the Huber model, one needs $n = \Omega(d \eps^3/\alpha^4)$ samples to solve robust mean testing. Our lower bound even holds in the restricted setting where under the null hypothesis, the distribution must be uncorrupted.

\subsection{Main Lower Bound}

We are now ready to prove our main lower bound.

\begin{theorem} \label{thm:huber-lb}
    Let $\cD_0$ represent the distribution of $X_1, \dots, X_n \overset{i.i.d.}{\sim} \mathcal{N}(0, I)$, and let $\cD_1$ represent the distribution of $(X_1, \dots, X_n)$ where we choose a random vector $v \sim \mathcal{N}(0, \frac{1}{d} \cdot I)$ and conditional on $v$, each $X_i$ is drawn i.i.d. from the mixture $(1-\eps) \cdot \mathcal{N}(\alpha \cdot v, I) + \eps \cdot \mathcal{N}(-\frac{1-\eps}{\eps} \cdot \alpha \cdot v, I)$.
    Then, there exists a small absolute constant $c > 0$ such that if $n = c \cdot \frac{d \eps^3}{\alpha^4}$, $\alpha \ge \eps$, and $c \cdot \frac{d \eps^3}{\alpha^4} \ge \frac{\sqrt{d}}{\alpha^2},$ then $\dtv(\cD_0, \cD_1) \le 0.1$.
\end{theorem}

Since the total variation distance is at most $0.1$, no algorithm can successfully distinguish between $\cD_0$ and $\cD_1$ with probability more than $0.55$. Moreover, $\|v\|_2 \le 1+o(1)$, and therefore $\|\alpha v\|_2 \le \alpha(1+o(1))$, with very high probability. Hence, this proves the desired lower bound when $c \cdot \frac{d \eps^3}{\alpha^4} \ge \frac{\sqrt{d}}{\alpha^2}.$ Alternatively, if $c \cdot \frac{d \eps^3}{\alpha^4} < \frac{\sqrt{d}}{\alpha^2}$, the lower bound is immediate from the non-robust lower bound~\cite{DiakonikolasKS17}. When $\alpha \ge \eps$, it is well-known that this problem is impossible, since the null and alternative distributions have total variation distance $\leq \cor$.

We will bound the $\dtv(\cD_0, \cD_1)$ through $\chi^2$ divergence. As $\dchi(\cD_1||\cD_0)$ is actually too large and thus does not suffice, we instead bound $\dchi(\cD_1'||\cD_0)$ for some $\cD_1'$ that is close in total variation distance to $\cD_1$. 
    
For a sample $X = (X_1, \dots, X_n) \sim \cD_1,$ we will let a set $S \subset [n]$ correspond to $X$ where $i \in S$ iff $X_i$ was drawn from the mixture component $\mathcal{N}(\alpha \cdot v, I)$. Note that $S$ is not determined by $X$. We will choose $\cD_1'$ to be $\cD_1$ restricted to having $S$ with size $(1-\eps) n \pm K \sqrt{\eps n}$ for some large constant $K$. Call such sets $S$ good, and let $\mathcal{S}$ be the set of all good sets. By standard properties of Binomial distributions, if $K \ge 100$, with probability at least $1 - 10^{-4}$, a random subset $S$ obtained by including each element $i \in [n]$ with probability $1-\eps$ is good. Hence, $\dtv(\cD_1, \cD_1') \le 2 \cdot 10^{-4}$.  Thus, it now suffices to upper bound $\dchi(\cD_1'||\cD_0)$ (which then upper bounds $\dtv(\cD_1',\cD_0)$ by \Cref{prop:chi-tv}).

It will be convenient to use the following notation throughout this section: let $Z$ be the probability that a random subset of $[n]$ obtained by including each element $i \in [n]$ with probability $1 - \eps$ is good.  We begin by computing the likelihood ratio.  
\begin{claim}\label{claim:huber-lb-likelihood-ratio}
Let $X = (X_1, \dots , X_n)$ be a set of samples.  Let $p_{\cD_0}(X), p_{\cD_1'}(X)$ be the PDFs of seeing that sample from $\cD_0$ and $\cD_1'$ respectively.  Then
\[
\frac{p_{\cD_1'}(X)}{p_{\cD_0}(X)} = \frac{1}{Z} \sum_{S \in \mathcal{S}} (1 - \eps)^{|S|} \eps^{n - |S|} \left( \frac{d + t_S}{d}\right)^{-d/2}\exp\left( \frac{\alpha_S(X)}{2(t_S + d)} \right)
\]
where we define for subsets $S \subset [n]$,
\[
\begin{split}
X_S & := \frac{\alpha}{\eps} \cdot \left(\eps \cdot \sum_{i \in S} X_i - (1-\eps) \cdot \sum_{i \not\in S} X_i\right) \\ 
\alpha_S(X) &:= \|X_S\|^2 \\
t_S &:= \frac{\alpha^2}{\eps^2} \cdot \left(\eps^2 \cdot |S| + (1-\eps)^2 \cdot (n-|S|)\right)
\end{split}
\]
\end{claim}
\begin{proof}
  For any sample $X = (X_1, \dots, X_n)$, the PDF of seeing that sample from $\cD_0$ is
\begin{equation} \label{eq:D0-pdf}
    p_{\cD_0}(X) = \prod_{i=1}^{n} e^{-\|X_i\|^2/2}
\end{equation}
The probability of seeing that sample from $\cD_1$ is 
\begin{align}
     p_{\cD_1}(X)
    &= \E_{v}\prod_{i=1}^{n} \left((1-\eps) \cdot e^{-\|X_i-\alpha v\|^2/2} + \eps \cdot e^{-\|X_i+(1-\eps) \alpha/\eps v\|^2/2}\right) \nonumber \\
    &= \sum_{S \subset [n]} \E_{v}\left((1-\eps)^{|S|} \eps^{(n-|S|)} \cdot \prod_{i \in S} e^{-\|X_i-\alpha v\|^2/2} \cdot \prod_{i \not\in S}e^{-\|X_i+(1-\eps) \alpha/\eps v\|^2/2}\right). \label{eq:D1-pdf}
\end{align}
    By restricting ourselves to good sets $S \in \mathcal{S}$, the probability of seeing $X$ drawn from $\cD_1'$ is 
\begin{equation} \label{eq:D1'-pdf}
    p_{\cD_1'}(X) = \frac{1}{Z} \sum_{S \in \mathcal{S}} \E_{v}\left((1-\eps)^{|S|} \eps^{(n-|S|)} \cdot \prod_{i \in S} e^{-\|X_i-\alpha v\|^2/2} \cdot \prod_{i \not\in S}e^{-\|X_i+(1-\eps) \alpha/\eps v\|^2/2}\right),
\end{equation}
where $Z$ is the probability of a random set $S$ being good if each $i \in [n]$ is included in $S$ independently with probability $1-\eps$.

 From \eqref{eq:D0-pdf} and \eqref{eq:D1'-pdf}, it is simple to compute the ratio 
\[\frac{p_{\cD_1'} (X)}{p_{\cD_0} (X)} = \frac{1}{Z} \cdot \sum_{S \in \mathcal{S}}\left((1-\eps)^{|S|} \eps^{(n-|S|)} \cdot \underbrace{\E_{v}\left[\prod_{i \in S} e^{- \alpha \langle X_i, v \rangle - \alpha^2 \|v\|^2/2} \cdot \prod_{i \not\in S}e^{(1-\eps) \alpha/\eps \cdot \langle X_i, v \rangle -(1-\eps)^2 \alpha^2/\eps^2 \cdot \|v\|^2/2}\right]}_{A_S(X)}\right).\]
We use $A_S(X)$ as a shorthand in simplifying the expression above.  Now we can explicitly compute $A_S(X)$. With $X_S, t_S$ as defined above, we can write 
\[A_S(X) = \E_v\left[e^{-\langle X_S, v\rangle- t_S/2 \cdot \|v\|^2}\right].\]
    Since $v \sim \mathcal{N}(0, \frac{1}{d} \cdot I) = \frac{1}{\sqrt{d}} \cdot \mathcal{N}(0, I)$, we can use the rotational symmetry of $v$ and \Cref{prop:mgf_squared_gaussian} to rewrite
\begin{align*}
    A_S(X) &= \E_{x \sim \mathcal{N}(0, 1)}\left[e^{-\sqrt{\alpha_S(X)/d} \cdot x - (t_S/2d) \cdot x^2} \right] \cdot \left(\E_{x \sim \mathcal{N}(0, 1)}\left[e^{-(t_S/2d) \cdot x^2} \right]\right)^{d-1} \\
    &= \frac{\exp\left(\frac{\alpha_S(X)/d}{2+2t_S/d}\right)}{\sqrt{1+t_S/d}} \cdot \left(\frac{1}{\sqrt{1+t_S/d}}\right)^{d-1} \\
    &= \exp\left(\frac{\alpha_S(X)}{2(t_S+d)}\right) \cdot \left(\frac{d+t_S}{d}\right)^{-d/2}.
\end{align*}
    
\end{proof}

Using expression for the likelihood ratio in Claim~\ref{claim:huber-lb-likelihood-ratio}, we can explicitly compute the $\chi^2$ divergence $\dchi(\cD_1' || \cD_0)$.
\begin{lemma}\label{lem:huber-lb-explicit-chi2}
We have
\[
\dchi(\cD_1' || \cD_0) = \frac{1}{Z^2} \cdot \sum_{S, T \subset \mathcal{S}} (1-\eps)^{|S|+|T|} \eps^{(n-|S|)+(n-|T|)} \cdot \left(1 - \left(\frac{t_{S, T}}{d}\right)^2\right)^{-d/2}
\]
where $t_{S, T} = \frac{\alpha^2}{\eps^2} \cdot \left(\eps^2 |S \cap T| - \eps (1-\eps) |S \triangle T| + (1-\eps)^2 |(S \cup T)^c|\right)$ and $\triangle$ denotes symmetric difference.
\end{lemma}
\begin{proof}
Using Claim~\ref{claim:huber-lb-likelihood-ratio}, we can write
\begin{equation} \label{eq:chi_square_goal_2}
\begin{split}
\dchi(\cD_1' || \cD_0) = \frac{1}{Z^2} \cdot \sum_{S, T \subset \mathcal{S}} (1-\eps)^{|S|+|T|} \eps^{(n-|S|)+(n-|T|)} \cdot \left(\frac{d+t_S}{d}\right)^{-d/2}\left(\frac{d+t_T}{d}\right)^{-d/2} \\
\cdot \underbrace{\E_{X \sim \cD_0} \left[\exp\left(\frac{\alpha_S(X)}{2(t_S+d)} + \frac{\alpha_T(X)}{2(t_T+d)}\right)\right]}_{B_{S, T}}.
\end{split}
\end{equation}
where $\alpha_S(X) = \|X_S\|^2$, $\alpha_T(X) = \|X_T\|^2$ and
\[
\begin{split}
X_S &= \frac{\alpha}{\eps}\left( \eps \cdot \sum_{i \in S} X_i - (1-\eps) \cdot \sum_{i \not\in S} X_i \right) \\
X_T &= \frac{\alpha}{\eps}\left( \eps \cdot \sum_{i \in T} X_i - (1-\eps) \cdot \sum_{i \not\in T} X_i\right)
\end{split}
\]
are as defined in Claim~\ref{claim:huber-lb-likelihood-ratio}.  Now we explicitly compute the expression above labelled $B_{S,T}$.    In each coordinate $j \in [d]$, $((X_S)_j, (X_T)_j)$ forms a bivariate Gaussian, and $((X_S)_j, (X_T)_j)$ over all $j \in [d]$ are independent and identically distributed. Through direct computation, we get that $((X_S)_j, (X_T)_j) \sim \mathcal{N}(\textbf{0}, \Sigma)$, where 
\[
\Sigma = \left(\begin{matrix} t_S & t_{S, T} \\ t_{S, T} & t_T \end{matrix}\right)
\]
and 
\[
\begin{split}
t_S &= \frac{\alpha^2}{\eps^2} \cdot \left(\eps^2 \cdot |S| + (1-\eps)^2 \cdot (n-|S|)\right) \\
t_T &= \frac{\alpha^2}{\eps^2} \cdot \left(\eps^2 \cdot |T| + (1-\eps)^2 \cdot (n-|T|)\right) \\
t_{S, T} &= \frac{\alpha^2}{\eps^2} \cdot \left(\eps^2 |S \cap T| - \eps (1-\eps) |S \triangle T| + (1-\eps)^2 |(S \cup T)^c|\right) \,.
\end{split}
\]
Therefore, $\left(\frac{(X_S)_j}{\sqrt{2(t_S+d)}}, \frac{(X_T)_j}{\sqrt{2(t_T+d)}}\right) \sim \mathcal{N}(\textbf{0}, \Sigma'),$ where 
\[
\Sigma' = \left(\begin{matrix} \frac{t_S}{2(t_S+d)} & \frac{t_{S, T}}{2\sqrt{(t_S+d)(t_T+d)}} \\ \frac{t_{S, T}}{2\sqrt{(t_S+d)(t_T+d)}} & \frac{t_T}{2(t_T+d)} \end{matrix}\right) \,.
\]
By \Cref{cor:mgf_bivariate_gaussian}, this implies that
\[
\E\left[\exp\left(\frac{(X_S)_j^2}{2(t_S+d)}+\frac{(X_T)_j^2}{2(t_T+d)}\right)\right] = \frac{1}{\sqrt{\left(1-\frac{t_S}{t_S+d}\right) \cdot \left(1-\frac{t_T}{t_T+d}\right) - \frac{t_{S, T}^2}{(t_S+d)(t_T+d)}}} \,.
\]

Therefore, by multiplying this over all $j$ (since $((X_S)_j, (X_T)_j)$ are i.i.d. across all $j \in [d]$), we have that
\begin{align*}
    B_{S, T} &= \left(\left(1-\frac{t_S}{t_S+d}\right) \cdot \left(1-\frac{t_T}{t_T+d}\right) - \frac{t_{S, T}^2}{(t_S+d)(t_T+d)}\right)^{-d/2} \\
    &= \left(\frac{d^2-t_{S, T}^2}{(t_S+d)(t_T+d)}\right)^{-d/2}
\end{align*}
    So, \eqref{eq:chi_square_goal_2} can be rewritten as 
\[\dchi(\cD_1' || \cD_0) = \frac{1}{Z^2} \cdot \sum_{S, T \subset \mathcal{S}} (1-\eps)^{|S|+|T|} \eps^{(n-|S|)+(n-|T|)} \cdot \left(1 - \left(\frac{t_{S, T}}{d}\right)^2\right)^{-d/2}.\]

\end{proof}

Now we can complete the proof of Theorem~\ref{thm:huber-lb} by upper bounding the RHS of Lemma~\ref{lem:huber-lb-explicit-chi2}.

\begin{proof}[Proof of Theorem~\ref{thm:huber-lb}]

Recall that it suffices to prove that $\dchi(\cD_1'||\cD_0) = \E_{X \sim \cD_0} \left(\frac{p_{\cD_1'} (X)}{p_{\cD_0} (X)}\right)^2 \le 1.01$ as, by \Cref{prop:chi-tv}, this implies that the TV distance between $\cD_1'$ and $\cD_0$ is at most $0.1$.  By Lemma~\ref{lem:huber-lb-explicit-chi2} it now suffices to bound the expression
\[
\frac{1}{Z^2} \cdot \sum_{S, T \subset \mathcal{S}} (1-\eps)^{|S|+|T|} \eps^{(n-|S|)+(n-|T|)} \cdot \left(1 - \left(\frac{t_{S, T}}{d}\right)^2\right)^{-d/2} \,.
\]
We can think of $S, T$ as random subsets of $[n]$ where each element $i$ is chosen to be in $S$ (and likewise $T$) with probability $1-\eps$, and then conditioning on $S, T$ having size $(1 - \eps)n \pm K\sqrt{\eps n}$ for some sufficiently large constant $K$. In this case, if we use $S, T \sim \mathcal{S}$ to denote this distribution,  the above expression is equivalent to
\[\E_{S, T \sim \mathcal{S}} \left(1 - \left(\frac{\alpha^2}{\eps^2} \cdot \frac{\eps^2 |S \cap T| - \eps (1-\eps) |S \triangle T| + (1-\eps)^2 |(S \cup T)^c|}{d}\right)^2\right)^{-d/2}.\]
So, now we just need to show that if $n = c \cdot \frac{d \eps^3}{\alpha^4}$ for some small constant $c$, that the above expectation is at most $1.01$.

    Now, recall that we assumed $\frac{\sqrt{d}}{\alpha^2} \le c \cdot \frac{d \eps^3}{\alpha^4}.$ This means that $c \cdot \frac{\sqrt{d} \eps^3}{\alpha^2} \ge 1$, or equivalently $\frac{d \eps^3}{\alpha^4} \ge c^{-2} \eps^{-3}$. Hence, we may assume that $n \ge c^{-1} \cdot \eps^{-3}$. 
    
    If $|S| = a$ and $|T| = b$, and we let $Y := |(S \cup T)^c|$, 
    \[
    \begin{split}
    &\hspace{0.5cm} \eps^2 |S \cap T| - \eps (1-\eps) |S \triangle T| + (1-\eps)^2 |(S \cup T)^c| \\ &= (1-\eps)^2 \cdot Y - \eps (1-\eps) (n-a-Y+n-b-Y) + \eps^2 (a+b-n+Y) \\ &= Y - \eps (1-\eps) (2n-a-b) + \eps^2 (a+b-n) \\ 
    &= Y + \eps (a+b) - \eps (2-\eps) n \,.
    \end{split}
    \]
    Recall that we may always assume $a, b = (1-\eps) n \pm K \sqrt{\eps n}$. Also, note that $Y \sim \hgeom(n, n-a, n-b)$. Therefore, if we condition on fixed $a, b \in [(1-\eps) n - K \sqrt{\eps n}, (1-\eps) n + K \sqrt{\eps n}]$, we have that $\E[Y|a, b] = \frac{(n-a)(n-b)}{n} = \eps^2 n \pm 2 K \eps \sqrt{\eps n} \pm K^2 \eps$. By our assumption that $n \ge c^{-1} \cdot \eps^{-3}$ and choosing $c$ sufficiently small in terms of $K$, this can be bounded as $\eps^2 n \pm 3 K \eps \sqrt{\eps n}$. 
    
    Moreover, by \Cref{prop:hypergeometric-subgaussian}, since $n-a, n-b \le \eps n + K \sqrt{\eps n} \le 2 \eps n$, 
    \[
    \BP\left(\left\lvert Y - \BE[Y | a, b]\right\rvert \ge t \sqrt{\eps n}  | a, b\right) \le 2 e^{-2 t^2 (\eps n)/(n-a)} \le 2 e^{-t^2} \,. 
    \]
    Because $|\BE[Y|a, b] - \eps^2 n| \le 3 K \eps \sqrt{\eps n}$, this means $\BP(|Y - \eps^2 n| \ge (3K \eps +t) \sqrt{\eps n}) \le 2e^{-t^2}$.
    Hence, because $\eps (a+b) - \eps (2-\eps) n = -\eps^2 n^2 \pm 2K \sqrt{\eps n}$, this means $\BP(|Y+\eps(a+b)-\eps(2-\eps) n| \ge (5K+t) \sqrt{\eps n}) \le 2e^{-t^2}$.
    In addition, we know that $Y$ is bounded by $\min(n-a, n-b) \le 2 \eps n$, so overall $|Y+\eps(a+b)-\eps(2-\eps) n|$ is also bounded by $4 \eps n$ with probability 1.
    
    We can rewrite our goal as bounding
\[\E_{S, T \sim \mathcal{S}} \left(1 - \left(\frac{\alpha^2}{\eps^2} \cdot \frac{Y+\eps(a+b) - \eps(2-\eps) n}{d}\right)^2\right)^{-d/2}.\]

    Note that if $|x| \le 0.2$, then $1-x^2 \ge e^{-2x^2}$, so $(1-x^2)^{-d/2} \le e^{-2x^2 \cdot -d/2} = e^{d x^2}$.
    We know that $|Y + \eps (a+b)-\eps (2-\eps) n| \le 4 \eps n$ with probability $1$, so as long as $\frac{\alpha^2}{\eps^2} \cdot \frac{4 \eps n}{d} \le 0.4,$ which holds when $n \le 0.1 \cdot \frac{d \eps^3}{\alpha^4} \le 0.1 \cdot \frac{d \eps}{\alpha^2},$ we just need to bound
\begin{align} \label{eq:something}
    \E_{S, T \sim \mathcal{S}} \left[\exp\left(\frac{\alpha^4}{\eps^4} \cdot \frac{1}{d} \cdot (Y + \eps(a+b) - \eps (2-\eps) n)^2 \right)\right].
\end{align}
    Defining $C$ such that $Y + \eps(a+b) - \eps (2-\eps) n = C \sqrt{\eps n}$, then $\BP(|C| \ge 5K+t) \le 2e^{-t^2}$. So, \eqref{eq:something} equals
\[\E_{S, T \sim \mathcal{S}} \left[\exp\left(\frac{\alpha^4}{\eps^4} \cdot \frac{1}{d} \cdot C^2 \eps n \right)\right] = \E_{S, T \sim \mathcal{S}} \left[\exp\left(C^2 \cdot \frac{\alpha^4}{\eps^3} \cdot \frac{n}{d}\right)\right] = \E_{S, T \sim \mathcal{S}} \left[e^{C^2 \cdot c}\right],\]
    since $n \le c d \cdot \frac{\eps^3}{\alpha^4}$. By our bounds on $C$, if we assume $c$ is sufficiently small in terms of $K$, this is at most $1.01$, which means $\dchi(\cD_1'||\cD_0) \le 1.01.$ This concludes the proof, since \Cref{prop:chi-tv} implies $\dtv(\cD_1', \cD_0) \le 0.05$, and we already know that $\dtv(\cD_1', \cD_1) \le 2 \cdot 10^{-4}$.
\end{proof}
 
\section{Improved Lower Bound against Oblivious Adversaries} \label{sec:oblivious-lb}

In this section, we further improve our lower bound from \Cref{sec:huber-lb} against an oblivious adversary.

\subsection{Lower bound instance}

We first construct the distributions for the lower bound instance. Fix parameters $\eps < \alpha \le 1$ and dimension $d$, and consider drawing $n$ samples for some choice of $n$. We will also set an auxiliary parameter $\beta$, which will depend on $\eps, \alpha, d, n$.

The null distribution $\cD_0$ will simply be $n$ i.i.d. samples from $\cN(0, I)$. To generate the alternative distribution $\cD_1$, we perform the following steps: 
\begin{enumerate}
\item Select a subset $A \subset [n]$ of size $\eps n$ randomly. Let $A^c = [n] \backslash A$
\item Draw $\eps \cdot n$ points $\{X_i\}_{i \in A}$ i.i.d. from $\cN(0,I)$.  Set $\mathbf{R}_A := \Sum(A) = \sum_{i \in A} X_i$. 
\item Draw the vector $z \in \R^d$ from  $\cN\left(0, \frac{\alpha^2}{d} \cdot I\right).$
\item Define $\mu := -\beta \cdot \mathbf{R}_A - z$, and draw $(1-\eps) n$ points $\{X_i\}_{i \in A^c}$ from the distribution $\cN\left(\mu, I\right)$.
\end{enumerate}
For simplicity, we may write $X = (X_1, \dots, X_n)$, both in the null and alternative settings.

Note that with very high probability, $\|z\|_2 \le 2 \alpha$. We will also ensure that $\beta$ is chosen so that with very high probability, $\beta \cdot \|\mathbf{R}_A\|_2 \le 2 \alpha$. As a result, this alternative construction indeed has $\|\mu\|_2 \le O(\alpha)$.

In the rest of this section, we prove that it is statistically hard to distinguish between $\mcl{D}_0$ and $\mcl{D}_1$, for an appropriate choice of $\beta$.

\begin{theorem} \label{thm:oblivious-lb}
    Suppose that $n \le c \cdot \min\left(\frac{d^{2/3} \eps^{2/3}}{\alpha^{8/3}}, \frac{d \eps}{\alpha^2}\right)$ for some sufficiently small constant $c > 0$, and that $\eps \le \alpha \le 1$ and $n \ge \frac{\sqrt{d}}{\alpha^2} + \frac{d \eps^3}{\alpha^4}$. Then $\dtv(\mathcal{D}_0, \mathcal{D}_1) \le 0.1$.
\end{theorem}

This implies that no algorithm can successfully distinguish between $\mathcal{D}_0$ and $\mathcal{D}_1$ with probability more than $0.55$, which proves the desired lower bound when $c \cdot \min\left(\frac{d^{2/3} \eps^{2/3}}{\alpha^{8/3}}, \frac{d \eps}{\alpha^2}\right) \ge \frac{\sqrt{d}}{\alpha^2} + \frac{d \eps^3}{\alpha^4}.$ Alternatively, we may either use the non-robust lower bound~\cite{DiakonikolasKS17} or \Cref{thm:huber-lb}.
Finally, when $\alpha \ge \eps$, it is well-known that this problem is impossible, since the null and alternative distributions have total variation distance $\leq \cor$.

We will prove the lower bound via a chi-square computation. For this, we must compute likelihood ratios, which we will do in the next subsection.

\subsection{Likelihood Ratio Computation}

First, we will compute a formula for the likelihood ratio between $\cD_1$ and $\cD_0$, if we condition on the set $A \subset [n]$ in the alternative hypothesis.

\begin{definition}
Recall that $p_{\cD_0}(X), p_{\cD_1}(X)$ denote the joint PDF of the points $X = (X_1, \dots, X_n)$ drawn according to $\cD_0$ and $\cD_1$, respectively.
We also define $p_{A}(X)$ to denote the PDF of $X_1, \dots, X_n$ drawn according to $\cD_1$, conditioned on the first step selecting $A$.

In addition, we will define $\mathbf{R}_A := \Sum(A) = \sum_{i \in A} X_i$, and $\mathbf{T}_A(X) := \sum_{i \in A^c} X_i$. Usually, the choice of $X_1, \dots, X_n$ will be clear, in which case we will drop the argument $X$.
\end{definition}

\begin{lemma} \label{lem:likelihood-conditioned-subset}
    Conditioned on $A$, the likelihood ratio is
\[\frac{p_{A}(X)}{p_{\cD_0}(X)} = \left(1 + \frac{(1-\eps) n \cdot \alpha^2}{d}\right)^{-d/2} \cdot \exp\left(-\frac{(1-\eps) n \beta^2 d \cdot \|\mathbf{R}_A\|_2^2 + 2 \beta d \cdot \langle \mathbf{R}_A, \mathbf{T}_A \rangle - \alpha^2 \cdot \|\mathbf{T}_A\|_2^2}{2((1-\eps) \alpha^2 n+d)}\right).\]
\end{lemma}

\begin{proof}
    Suppose we additionally condition on the value $z \sim \cN(0, \frac{\alpha^2}{d} \cdot I)$. Then,
\begin{align*}
    \log \frac{p_A(X|z)}{p_{\cD_0}(X)}
    &= \sum_{i \in A^c} \left(-\frac{1}{2} \left\|X_i + \beta \cdot \mathbf{R}_A + z\right\|_2^2 + \frac{1}{2} \left\|X_i\right\|_2^2 \right) \\
    &= \sum_{i \in A^c} \left(-\langle X_i, \beta \cdot \mathbf{R}_A+z \rangle - \frac{1}{2} \left\|\beta \cdot \mathbf{R}_A + z\right\|_2^2\right) \\
    &= -\langle \mathbf{T}_A, \beta \cdot \mathbf{R}_A + z \rangle - \frac{(1-\eps) n}{2} \cdot \|\beta \cdot \mathbf{R}_A+z\|_2^2 \\
    &= -\beta \cdot \langle \mathbf{T}_A, \mathbf{R}_A \rangle - \frac{(1-\eps) n}{2} \cdot \beta^2 \cdot \|\mathbf{R}_A\|_2^2 - \langle \mathbf{T}_A + (1-\eps) n \cdot \beta \cdot \mathbf{R}_A, z \rangle - \frac{(1-\eps) n}{2} \cdot \|z\|_2^2.
\end{align*}
    So,
\[\frac{p_A(X|z)}{p_{\cD_0}(X)} = \exp\left(-\beta \cdot \langle \mathbf{T}_A, \mathbf{R}_A \rangle - \frac{(1-\eps) n}{2} \cdot \beta^2 \cdot \|\mathbf{R}_A\|_2^2 - \langle \mathbf{T}_A + (1-\eps) n \cdot \beta \cdot \mathbf{R}_A, z \rangle - \frac{(1-\eps) n}{2} \cdot \|z\|_2^2\right).\]

    Next, we remove the conditioning on $z$. Indeed, by using the above equation followed by \Cref{prop:shifted-Gaussian-integral}, we have
\begin{align*}
    &\hspace{0.5cm} \frac{p_A(X)}{p_{\cD_0}(X)} \\
    &= \BE_{z \sim \cN(0, \frac{\alpha^2}{d} \cdot I)} \exp\left(-\beta \cdot \langle \mathbf{T}_A, \mathbf{R}_A \rangle - \frac{(1-\eps) n}{2} \cdot \beta^2 \cdot \|\mathbf{R}_A\|_2^2 - \langle \mathbf{T}_A + (1-\eps) n \cdot \beta \cdot \mathbf{R}_A, z \rangle - \frac{(1-\eps) n}{2} \cdot \|z\|_2^2\right)
    \\
    &= \exp\left(-\beta \langle \mathbf{T}_A, \mathbf{R}_A \rangle - \frac{(1-\eps) n}{2}  \beta^2  \|\mathbf{R}_A\|_2^2\right)  \BE_{z \sim \cN(0, \frac{\alpha^2}{d} \cdot I)} \exp\left(- \frac{(1-\eps) n}{2}  \|z\|_2^2 +\langle \mathbf{T}_A + (1-\eps) n \beta \mathbf{R}_A, z \rangle\right) \\
    &= \exp\left(-\beta  \langle \mathbf{T}_A, \mathbf{R}_A \rangle - \frac{(1-\eps) n}{2}  \beta^2 \|\mathbf{R}_A\|_2^2\right) \cdot \left(1 + \frac{(1-\eps) n \cdot \alpha^2}{d}\right)^{-d/2} \cdot \exp\left(\frac{\|(1-\eps) n \beta \mathbf{R}_A + \mathbf{T}_A \|_2^2}{2 ((1-\eps) n + d/\alpha^2)}\right).
\end{align*}
We can combine the terms that are in terms of $\|\mathbf{R}_A\|_2^2, \|\mathbf{T}_A\|_2^2,$ and $\langle \mathbf{R}_A, \mathbf{T}_A \rangle$, to simplify this as
\[\left(1 + \frac{(1-\eps) n \cdot \alpha^2}{d}\right)^{-d/2} \cdot \exp\left(-\frac{(1-\eps) n \beta^2 d \cdot \|\mathbf{R}_A\|_2^2 + 2 \beta d \cdot \langle \mathbf{R}_A, \mathbf{T}_A \rangle - \alpha^2 \cdot \|\mathbf{T}_A\|_2^2}{2((1-\eps) \alpha^2 n+d)}\right). \qedhere\]
\end{proof}

Now, recall that the $\chi^2$ divergence $\dchi(\cD_1||\cD_0)$ equals 
\[\mathop{\BE}\limits_{X \sim \cD_0} \left(\frac{p_{\cD_1}(X)}{p_{\cD_0}(X)}\right)^2 = \mathop{\BE}\limits_{X_1, \dots, X_n \sim \cN(0, I)} \mathop{\BE}\limits_{A, B \subset [n]} \left(\frac{p_{A}(X) p_{B}(X)}{p_{\cD_0}(X)^2}\right),\]
where $A, B$ will always denote random subsets of size $\eps n$ in $[n]$. Using \Cref{lem:likelihood-conditioned-subset}, we can write this as
\begin{multline} 
    \hspace{-0.4cm} \left(1 + \frac{(1-\eps) n \alpha^2}{d}\right)^{-d} \cdot \mathop{\BE}\limits_{A, B} \mathop{\BE}\limits_{X_1, \dots, X_n \sim \cN(0, I)} 
    \exp\Biggr[ \\
    \hspace{0.5cm} -\frac{1}{2} \cdot \frac{(1-\eps) n \beta^2 d \cdot (\|\mathbf{R}_A\|_2^2+\|\mathbf{R}_B\|_2^2) + 2 \beta d \cdot (\langle \mathbf{R}_A, \mathbf{T}_A \rangle + \langle \mathbf{R}_B, \mathbf{T}_B \rangle) - \alpha^2 \cdot (\|\mathbf{T}_A\|_2^2 + \|\mathbf{T}_B\|_2^2)}{(1-\eps) \alpha^2 n+d}\Biggr].
    \label{eq:chi-square-bash1}  
\end{multline}

Now, the exponential term can be decomposed coordinate-wise, and since each coordinate of $X_1, \dots, X_n$ is independent if we condition on $A, B$, we can therefore write \eqref{eq:chi-square-bash1} after removing the expectation on $A, B$ as
\begin{equation} \label{eq:chi-square-bash2}
    \left(1 + \frac{(1-\eps) n \alpha^2}{d}\right)^{-d} \cdot \left(\mathop{\BE}_{x_1, \dots, x_n \sim \cN(0, 1)} \exp\left(-\frac{1}{2} \cdot \frac{x^\top (M_A + M_B) x}{(1-\eps) \alpha^2 n + d}\right)\right)^d
\end{equation}
Above, each $x_i$ is a standard univariate Gaussian, and $M_A$ is the $n \times n$ matrix with blocks
\[
\newcommand\explainA{\overbrace{\hphantom{\begin{matrix}0&0&0&0\end{matrix}}}^{A}}
\newcommand{\explainB}{\overbrace{\hphantom{\begin{matrix}0&0&0&0\end{matrix}}}^{A^c}}
\newcommand{\explainC}{A\hphantom{^c}\left\{\vphantom{\begin{matrix}0\\0\end{matrix}}\right.}
\newcommand{\explainD}{A^c\left\{\vphantom{\begin{matrix}0\\0\end{matrix}}\right.}
\settowidth{\dimen0}{$\begin{pmatrix}\vphantom{\begin{matrix}0\\0\\0\end{matrix}}\end{pmatrix}$}
\settowidth{\dimen2}{$\explainB$}
\begin{matrix}
\hspace*{2em}\begin{matrix}\hspace*{0.5\dimen0}\explainA&\explainB\hspace*{0.5\dimen0}\end{matrix}
\\[-0.5ex]
\hspace*{-1em}\begin{matrix}\explainC\\ \explainD \end{matrix}
\begin{pmatrix}&\\[-0.5ex]
\makebox[\dimen2]{$(1-\eps) n \beta^2 d$}
&
\makebox[\dimen2]{$\beta d$} \\[-0.5ex]&\\
\vphantom{\begin{matrix}0\\0\end{matrix}}\text{$\beta d$} & \text{$-\alpha^2$}
\end{pmatrix}
\end{matrix} \,
\]
and $M_B$ is defined similarly.
Here, each block is dependent on whether the row/column indices are in $A$ or $A^c$, and all entries in the same block are the same.
Note that $M_A$ has rank at most $2$. Moreover, by projecting onto the space of vectors $v$ where $v_i$ is constant for all $i \in A$, and constant for all $i \in A^c$, we have that $M_A$ has the same nonzero eigenvalues as $\sqrt{D_A} \Sigma_{A} \sqrt{D_A}$, where
\[D_A = \begin{pmatrix} \eps n & 0 \\ 0 & (1-\eps) n\end{pmatrix}, \hspace{0.5cm} \Sigma_{A} = \begin{pmatrix} (1-\eps) n \beta^2 d & \beta d \\ \beta d & -\alpha^2 \end{pmatrix}.\]
If we define $M_{A, B} = M_A+M_B$, we can write $M_{A, B}$ in a similar block-diagonal format, where the rows/columns are split based on the index being in $A \cap B, A \cap B^c, A^c \cap B,$ or $A^c \cap B^c$.
Therefore, if $|A \cap B| = \gamma \cdot n$ for some $0 \le \gamma \le \eps$, $M_{A, B}$ has the same nonzero eigenvalues as $\sqrt{D_{A, B}} \Sigma_{A, B} \sqrt{D_{A, B}}$, where
\[D_{A, B} = \begin{pmatrix} \gamma n & 0 & 0 & 0\\ 0 & (\eps-\gamma) n & 0 & 0 \\ 0 & 0 & (\eps-\gamma) n & 0 \\ 0 & 0 & 0 & (1-2\eps+\gamma) n \end{pmatrix}\]
and
\[\Sigma_{A, B} := \begin{pmatrix} 2(1-\eps) n \beta^2 d & (1-\eps) n \beta^2 d + \beta d & (1-\eps) n \beta^2 d + \beta d & 2 \beta d \\ (1-\eps) n \beta^2 d + \beta d & (1-\eps) n \beta^2 d - \alpha^2 & 2\beta d & \beta d - \alpha^2 \\ (1-\eps) n \beta^2 d + \beta d & 2 \beta d & (1-\eps) n \beta^2 d - \alpha^2 & \beta d - \alpha^2 \\ 2 \beta d & \beta d - \alpha^2 & \beta d - \alpha^2 & -2 \alpha^2 \end{pmatrix}.\]

Now, for any subsets $A, B \subset [n]$ of size $\eps \cdot n$, we define $G_{A} = \frac{1}{(1-\eps) \alpha^2 n + d} \cdot M_{A}$ and $G_{A, B} = G_A+G_B = \frac{1}{(1-\eps) \alpha^2 n + d} \cdot M_{A, B}$.
We note the following basic proposition.

\begin{proposition} \label{prop:no-huge-neg-eigs}
    Suppose that $n \le \frac{0.1 d}{\alpha^2}$ and $0 \le \beta \le \frac{0.1}{n}$. Then, all eigenvalues of $G_A$ are strictly greater than $-\frac{1}{2}$.

    As a direct corollary, all eigenvalues of $G_{A, B}$, for any $A, B,$ are strictly greater than $-1$.
\end{proposition}

\begin{proof}
    It suffices to prove the claim for $\hat{G}_A := \frac{1}{(1-\eps) \alpha^2 n + d} \cdot \sqrt{D_A} \Sigma_{A} \sqrt{D_A}$.
    Note that $\hat{G}_A$ is a $2 \times 2$ symmetric matrix. If $\hat{G}_A$ has eigenvalues $\lambda_1, \lambda_2$, then we need to show that $\lambda_1+\frac{1}{2}, \lambda_2+\frac{1}{2} > 0$. It therefore suffices to show that $(\lambda_1+\frac{1}{2}) + (\lambda_2+\frac{1}{2}) = \Tr(\hat{G}_A)+1$ and $(\lambda_1+\frac{1}{2}) \cdot (\lambda_2+\frac{1}{2}) = \det(\hat{G}_A) + \frac{1}{2} \Tr(\hat{G}_A) + \frac{1}{4}$ are both strictly greater than $0$.

    Note that 
\[\Tr(\hat{G}_A) = \frac{1}{(1-\eps) \alpha^2 n + d} \cdot \Tr(\Sigma_{A} \cdot D_A) = \frac{1}{(1-\eps) \alpha^2 n + d} \cdot \left[(1-\eps) n \beta^2 d \cdot \eps n - \alpha^2 \cdot (1-\eps) n\right] \ge \frac{-\alpha^2 (1-\eps) n}{(1-\eps) \alpha^2 n + d}.\]
    We are assuming that $n \le \frac{0.1 d}{\alpha^2},$ which means that $(1-\eps) \alpha^2 n \le 0.1 d$. So in fact, $\Tr(\hat{G}_A) \ge -0.1,$ so $\Tr(\hat{G}_A)+1 \ge 0.9 > 0$.

    Next,
\begin{align*}
    \det(\hat{G}_A) &= \frac{1}{((1-\eps) \alpha^2 n + d)^2} \cdot \det(D_A) \cdot \det(\Sigma_{A})  \\
    &= \frac{\eps n \cdot (1-\eps) n}{((1-\eps) \alpha^2 n + d)^2} \cdot \left((1-\eps) n \beta^2 d \cdot (-\alpha^2) - (\beta d)^2\right) \\
    &= -\frac{\eps (1-\eps) n^2 \cdot \beta^2 d \cdot ((1-\eps) \alpha^2 n + d)}{((1-\eps) \alpha^2 n + d)^2}\\
    &= - \frac{\eps (1-\eps) n^2 \cdot \beta^2 d}{(1-\eps) \alpha^2 n + d}.
\end{align*}
    Since $0 \le \beta \le 0.1/n$, this means $\det(\hat{G}_A) \ge -\frac{0.01 \eps (1-\eps) d}{(1-\eps) \alpha^2 n + d} \ge -\frac{0.01 d}{d} = -0.01$. So, $\det(\hat{G}_A) + \frac{1}{2} \Tr(\hat{G}_A) + \frac{1}{4} \ge -0.01 - 0.05 + 0.25 > 0$.
\end{proof}

As a result of \Cref{prop:no-huge-neg-eigs}, we can apply \Cref{prop:exponential-expectation-determinant} to obtain the following.

\begin{lemma} \label{lem:chi-square-det-calc-1}
    Assuming that $n \le \frac{0.1d}{\alpha^2}$ and $\beta \le \frac{0.1}{n},$ the $\chi^2$ divergence $\dchi(\cD_1||\cD_0)$ equals
\[\mathop{\BE}\limits_{A, B} \left[\left(\left(\frac{d + (1-\eps) n \alpha^2}{d}\right)^2 \cdot \det(I+G_{A,B})\right)^{-d/2}\right].\]
\end{lemma}

\begin{proof}
    We have the following chain of equalities. The first equality follows by combining \eqref{eq:chi-square-bash1} and \eqref{eq:chi-square-bash2}, the second follows by the definition of $G_{A, B}$, the third follows by \Cref{prop:exponential-expectation-determinant}, and the final follows by basic manipulation.
\begin{align*}
    \dchi(\cD_1||\cD_0)
    &= \mathop{\BE}\limits_{A, B} \left[\left(1 + \frac{(1-\eps) n \alpha^2}{d}\right)^{-d} \cdot \left(\mathop{\BE}_{x_1, \dots, x_n \sim \cN(0, 1)} \exp\left(-\frac{1}{2} \cdot \frac{x^\top (M_A + M_B) x}{(1-\eps) \alpha^2 n + d}\right)\right)^d\right] \\
    &= \mathop{\BE}\limits_{A, B} \left[\left(1 + \frac{(1-\eps) n \alpha^2}{d}\right)^{-d} \cdot \left(\mathop{\BE}_{x_1, \dots, x_n \sim \cN(0, 1)} \exp\left(-\frac{1}{2} \cdot x^\top \cdot G_{A, B} \cdot x\right)\right)^d\right] \\
    &= \mathop{\BE}\limits_{A, B} \left[\left(1 + \frac{(1-\eps) n \alpha^2}{d}\right)^{-d} \cdot \det (I+G_{A, B})^{-d/2}\right] \\
    &= \mathop{\BE}\limits_{A, B} \left[\left(\left(\frac{d + (1-\eps) n \alpha^2}{d}\right)^2 \cdot \det (I+G_{A, B})\right)^{-d/2}\right] \, . \qedhere
\end{align*}
\end{proof}

\subsection{Final Computation}
    \label{sec:final:computation:mathematica}

Through some tedious computations, one can show the following:
\begin{lemma}
    \label{lemma:tedious:compputation}
Suppose that $|A \cap B| = \gamma \cdot n$, for some $0 \le \gamma \le \eps$. Then,
\begin{multline} \label{eq:factoring}
    \det(I+G_{A, B}) = \frac{1}{(d+(1-\eps) \alpha^2 n)^2} \cdot \\
    \left[(d+\beta^2 d (\eps^2-\gamma) n^2)^2 - (\alpha^2 n - 2 \beta d \eps n + \beta^2 d \eps n^2 - 2 \alpha^2 \eps n + \alpha^2 \gamma n + 2 \beta d \gamma n - 2 \beta^2 d \eps^2 n^2 + \beta^2 d \eps \gamma n^2)^2\right]
\end{multline}
\end{lemma}

\begin{proof}
    Note that $G_{A, B}$ has the same eigenvalues as $\hat{G}_{A, B} := \frac{1}{(1-\eps) \alpha^2 n + d} \cdot \sqrt{D_{A, B}} \Sigma_{A, B} \sqrt{D_{A, B}}$. Hence,
\[\det(I+G_{A, B}) = \det\left(I + \frac{1}{(1-\eps) \alpha^2 n + d} \cdot \sqrt{D_{A, B}} \Sigma_{A, B} \sqrt{D_{A, B}}\right) = \det\left(I+\frac{1}{(1-\eps) \alpha^2 n + d} D_{A, B} \cdot \Sigma_{A, B}\right).\]
    We can then can compute and factor the determinant as an expression of $\alpha, \eps, \beta, \gamma, d,$ and $n$. Writing the output as a difference of squares, one then obtains \eqref{eq:factoring}.\footnote{Some Mathematica code to verify the computation is provided in~\cref{app:mathematica}.}
\end{proof}

Now, we will set $\beta$ to satisfy the quadratic equation $\alpha^2 n - 2 \beta d \eps n + \beta^2 d \eps n^2 = 0$. This is equivalent to $\beta^2 (d \eps n) - (2 d \eps) \beta + \alpha^2 = 0$, for which we will set $\beta$ to be the solution
\[\beta = \frac{d \eps - \sqrt{d^2 \eps^2 - \alpha^2 \cdot d \eps n}}{d \eps n} = \frac{1}{n} \cdot \left(1 - \sqrt{1 - \frac{\alpha^2 n}{d \eps}}\right).\]
Note that this is only possible if $\alpha^2 n < d \eps$, so $n < \frac{d \eps}{\alpha^2}.$
In this case, we can simplify our expression as
\[\det(I+G_{A, B}) = \frac{(d+\beta^2 d(\eps^2-\gamma)n^2)^2 - O(\alpha^2 \eps n + \alpha^2 \gamma n + \beta d \gamma n + \beta^2 d \eps^2 n^2 + \beta^2 d \eps \gamma n^2)^2}{(d+(1-\eps) \alpha^2 n)^2}.\]
Using the fact that $\gamma \le \eps$, we can ignore the terms $\alpha^2 \gamma n$ (smaller than $\alpha^2 \eps n$) and $\beta^2 d \eps \gamma n^2$ (smaller than $\beta^2 d \eps^2 n^2$). So, this simplifies to 
\[\det(I+G_{A, B}) = \frac{(d+\beta^2 d(\eps^2-\gamma)n^2)^2 - O(\alpha^2 \eps n + \beta d \gamma n + \beta^2 d \eps^2 n^2)^2}{(d+(1-\eps) n \alpha^2)^2}.\]
In addition, note that if $n < \frac{d \eps}{\alpha^2}$, then $\sqrt{1 - \frac{\alpha^2 n}{d \eps}} \ge 1 - \frac{\alpha^2 n}{d \eps}$, which means $\beta \le \frac{1}{n} \cdot \frac{\alpha^2 n}{d \eps} = \frac{\alpha^2}{d \eps}.$ Hence, if $n < \frac{d \eps}{\alpha^2}$,
\begin{equation} \label{eq:some-more-det-bash}
    \det(I+G_{A, B}) \ge \frac{(d - \frac{\alpha^4}{d \eps^2} \cdot \max(0, \gamma-\eps^2) \cdot n^2)^2 - O(\alpha^2 \eps n + \frac{\alpha^2}{\eps} \cdot \gamma n + \frac{\alpha^4}{d} \cdot n^2)^2}{(d+(1-\eps) n \alpha^2)^2}.
\end{equation}

Note that if $n \le \frac{0.1 d \eps}{\alpha^2},$ then $\beta \le \frac{\alpha^2}{d \eps} \le \frac{0.1}{n}$. So, by combining \eqref{eq:some-more-det-bash} with \cref{lem:chi-square-det-calc-1}, we have the following lemma.

\begin{lemma} \label{lem:chi-square-bash-3}
    Suppose that $n \le \frac{0.1 d \eps}{\alpha^2}$ and $\beta = \frac{1}{n} \cdot (1-\sqrt{1-(\alpha^2 n)/(d \eps)})$. Then, for $\gamma := \frac{|A \cap B|}{n}$, we have
\begin{equation} \label{eq:chi-square-bash3}
    \dchi(\cD_1||\cD_0) \le \mathop{\BE}_{A, B} \left[\left(\frac{(d - \frac{\alpha^4}{d \eps^2} \cdot \max(0, \gamma-\eps^2) \cdot n^2)^2 - O(\alpha^2 \eps n + \frac{\alpha^2}{\eps} \cdot \gamma n + \frac{\alpha^4}{d} \cdot n^2)^2}{d^2}\right)^{-d/2}\right].
\end{equation}
\end{lemma}

Recall that $\gamma$ is the fraction of $[n]$ in both $A$ and $B$, so the distribution of $\gamma$ is $\frac{1}{n} \cdot \hgeom(n, \eps n, \eps n)$. Hence, $\gamma \in [0, \eps]$ with probability $1$, and by \Cref{cor:hypergeometric}, $\BP(\max(\gamma-\eps^2, 0) > t) \le \exp\left(-\min\left(\frac{t^2 \cdot n}{4\eps^2}, \frac{t \cdot n}{4}\right)\right).$

We note the following simple proposition.

\begin{proposition} \label{prop:something_dumb}
    Suppose that $n \le c \cdot \frac{d \eps}{\alpha^2}$ for some small constant $c$. Then, for any $A$ and $B$, each of $\frac{\alpha^4}{d \eps^2} \cdot \max(0, \gamma-\eps^2) \cdot n^2$, $\alpha^2 \eps n,$ $\frac{\alpha^2}{\eps} \cdot \gamma \cdot n$, and $\frac{\alpha^4}{d} \cdot n^2$ is smaller than $c \cdot d.$
\end{proposition}

\begin{proof}
    Since $\gamma \le \eps$, $\frac{\alpha^4}{d \eps^2} \cdot \max(0, \gamma-\eps^2) \cdot n^2 \le \frac{\alpha^4 n^2}{d \eps}.$ If $\frac{\alpha^4 n^2}{d \eps} \ge c \cdot d$, then $n \ge \sqrt{c} \cdot \frac{d \sqrt{\eps}}{\alpha^2} > c \cdot \frac{d \eps}{\alpha^2}$. Next, $\alpha^2 \eps n, \frac{\alpha^2}{\eps} \cdot \gamma n \le \alpha^2 n$. If $\alpha^2 n \ge c \cdot d,$ then $n \ge c \cdot \frac{d}{\alpha^2} > c \cdot \frac{d \cdot \eps}{\alpha^2}.$ Finally, if $\frac{\alpha^4}{d} \cdot n^2 \le c \cdot d$, then $n \ge \sqrt{c} \cdot \frac{d}{\alpha^2} > c \cdot \frac{d \cdot \eps}{\alpha^2}$.
\end{proof}

The importance of \Cref{prop:something_dumb} is that if $0 < x \le c$ for a sufficiently small constant $c$, $1-x \ge e^{-2x}$. Therefore, we can rewrite the right-hand side of \eqref{eq:chi-square-bash3} as at most
\begin{align}
    &\hspace{0.5cm} \mathop{\BE}\limits_{A, B} \left[\exp\left(O\left(\frac{\frac{\alpha^4}{d \eps^2} \cdot \max(0, \gamma-\eps^2) \cdot n^2}{d} + \left(\frac{\alpha^2 \eps n + \frac{\alpha^2}{\eps} \cdot \gamma n + \frac{\alpha^4}{d} \cdot n^2}{d}\right)^2\right) \cdot \frac{d}{2}\right)\right] \nonumber \\
    &= \mathop{\BE}\limits_{A, B} \left[\exp\left(O\left(\frac{\alpha^4}{d \eps^2} \cdot \max(0, \gamma-\eps^2) \cdot n^2 + \left(\frac{\alpha^2 \eps n + \frac{\alpha^2}{\eps} \cdot \gamma n + \frac{\alpha^4}{d} \cdot n^2}{\sqrt{d}}\right)^2\right)\right)\right] . \label{eq:chi-square-bash4}
\end{align}

First, note that if we additionally have $n \le c \cdot \min\left(\frac{\sqrt{d}}{\alpha^2 \eps}, \frac{d^{3/4}}{\alpha^2}\right)$ for a small constant $c$, then $\alpha^2 \eps n \le c \sqrt{d}$ and $\frac{\alpha^4}{d} \cdot n^2 \le c^2 \sqrt{d}$. Also, note that $\frac{\alpha^2}{\eps} \cdot \gamma n \le \alpha^2 \eps n + \frac{\alpha^2}{\eps} \cdot \max(0, \gamma-\eps^2) n,$ and that $\left[(\frac{\alpha^2}{\eps} \cdot \max(0, \gamma-\eps^2) n)/\sqrt{d}\right]^2 = \frac{\alpha^4}{d \eps^2} \cdot \max(0, \gamma-\eps^2)^2 \cdot n^2 \le \frac{\alpha^4}{d \eps^2} \cdot \max(0, \gamma-\eps^2) \cdot n^2,$ since $\max(0, \gamma-\eps^2) \le 1.$ As a result, we can bound \eqref{eq:chi-square-bash4} as at most
\[\mathop{\BE}_{A, B} \left[\exp\left(O\left(\frac{\alpha^4}{d \eps^2} \cdot \max(0, \gamma-\eps^2) \cdot n^2 + c^2\right)\right)\right],\]
assuming $n \le c \cdot \min\left(\frac{d \eps}{\alpha^2}, \frac{\sqrt{d}}{\alpha^2 \eps}, \frac{d^{3/4}}{\alpha^2}\right)$.

Now, for any value $t > 0$, we have that by \Cref{cor:hypergeometric},
\begin{align*}
    \BP_{A, B}\left(\frac{\alpha^4}{d \eps^2} \cdot \max(0, \gamma-\eps^2) \cdot n^2 > t\right)
    &= \BP_{A, B}\left(\gamma-\eps^2 > t \cdot \frac{d \eps^2}{\alpha^4 n^2}\right) \\
    &\le \exp\left(-\min\left(t^2 \cdot \frac{d^2 \eps^4}{\alpha^8 n^4} \cdot \frac{n}{4 \eps^2}, t \cdot \frac{d \eps^2}{\alpha^4 n^2} \cdot \frac{n}{4}\right)\right) \\
    &= \exp\left(-\min\left(t^2 \cdot \frac{d^2 \eps^2}{4 \alpha^8 n^3}, t \cdot \frac{d \eps^2}{4 \alpha^4 n}\right)\right).
\end{align*}
If we additionally assume that $n \le c \cdot \frac{d^{2/3} \eps^{2/3}}{\alpha^{8/3}}$ and $n \le c \cdot \frac{d \eps^2}{\alpha^4},$ this is at most $\exp\left(-\min(t^2/4 c^3, t/4 c)\right).$ So, for $t \ge c$, the probability that $\frac{\alpha^4}{d \eps^2} \cdot \max(0, \gamma-\eps^2) \cdot n^2 \ge t$ is at most $e^{-t/4 c}$, which means that
\[\mathop{\BE}_{A, B} \left[\exp\left(O\left(\frac{\alpha^4}{d \eps^2} \cdot \max(0, \gamma-\eps^2) \cdot n^2 + c^2\right)\right)\right] \le e^{O(c)} \le 1.01.\]

To summarize what we have proved, in combination with \Cref{lem:chi-square-bash-3}, we have the following.

\begin{lemma} \label{lem:almost-lb}
    Assuming that
\begin{equation} \label{eq:assumptions-big}
    n \le c \cdot \min\left(\frac{d \eps}{\alpha^2}, \frac{\sqrt{d}}{\alpha^2 \eps}, \frac{d^{3/4}}{\alpha^2}, \frac{d^{2/3} \eps^{2/3}}{\alpha^{8/3}}, \frac{d \eps^2}{\alpha^4}\right),
\end{equation}    
    for some sufficiently small constant $c$, we have that 
\[\dchi(\cD_1||\cD_0) \le 1.01.\]
\end{lemma}

However, we note that we can remove several of the terms in \eqref{eq:assumptions-big}. More precisely, we have the following proposition.

\begin{proposition} \label{prop:remove-assumptions}
    Suppose that $n \le c \cdot \min\left(\frac{d^{2/3} \eps^{2/3}}{\alpha^{8/3}}, \frac{d \eps}{\alpha^2}\right)$ for some sufficiently small constant $c > 0$, and that $\eps \le 1$ and $n \ge \frac{\sqrt{d}}{\alpha^2} + \frac{d \eps^3}{\alpha^4}$. Then, $n \le c \cdot \min\left(\frac{d \eps}{\alpha^2}, \frac{\sqrt{d}}{\alpha^2 \eps}, \frac{d^{3/4}}{\alpha^2}, \frac{d^{2/3} \eps^{2/3}}{\alpha^{8/3}}, \frac{d \eps^2}{\alpha^4}\right).$
\end{proposition}

\begin{proof}
    First, note that $\left(\frac{\sqrt{d}}{\alpha^2 \eps}\right)^{2/3} \cdot \left(\frac{d \eps^3}{\alpha^4}\right)^{1/3} = \frac{d^{2/3} \eps^{1/3}}{\alpha^{8/3}} \ge \frac{d^{2/3} \eps^{2/3}}{\alpha^{8/3}},$ since $\eps \le 1$. Therefore, if $\frac{d^{2/3} \eps^{2/3}}{\alpha^{8/3}} > \frac{d \eps^3}{\alpha^4}$, then $\frac{\sqrt{d}}{\alpha^2 \eps} > \frac{d^{2/3} \eps^{2/3}}{\alpha^{8/3}}$. Thus, if $n \le c \cdot \min\left(\frac{d \eps}{\alpha^2}, \frac{d^{2/3} \eps^{2/3}}{\alpha^{8/3}}\right)$ and $n \ge \frac{d \eps^3}{\alpha^4} + \frac{\sqrt{d}}{\alpha^2}$, then also $n \le c \cdot \frac{\sqrt{d}}{\alpha^2 \eps}.$ 
    In addition, $\frac{d^{3/4}}{\alpha^2} = \sqrt{\frac{\sqrt{d}}{\alpha^2 \eps} \cdot \frac{d \eps}{\alpha^2}},$ which means we also obtain $n \le c \cdot \frac{d^{3/4}}{\alpha^2}$, because we just showed that $n \le c \cdot \frac{\sqrt{d}}{\alpha^2 \eps}$ and we are assuming that $n \le c \cdot \frac{d \eps}{\alpha^2}$. Finally, $\left(\frac{\sqrt{d}}{\alpha^2}\right)^{2/3} \cdot \left(\frac{d \eps^2}{\alpha^4}\right)^{1/3} = \frac{d^{2/3} \eps^{2/3}}{\alpha^{8/3}},$ which means that if $\frac{d^{2/3} \eps^{2/3}}{\alpha^{8/3}} > \frac{\sqrt{d}}{\alpha^2},$ then $\frac{d \eps^2}{\alpha^4} > \frac{d^{2/3} \eps^{2/3}}{\alpha^{8/3}}$. So, by our assumptions, $n \le c \cdot \frac{d \eps^2}{\alpha^4}$ as well.
\end{proof}

From here, the proof of \Cref{thm:oblivious-lb} is straightforward.

\begin{proof}[Proof of~\Cref{thm:oblivious-lb}]
    By \Cref{lem:almost-lb} and \Cref{prop:remove-assumptions}, we have that under the assumptions of \Cref{thm:oblivious-lb}, $\dchi(\cD_1||\cD_0) \le 1.01$. By \Cref{prop:chi-tv}, we have that $\dtv(\cD_1||\cD_0) \le 0.1$.

    Finally, note that we created the adversarial samples and the mean vector $\mu$ first, and then generated the uncorrupted data, so the adversary is oblivious. Finally, $\|\mu\|_2 \le \|z\|_2+\beta \cdot \|\mathbf{R}_A\|_2$. However, $z \sim \cN(0, \frac{\alpha^2}{d} \cdot I)$ means $\|z\|_2 \le 2 \alpha$ with very high probability. Moreover, $\beta \le \frac{\alpha^2}{\eps d}$ and $\mathbf{R}_A$ is the sum of $\eps \cdot n$ i.i.d. $\cN(0, 1)$, so $\|\mathbf{R}_A\|_2 \le 2 \sqrt{\eps n d}$ with very high probability. So, $\beta \cdot \|\mathbf{R}_A\|_2 \le \frac{\alpha^2 \cdot 2 \sqrt{\eps n d}}{\eps d} = 2 \alpha^2 \sqrt{\frac{n}{\eps d}}.$ Assuming that $n \le \frac{d \eps}{\alpha^2},$ this is at most $2 \alpha$. So overall, $\|\mu\|_2 \le 4 \alpha$. We can replace $\alpha$ with $\alpha/4$ in the construction to finish the proof.
\end{proof}

\section{The Sample Complexity under Strong Contamination}
\label{sec:strong-sample-complexity}

In this section, we leverage the tight sample complexity bounds for differentially private mean testing~\cite{narayanan2022privatetesting}, along with the robust-private equivalence of~\cite{georgiev2022privatetorobust, hopkins2023robusttoprivate, asi2023robusttoprivate}, to obtain the optimal sample complexity of robust mean testing under the strong contamination model:
\begin{theorem} \label{thm:adaptive-sample-complexity-main}
    For $\dst \ge \cor \cdot \poly\!\log (d, \frac{1}{\cor}, \frac{1}{\dst})$,\footnote{While this condition may seem somewhat restrictive, it is in fact inconsequential. Indeed, for $\cor\leq \dst\leq \cor \cdot \poly\!\log (d, {1}/{\cor}, {1}/{\dst})$, one can use a robust \emph{learning} algorithm with sample complexity $\bigO{{\dims}/{\dst^2}}$, which in this parameter regime becomes $\tilde{O}\mleft({\dims \cor^2}/{\dst^4}\mright)$.} the sample complexity of Gaussian mean testing in the  adaptive contamination model is
\[
\tilde{\Theta}\left(\frac{d^{1/2}}{\alpha^2} + \frac{d \eps^2}{\alpha^4}\right).
\]
\end{theorem}
The rest of this section is dedicated to the proof of this theorem. First, we recall the definition of differential privacy: for simplicity, and as it suffices for our purposes, we focus on ``fully approximate'' differentially private decision algorithms.
\begin{definition}
    A randomized algorithm $\cA\colon \cX^\ns \to \{0, 1\}$ is $(0, \delta)$-\emph{differentially private} (DP) if for all datasets $X, X'\in\cX^\ns$ that only differ in a single data point $X_i \neq X_i'$, 
\[
    \left|\BP(\mathcal{A}(X) = 1) -\BP(\mathcal{A}(X') = 1)\right| \le \delta.
\]
\end{definition}
\paragraph{Upper bound.} To prove our upper bound, we will require the tight upper bound for DP mean testing:
\begin{theorem}[\cite{narayanan2022privatetesting}]\label{thm:narayanan-private-testing}
    For any parameters $0 < \alpha, \delta \le \frac{1}{2}$, there exists a $(0, \delta)$-DP algorithm $\cA$ that on
\[
n = \tilde{O}\mleft(\frac{d^{1/2}}{\alpha^2} + \frac{d^{1/3}}{\alpha^{4/3}  \delta^{2/3}} + \frac{1}{\alpha \delta}\mright)
\]
    samples $X_1, \dots, X_n$, satisfies:
\begin{itemize}
    \item If $X_1, \dots, X_n \overset{i.i.d.}{\sim} \cN(0, I)$, then with probability at least $0.99$ (over both the randomness of the samples and the algorithm), $\mathcal{A}(X) = 0$.\footnote{While~\cite{narayanan2022privatetesting} did not state a $0.99$ success probability, one can amplify the success probability by running several independent copies and using the majority output.}
    \item For any vector $\mu$ with $\|\mu\|_2 \ge \alpha$, if $X_1, \dots, X_n \overset{i.i.d.}{\sim} \cN(\mu, I)$, then with probability at least $0.99$, $\mathcal{A}(X) = 1$.
\end{itemize}
Moreover, this is tight: any $(0, \delta)$-DP algorithm with these guarantees must take $\tilde{\Omega}\mleft(\frac{d^{1/2}}{\alpha^2} + \frac{d^{1/3}}{\alpha^{4/3}  \delta^{2/3}} + \frac{1}{\alpha \delta}\mright)$ samples.
\end{theorem}

It is essentially folklore (see also~\cite{georgiev2022privatetorobust}) that any $(0, \delta)$-DP decision algorithm that succeeds with $0.99$ probability given $\ns$ samples is automatically $\cor$-robust in the strong (adaptive) corruption model for $\cor \eqdef \frac{1}{10 \delta n}$, and succeeds with at least $2/3$ probability over the input. For completeness, we briefly reproduce the argument here for the algorithm $\cA$. If $X_1, \dots, X_n $ are \iid, then with at least $0.9$ probability, $\BP(\mathcal{A}(X_1, \dots, X_n) = 0) \ge 0.9$ over the randomness of the algorithm $\mathcal{A}$. Hence, for any $\cor$-corruption of the data $X'$ (i.e., $\cor\ns = \frac{1}{10 \delta}$ individual data points are possibly adaptively changed from $X$ to $X'$), by the definition of privacy, 
\[
|\BP(\mathcal{A}(X_1, \dots, X_n) = 0) - \BP(\mathcal{A}(X_1', \dots, X_n') = 0)| \le \delta \cdot \frac{1}{10 \delta} = \frac{1}{10}
\]
Hence, for any such corruption $X'$, $\BP(\cA(X') = 0) \ge 0.89 > 2/3$.
The same argument can be used to show that if $X_1, \dots, X_n \overset{\iid}{\sim} \cN(\mu,I)$ where $\|\mu\|_2 \ge \alpha$, with probability at least $0.9$ over $X_1, \dots, X_n$, $\BP(\mathcal{A}(X') = 1) > 2/3$ for any $\cor$-corruption of $X$.

Thus, the algorithm of~\cref{thm:narayanan-private-testing} readily implies an $\cor$-robust one for robust mean testing, for $\cor = \frac{1}{10 \delta n}$. Plugging $\delta = \frac{1}{10\cor\ns}$ in its sample complexity, it suffices for $\ns$ to satisfy
\[
    n \ge \tilde{O}\left(\frac{d^{1/2}}{\alpha^2} + \frac{d^{1/3}}{\alpha^{4/3} \cdot (1/\eps n)^{2/3}} + \frac{1}{\alpha \cdot (1/\eps n)}\right) = \tilde{O}\left(\frac{d^{1/2}}{\alpha^2} + \frac{d^{1/3} \eps^{2/3} n^{2/3}}{\alpha^{4/3}} + \frac{\eps}{\alpha} \cdot n\right).
\]
This is equivalent to requiring 
\[\alpha \ge \tilde{O}(\eps) \hspace{0.5cm} \text{and} \hspace{0.5cm} n \ge \tilde{O}\left(\frac{d^{1/2}}{\alpha^2} + \frac{d \eps^2}{\alpha^4}\right),\]
where $\tilde{O}$ may hide polylogarithmic factors in $d, \alpha^{-1}, \eps^{-1}$.
Hence, there exists a robust algorithm against strong contamination with sample complexity $\tilde{O}\mleft(\frac{d^{1/2}}{\alpha^2} + \frac{d \eps^2}{\alpha^4}\mright)$.

\paragraph{Lower bound.} We next show this sample complexity is optimal, again by a reduction between robust and private algorithms. 
Suppose there exists a robust mean testing algorithm $\mathcal{A}$ that uses $n$ samples. We set $\delta \eqdef \frac{1}{\cor\ns}$, and construct a $(0, \delta)$-differentially private algorithm for mean testing using a black-box robustness-to-privacy transformation~\cite{hopkins2023robusttoprivate, asi2020instance}. We will then use a lower bound from~\cite{narayanan2022privatetesting}, which will create a contradiction if $n$ is too small.

To explain this transformation, first, for any two datasets $X, X'$ of size $n$, we define the Hamming distance $d_{\rm{}H}(X, X')$ to be the number of indices $i$ such that $X_i \neq X_i'$.
Now, for any dataset $X = (X_1, \dots, X_n)$, define the \emph{score} $\mathcal{S}(X ; \cA)$ of $X$ (for $\cA$) to be the smallest nonnegative integer $k$ such that there exists a dataset $X'$ of size $n$ with $d_{\rm{}H}(X, X') = k$ and $\mathcal{A}(X') = 1$. Equivalently, $\mathcal{S}(X ; \cA)$ represents the smallest number of points we need to alter from $X$ to obtain some $X'$ on which the robust algorithm would reject. (Note that if $\mathcal{A}(X) = 1$, then the score of $X$ is simply $0$.) 

The differentially private algorithm $\mathcal{A}'$ on $X$ computes $\mathcal{S}(X ; \cA)$, and then outputs $1$ with probability $\min(0, 1 - \delta \cdot \mathcal{S}(X ; \cA))$.

Note that $\mathcal{S}(X ; \cA)$ changes by at most $1$ between adjacent datasets $X, X'$, because if $\mathcal{S}(X ; \cA) = k$, there exists $X''$ with $d_{\rm{}H}(X, X'') = k$ and $\mathcal{A}(X'') = 1$. But then, $d_{\rm{}H}(X', X'') \le k+1$, so $\mathcal{S}(X' ; \cA) \le k+1$. Likewise, we can show $\mathcal{S}(X' ; \cA) \ge k-1$. This proves that the algorithm is $(0, \delta)$-differentially private, since the probability of outputting $1$ changes by at most $\delta$ if the score changes by at most $1$.

Next, if $X = (X_1, \dots, X_n) \overset{i.i.d.}{\sim} \cN(0, I)$, then by the property of the robust algorithm, with probability at least $2/3$, every dataset $X'$ of Hamming distance at most $\cor\ns$ from $X$ satisfies $\mathcal{A}(X') = 0$. Whenever this happens, $\mathcal{S}(X ; \cA) \ge \cor\ns$, and thus conditioned on this the algorithm $\cA'$ outputs $1$ with probability $0$, and hence always outputs $0$.

Finally, if $X = (X_1, \dots, X_n) \overset{i.i.d.}{\sim} \cN(\mu, I)$, then with probability at least $2/3$, $\mathcal{A}(X) = 1$. Hence, with probability at least $2/3$ we have $\mathcal{S}(X ; \cA) = 0$, and conditioned on this the algorithm $\cA'$ outputs $1$ with probability $1$.

Overall, this means that if $\cA$ is robust against strong contamination, then there exists an algorithm $\cA'$ that is $(0, \frac{1}{\eps n})$-differentially private for the Gaussian mean testing problem, with the same number of samples $n$.

However, we can now invoke the lower bound part of~\cref{thm:narayanan-private-testing} for DP Gaussian mean testing.
From the above reduction, a robust algorithm using $n$ samples yields an $(0, \frac{1}{\eps n})$-DP algorithm with the same sample complexity, which by~\cref{thm:narayanan-private-testing} means that one must have
\[
n = \tilde{\Omega}\left(\frac{d^{1/2}}{\alpha^2} + \frac{d^{1/3}}{\alpha^{4/3}/(\eps n)^{2/3}} + \frac{1}{\alpha/(\eps n)}\right) = \tilde{\Omega}\left(\frac{d^{1/2}}{\alpha^2} + \frac{d^{1/3}\eps^{2/3}}{\alpha^{4/3}} \cdot n^{2/3} + \frac{\eps}{\alpha} \cdot n\right).
\]
This implies that
\[n = \tilde{\Omega}\left(\frac{d^{1/2}}{\alpha^2}\right) \hspace{1cm} \text{and} \hspace{1cm} n = \tilde{\Omega}\left(\frac{d \eps^2}{\alpha^4}\right).\]
This concludes the proof of~\cref{thm:adaptive-sample-complexity-main}. \qed

\paragraph{Lower Bound against Additive Adversaries.} Finally, we note that the same lower bound holds even if we restrict ourselves to adaptive adversaries that only can \emph{add} points, and can never remove points. This again follows readily from the results of~\cite{narayanan2022privatetesting}, but is a consequence of an intermediate result proven in the paper rather than a direct black-box application of their private sample complexity lower bound. The lemma that we require is the following.

\begin{lemma}[{\cite[Theorem D.6, restated]{narayanan2022privatetesting} }]\label{lem:adaptive-main-lemma}
    Fix any $\alpha, \delta \le 1$ and any dimension $d$. There exists a distribution $\mathcal{D}$ over $\R^d$ with support only on $\{\mu\in\R^d: \|\mu\|_2 \ge \alpha\}$, with the following property. Suppose $\mathcal{U}$ is the distribution over $(X_1, \dots, X_n) \in (\R^d)^n$ where each $X_i \sim \cN(0, I)$, and $\mathcal{V}$ is the distribution over $(X_1, \dots, X_n) \in (\R^d)^n$ where we first draw $\mu \sim \mathcal{D}$ and then draw each $X_i \sim \cN(\mu, I)$.

    Then, for some universal constants $c_1, c_2 > 0$, if $n \le c_1 \cdot \frac{d^{1/3}}{\alpha^{4/3} \cdot \delta^{2/3}}$ there exist distributions $\mathcal{U}', \mathcal{V}'$ over $(\R^d)^n$ such that $\dtv(\mathcal{U}, \mathcal{U}') \le 1/4$, $\dtv(\mathcal{V}, \mathcal{V}') \le 1/4$, and there is a coupling of $(\mathcal{U}', \mathcal{V}')$ such that $\BE_{(X, Y) \sim (\mathcal{U}', \mathcal{V}')} [d_{\rm{}H}(X, Y)] \le {c_2}/{\delta}$, where $d_{\rm{}H}$ denotes the Hamming distance, i.e., the number of points that differ between $X$ and $Y$.
\end{lemma}

Now, we show why \Cref{lem:adaptive-main-lemma} implies that any robust algorithm cannot distinguish between i.i.d. samples from $\cN(0, I)$ and $\cN(\mu, I)$, where $\mu$ is drawn from the distribution $\mathcal{D}$ in \Cref{lem:adaptive-main-lemma} under adaptive $\eps$-additive contamination, unless the number of samples is at least $\Omega\left(\frac{d \eps^2}{\alpha^4}\right)$. This would conclude the claim.

Fix $\alpha, \eps \le 1$, and define $\delta = \frac{10 c_2}{\eps n}$, so that $\frac{10 c_2}{\delta} = \eps n$. Suppose that $n \le c_1 \cdot \frac{d^{1/3}}{\alpha^{4/3} \delta^{2/3}},$ which for $\delta = \frac{10 c_2}{\eps n}$ is equivalent to $n \le \frac{c_1^3}{100 c_2^2} \cdot \frac{d \eps^2}{\alpha^4}$.
By \Cref{lem:adaptive-main-lemma} and Markov's inequality, $\BP_{(X, Y) \sim (\mathcal{U}', \mathcal{V}')} \left[d_{\rm{}H}(X, Y) \ge {10 c_2}/{\delta}\right] \le 1/10$, which means by the coupling between $\mathcal{U}$ and $\mathcal{U}'$ and between $\mathcal{V}$ and $\mathcal{V}'$, there exists a coupling between $\mathcal{U}$ and $\mathcal{V}$ with $\BP_{(X, Y) \sim (\mathcal{U}, \mathcal{V})} \left[d_{\rm{}H}(X, Y) \ge {10 c_2}/{\delta}\right] \le 1/4+1/4+1/10=3/5,$ i.e., such that $\BP_{(X, Y) \sim (\mathcal{U}, \mathcal{V})} \left[d_{\rm{}H}(X, Y) \le {10 c_2}/{\delta}\right] \ge 2/5.$

Consider such a coupling between $\mathcal{U}$ and $\mathcal{V}$. Suppose we generate $(X, Y) \sim (\mathcal{U}, \mathcal{V})$, and in the $2/5$ probability event $\left\{d_{\rm{}H}(X, Y) \le {10 c_2}/{\delta}\right\}$, we let $\hat{X} = \hat{Y} = X \cup Y$. Note that $\hat{X}$ can be created by adding at most ${10 c_2}/{\delta}$ points to $X$ and at most ${10 c_2}/{\delta}$ points to $Y$. Otherwise, we let $\hat{X} = X$ and $\hat{Y} = Y$. Importantly, this means there exists a distribution over $\hat{X}$ and $\hat{Y}$ (which are generated only by additive adaptive contamination of $\frac{10 c_2}{\delta} = \eps n$ points) such that with $2/5$ probability, $\hat{X}$ and $\hat{Y}$ are the same. So, the total variation distance between the distributions is at most $3/5$, which means no algorithm can successfully distinguish between the two distributions with more than $80\%$ probability.

In summary, there cannot exist an algorithm that uses $n \le \frac{c_1^3}{100 c_2^2} \cdot \frac{d \eps^2}{\alpha^4}$ samples and distinguishes between samples from $\cN(0, I)$ and $\cN(\mu, I)$ where $\mu \sim \mathcal{D}$, under $\eps$-additive adaptive contamination. Finally, because there exists an $\Omega(\sqrt{d}/\alpha^2)$-lower bound even against uncorrupted samples~\cite{srivastava2008test, DiakonikolasKS17}, we conclude that the sample complexity of robust Gaussian mean testing against additive adaptive adversaries is
\[\Omega\left(\frac{\sqrt{d}}{\alpha^2} + \frac{d \eps^2}{\alpha^4}\right)\,,\]
as claimed. 
\section{Polynomial-Time Algorithm}
\label{sec:poly-time}

\begin{theorem}
\label{thm:poly-time-main}
    Let $d \in \N, \delta > 0$. Let $\alpha = O(1)$ and assume that $C \e \sqrt{\log 1 / \e} \leq \alpha$ and 
    \[
    n \ge \Omega\left( \frac{\sqrt{d} \log 1 / \delta}{\alpha^2} + \frac{d \e^2 \log 1 / \delta}{\alpha^4} \poly \log (d, 1 / \e, 1 / \alpha, \log 1 / \delta) \right ) \, .
    \]
Then, there is an algorithm which runs in time $O(\e n^2 d \min (n, d) + n d)$ with the following guarantees:
  \begin{itemize}
      \item For every $\mu \in \R^d$ with $\|\mu\| = \alpha$, with probability $1-\delta$ over $n$ independent samples $X_1,\ldots,X_n$ from $\cN(\mu,I)$, given any adaptive $\e$-corruption of $X_1,\ldots,X_n$, the algorithm outputs \textsc{yes}.
      \item With probability $1-\delta$ over $n$ independent samples $X_1,\ldots,X_n$ from $\cN(0,I)$, given any adaptive $\e$-corruption of $X_1,\ldots,X_n$, the algorithm outputs \textsc{no}.
  \end{itemize}
\end{theorem}
We briefly mention that the assumption that $\norm{\mu}$ is exactly $\alpha$ is largely for notational convenience this section.
It is straightforward to verify that the same arguments also extend to testing when the mean of the alternative hypothesis satisfies $\alpha \leq \norm{\mu} \leq O(1)$, which implies the same results for $\alpha \leq \norm{\mu}$ (see Subsection~\ref{subsec:wlog}).

\subsection{Regularity conditions}
We will seek to algorithmically enforce a set of regularity conditions which are guaranteed to be satisfied by any set of uncorrupted points, from either the null or alternate hypothesis.
We demonstrate that if these regularity conditions are satisfied, then the norm of the sum of the samples will suffice to distinguish between the two cases, with high probability.
Concretely, the regularity condition we will require is the following:
\begin{definition}
\label{def:regularity}
Let $S = \{X_1, \ldots, X_n\}$ be a set of points in $\R^d$.
We say that $S$ is $(\e, \beta_1, \beta_2)$-regular if for all sets $T \subset S$ with $|T| \leq \e n$, we have:
\begin{enumerate}[(i)]
    \item $\sum_{i \in T} \norm{X_i}^2 = |T| d \pm O(\beta_1)$,
    \item $\norm{\Sum(T)}^2 = |T| d \pm O(\beta_2)$, and
    \item $\abs{\iprod{\Sum(T), \Sum (S)}} = |T| d \pm O(\sqrt{n} \beta_1)$.
\end{enumerate}
\end{definition}
\noindent
We first note the following bound:
\begin{lemma}
\label{lem:good-set-regularity}
    Let $\alpha = O(1)$, let $\e, \delta > 0$ be at most a sufficiently small constant, and let $S = \{ X_1, \ldots, X_n \} \subset \R^d$ be a set of $n$ independent draws from $\cN(\mu, I)$, where $\|\mu\|_2 \le \alpha$, and suppose that $n \geq \log (1 / \delta) / \e$.
    Then, with probability $1 - \delta$, $S$ is 
    \[
    \left(\e,  \e n \sqrt{d} \log (n / \delta), (\e n)^2 \log 1 / \e + \e n \sqrt{\e n d \log 1 / \e} \right) \mbox{-regular} \; .
    \]
\end{lemma}
\begin{proof}
    We prove that $S$ satisfies each bullet point in sequence.
To prove the first bullet point, 
    by Fact~\ref{fact:bounded-norm}, with probability $1-\delta/3$, for all $i \in [n]$ $\left|\|X_i\|^2 - d\right| \le 10 \left(\sqrt{\log(3 n/\delta) d} + \log(3 n/\delta)\right) \le 30 \sqrt{d} \cdot \log(n/\delta)$. Assuming this holds for all $i$, then for any subset $T \subset S$, $\sum_{i \in T} \|X_i\|^2 = |T| d \pm |T| \cdot 30 \sqrt{d} \log(n/\delta)$. Hence, if $|T| \le \eps n$, this equals $|T| d \pm O(\eps n \sqrt{d} \log (n/\delta))$, as desired.

    We now prove the second bullet point.
    Fix any $T$ satisfying $|T| \leq \e n$.
    Then, $\Sum(T) \sim \cN(|T| \mu, |T| I)$, so if we let $Z = |T|^{-1/2} \left( \Sum(T) - |T| \mu \right)$, we have that $Z \sim \cN(0, I)$, and 
    \[
    \norm{\Sum(T)}^2 = \|\mu\|^2 \cdot |T|^2 + |T|^{3/2} \iprod{\mu, Z} + |T| \norm{Z}^2 \; .
    \]
    Hence, we have that, for any $C > 0$, there exists $C' > 0$ so that
    \begin{align*}
    \Pr \left[ \abs{\iprod{\mu, Z}} \geq C' \alpha \sqrt{\e n \log 1 / \e} \right] & \leq \exp (-C \e n \log 1 / \e) \; , \; \mbox{and} \\
    \Pr \left[ \abs{ \norm{Z}^2 - d} > C' \sqrt{\e n d \log 1 / \e}  + C' \e n \log 1 / \e \right]  &\leq \exp (-C \e n \log 1 / \e) \; .
    \end{align*}
    The number of subsets $T \subset S$ of size at most $\eps n$ is at most $\exp\left(O(\eps n \log 1/\eps)\right)$.
    Therefore, by a union bound over all choices of $T$, with probability $1 - \exp (- \Omega (\e n \log 1 / \e)) = 1 - \delta / 3$, we have that
    \begin{equation}
\abs{\norm{\Sum(T)}^2 - (\|\mu\|^2 \cdot |T|^2 + |T| d)} \leq O (\e n \sqrt{\e n d \log 1 / \e} + (\e n)^2 \log 1 / \e)  \; ,
    \end{equation}
    for all $T$ with $|T| \leq \e n$.
    Since $\|\mu\|^2 \cdot |T|^2 \le \alpha^2 |T|^2 = O(\e n)^2$, this implies that
    \begin{equation}
        \abs{\norm{\Sum(T)}^2 - |T| d} \leq O (\e n \sqrt{\e n d \log 1 / \e} + (\e n)^2 \log 1 / \e) \; ,
    \end{equation}
    as claimed.
    Condition on this event holding for the rest of the proof.

    Finally, we prove the third bullet point.
    For any $i = 1, \ldots, n$, note that
    \begin{align*}
    \iprod{X_i, \Sum (S)} &= \norm{X_i}^2 + \iprod{X_i, \Sum (S \setminus \{i\})} \; .        
    \end{align*}
    The second term on the RHS is the inner product of two independent Gaussians, and hence is subexponential, with variance proxy $(n-1)d$.
    Therefore, with probability $1 - \delta / 3$, we have that $\abs{\iprod{X_i, \Sum (S \setminus \{i\})}} \leq \sqrt{n d} \log (n / \delta)$ for all $i = 1, \ldots, n$.
    Condition on this holding.
    Then, for any fixed $T$ satisfying $|T| \leq \e n$, we have that 
    \begin{align*}
        \left\langle \Sum (T), \Sum(S) \right\rangle &= \sum_{i \in T} \iprod{X_i, \Sum (S)} \\
        &= \sum_{i \in T} \norm{X_i}_2^2 + \sum_{i \in T} \iprod{X_i, \Sum (S \setminus \{i\})} \\
        &= d |T| \pm O \left( \e n \sqrt{d} \log (n / \delta) + \e n \sqrt{n d } \log (n / \delta) \right) \\
        &= d |T| \pm O \left( \e n \sqrt{n d } \log (n / \delta) \right) \; ,
    \end{align*}
    as claimed.
    Combining these bounds immediately yields the desired claim.
    \end{proof}
\noindent

Next, we note some simple consequences of regularity. For any subset $T \subset S$, we let $\bone_T$ denote the indicator vector of $T$, and $\bone = \bone_S$ denote the vector which is all $1$'s. We also define $X_T := \sum_{i \in T} X_i$.

\begin{proposition} \label{prop:inner-product-regularity}
    Suppose $S$ is $(2 \eps, \beta_1, \beta_2)$-regular. Then, for any sets $T, T'$ of size at most $\eps n$,
\[\langle X_T, X_{T'} \rangle = d \cdot |T \cap T'| \pm O(\beta_2).\]
\end{proposition}

\begin{proof}
    It is straightfoward to verify that
\[\langle X_T, X_{T'} \rangle = \frac{1}{2}\left[\|X_{T \cup T'}\|^2 + \|X_{T \cap T'}\|^2 - \|X_{T \backslash T'}\|^2 - \|X_{T' \backslash T}\|^2\right].\]
    Each of $T \cup T', T \cap T', T \backslash T', T' \backslash T$ have size at most $2 \eps n$. So, by regularity (Part ii of Definition~\ref{def:regularity}),
\[\langle X_T, X_{T'} \rangle = \frac{1}{2} \cdot d \cdot \left(|T \cup T'| + |T \cap T'| - |T \backslash T'| - |T' \backslash T|\right) \pm O(\beta_2)= d \cdot |T \cap T'| \pm O(\beta_2). \qedhere\]
\end{proof}

We note the following useful convexity lemma.

\begin{fact}\label{fact:convexity}
    Let $k \le n$ be a nonnegative integer, and $w \in [0, 1]^n$ be an $n$-dimensional vector with $\|w\|_1 \le k$. Then, $w$ is a convex combination of the points $\bone_T$ over $T \subseteq S, |T| \le k$.
\end{fact}
\begin{proof}
    Without loss of generality assume that the coordinates of $w$ are sorted in increasing order, i.e. $0 \leq w_1 \leq \ldots, w_n \leq 1$, and additionally define $w_0 = 0$ and $w_{n + k} = 1$ for all $k \geq 0$.
    For all integers $i \geq 0$, let $S_i = \{i, \ldots, i + k\} \cap [n]$, so that in particular $S_{n + j} = \emptyset$ for all $j > 0$.
    Recursively define weights $a_1, \ldots, a_n$ by $a_1 = w_1$, and $a_i = w_i - \sum_{j = i - k}^{i - 1} a_j$, where we set $a_j = 0$ for $j < 0$.
    Then by construction, we have that $w = \sum_{i = 1}^n a_i \bone_{S_i}.$
    We will show that $a_i \geq 0$ for all $i$, and that $\sum_{i = 1}^n a_i \leq 1$, from which the claim immediately follows.
    To prove the first claim, we proceed by induction. 
    Note that the base case is trivial, and moreover, if the claim is true for some $i < n$, then
    \[
    \sum_{j = i - k + 1}^i a_i = a_i + \sum_{j = i - k + 1}^{i - 1} a_i = w_i - a_{i - k} \leq w_{i} \leq w_{i + 1} \; ,
    \]
    so in particular $a_{i + 1} \geq 0$, which proves the induction.

    To prove the second claim, we simply observe that by nonnegativity, we have that
    \[
        k \geq \sum_{i = 1}^n w_i = \sum_{i = 1}^n \sum_{j = i - k}^i a_i \geq k \sum_{i = 1}^n \; , 
    \]
    where the last inequality follows from the fact that each $a_i$ appears at most $k$ times in the sum.
    Simplifying then immediately yields the claim.
\end{proof}

We use Fact~\ref{fact:convexity} to generalize Proposition~\ref{prop:inner-product-regularity} as follows.

\begin{proposition} \label{prop:weighted-inner-product-regularity}
    Suppose $S$ is $(2 \eps, \beta_1, \beta_2)$-regular, and $a, b \in [0, 1]^n$ are $n$-dimensional vectors such that $\sum a_i, \sum b_i \le \eps \cdot n$. Then, we have
\[\left\langle \sum_{i \in S} a_i X_i, \sum_{i \in S} b_j X_j \right\rangle = d \cdot \left(\sum_{i \in S} a_i b_i\right) \pm O(\beta_2) \hspace{0.5cm} \text{and} \hspace{0.5cm} \left\langle \sum_{i \in S} a_i X_i, X_S \right\rangle = d \cdot \left(\sum_{i \in S} a_i\right) \pm O(\sqrt{n} \cdot \beta_1).\]
\end{proposition}

\begin{proof}
    By Fact~\ref{fact:convexity}, we can write $a$ as a convex combination of $\bone_T$ over $T \subset S, |T| \le \eps n$. In other words, there exists a distribution $\mathcal{T}_1$ over $T_1 \subset S, |T_1| \le \eps n$ such that $a_i = \BP_{T_1 \sim \mathcal{T}_1} (i \in T_1)$. Likewise, there exists an (independent) distribution $\mathcal{T}_2$ over $T_2 \subset S, |T_2| \le \eps n$ such that $b_j = \BP_{T_2 \sim \mathcal{T}_2} (j \in T_2)$. Now,
\[\left\langle \sum a_i X_i, \sum b_j X_j \right\rangle = \BE_{T_1 \sim \mathcal{T}_1, T_2 \sim \mathcal{T}_2} \langle X_{T_1}, X_{T_2}\rangle.\]
    By Proposition~\ref{prop:inner-product-regularity}, 
\[\BE_{T_1 \sim \mathcal{T}_1, T_2 \sim \mathcal{T}_2} \langle X_{T_1}, X_{T_2}\rangle = \BE_{T_1 \sim \mathcal{T}_1, T_2 \sim \mathcal{T}_2} \left[d \cdot |T_1 \cap T_2|\right] \pm O(\beta_2).\]
    Next, the expectation of $|T_1 \cap T_2|,$ using linearity of expectation and independence of $\mathcal{T}_1, \mathcal{T}_2$, equals
\[\sum_{i \in S} \BP_{T_1 \sim \mathcal{T}_1, T_2 \sim \mathcal{T}_2} (i \in T_1 \cap T_2) = \sum_{i \in S} \BP_{T_1 \sim \mathcal{T}_1} (i \in T_1) \cdot \BP_{T_2 \sim \mathcal{T}_2} (i \in T_2) = \sum_{i \in S} a_i b_i.\]
    Overall, this implies that
\[\left\langle \sum a_i X_i, \sum b_j X_j \right\rangle = d \cdot \left(\sum_{i \in S} a_i b_i\right) \pm O(\beta_2).\]

    Next, by regularity (part iii of Definition~\ref{def:regularity}), we have that
\[\left\langle \sum a_i X_i, X_S \right\rangle = \BE_{T_1 \sim \mathcal{T}_1} \langle x_{T_1}, x_S \rangle = \BE_{T_1 \sim \mathcal{T}_1} \left(d \cdot |T_1|\right) \pm O(\sqrt{n} \cdot \beta_1) = d \cdot \left(\sum_{i \in S} a_i\right) \pm O(\sqrt{n} \cdot \beta_1). \qedhere\]
\end{proof}

\subsection{Filtering preliminaries}
Our algorithm for doing so will be based on the (soft) filter framework developed for robust estimation in other contexts.
Here we establish the notation and preliminaries we will require to design and analyze our algorithm.
We note that it will be convenient for us to use slightly nonstandard versions of the notation compared to the literature.

Our algorithm will assign weights to each point, that we will monotonically decrease over time.
For any $n$, let $\Gamma_n$ denote the set of valid weights:
\[
\Gamma_n = \{w \in \R^n: w_i \in [0, 1] \mbox{ for all $i = 1, \ldots, n$}\} \; .
\]
Recall that for any set $T \subseteq S$, $\bone_T \in \Gamma_n$ denotes the indicator vector for $T$, and $\bone = \bone_S$.

Let $K$ be some value to be specified later.
Given a set of points $S = \{ X_1, \ldots, X_n \}$, we associate it weight vectors $w^{(t)} \in \Gamma_n$, for $i = 1, \ldots, n$ and $t = 1,  \ldots, K$ where initially we set $w^{(1)} = \bone$.\footnote{As is common in this literature, for simplicity of notation we will conflate $S$ with the set of indices in $S$.}
For any such weight vector $w$, we let
\[
\Sum(w, S) = \sum_{i \in S} \sqrt{w_i} X_i \; , \mbox{and} \; M (w, S) = \sum_{i \in S} w_i X_i X_i^\top \; .
\]
When the context is clear, we will drop the $S$ from the notation for simplicity, i.e. we will let $\Sum (w) = \Sum (w, S)$.
For any set $T$, and for any set of weights $w$ on $S$, we let $w_T$ denote the set of weights restricted to the indices in $T \cap S$.
We also let $\Gram (w, S) = \Gram (w)$  be the $n \times n$ matrix given by
\[
\Gram (w)_{ij} = \sqrt{w_i w_j} \iprod{X_i, X_j} \; .
\]
Note that by design the nontrivial eigenvalues of $\Gram (w)$ and $M(w)$ are identical.

Recall that when the samples $S$ are an $\e$-corruption of $G$, this means that there are sets $B, R$ so that $|B| = |R| = \e n$ so that $R \subset G$ and so that $S = (G \setminus R) \cup B$.
For the remainder of this section, $S, G, B, R$ will always refer to these sets.
We define the following important set:
\[
\fS_n = \{ w \in \Gamma_n : \norm{\bone_{G \setminus R}- w_{G \setminus R}}_1 \leq 5 \norm{\bone_B - w_B}_1 \} \; ,
\]
that is, $\fS_n$ is the set of weights where we have removed at most five times as much weight from the good samples as we have removed from the the bad samples.

Finally, we will also seek to enforce regularity conditions on weighted subsets of points.
We will require the following natural generalization of Definition~\ref{def:regularity}:
\begin{definition}
    Let $S = \{X_1, \ldots, X_n\}$ be a set of points in $\R^d$, and let $w \in \Gamma_n$.
We say that $w$ is $(\e, \beta_1, \beta_2)$-regular if for all sets $T \subset S$ with $|T| \leq \e n$, we have:
\begin{enumerate}[(i)]
    \item $\sum_{i \in T} \norm{X_i}^2 = |T| d \pm O(\beta_1)$,
    \item $\norm{\Sum(w, T)}^2 = \norm{w_T}_1 d \pm O(\beta_2)$, and
    \item $\abs{\iprod{\Sum(w, T), \Sum (w, S)}} = \norm{w_T}_1 d \pm O(\sqrt{n} \beta_1)$.
\end{enumerate}
\end{definition}

The key fact we will use is the following:
\begin{lemma}
\label{lem:regular-weights-imply-tester}
Let $\alpha = O(1)$, let $\e, \delta \in [0, 1)$, and let $G$ be $(4 \e, \beta_1, \beta_2)$-regular, where
\[
\beta_1 \geq \e n \sqrt{d} \log (n / \delta), \beta_2 \geq \e n \sqrt{\e n d \log 1 / \e} + (\e n)^2 \log 1 / \e \; . 
\]
Further, assume that
    \[
    \norm{\Sum(G)}^2 = d n + \norm{\mu}^2 n^2 \pm O( \alpha n^{3/2} \sqrt{\log 1 / \delta} + n\sqrt{d} \log (1 / \delta)) \; .
    \]
Let $S$ be an $\e$-contamination of $G$, and let $w \in \Gamma_n$ be a set of weights on $w$ that satisfy $\norm{w}_1 \geq (1 - \e) n$, and $w$ is $(\e, \beta_1, \beta_2)$-regular.
Then, we have that
\[
\norm{\Sum(w, S)}^2 = d \norm{w}_1 + \norm{\mu}^2 n^2 \pm O(\alpha n^{3/2} \sqrt{\log 1 / \delta} + n\sqrt{d} \log (1 / \delta) + \sqrt{n} \beta_1 + \beta_2) \; . 
\]
In particular, if $n \geq C \cdot \tfrac{\sqrt{d} \cdot \log 1 / \delta}{\alpha^2}$ and $\sqrt{n} \beta_1, \beta_2 < \frac{1}{C} \cdot \alpha^2 n^2$ for $C$ sufficiently large, then:
\begin{itemize}
    \item if $\mu = 0$, then $\abs{\norm{\Sum (w, S)}^2 - d \norm{w}_1} \leq 0.4 \alpha^2 n ^2$ \; , and
    \item if $\norm{\mu} = \alpha$, then $\abs{\norm{\Sum (w, S)}^2 - d \norm{w}_1} \geq 0.7 \alpha^2 n^2$ \; .
\end{itemize}
In other words, the norm of the sum of the set of points distinguishes between the null and alternative hypotheses.
\end{lemma}
\begin{proof}
    First, we bound $\|\Sum(w, G \setminus R)\|^2$. Let $a$ be the vector that equals $1$ on the indices in $R$ and $1-\sqrt{w_i}$ on other indices, so that $\sum_{i \in G} a_i G_i + \Sum(w, G \setminus R) = \Sum(G)$. Then,
\[\|\Sum(w, G \setminus R)\|^2 = \|\Sum(G)\|^2 - 2 \bigg\langle \Sum(G), \sum_{i \in G} a_i G_i \bigg\rangle + \bigg\|\sum_{i \in G} a_i G_i\bigg\|^2.\]
    Let $\beta_3 := \alpha n^{3/2} \sqrt{\log 1 / \delta} + n\sqrt{d} \log (1 / \delta)$. Because $\|a\|_1 \le 2 \eps$ and $G$ is $(4 \eps, \beta_1, \beta_2)$ regular, Proposition~\ref{prop:weighted-inner-product-regularity} and our assumption on $\|\Sum(G)\|^2$ imply that
 \begin{align*}
        \norm{\Sum(w, G \setminus R)}^2 &= d n + \|\mu\|^2 n^2 \pm O(\beta_3) - 2 d \sum a_i \pm O(\sqrt{n} \beta_1) + d \sum a_i^2 \pm O(\beta_2)\\
        &= d \cdot \left(\sum (1-a_i)^2\right) + \|\mu\|^2 n^2 \pm O(\sqrt{n} \beta_1 + \beta_2 + \beta_3) \\
        &= d \norm{w_{G \setminus R}}_1 + \norm{\mu}^2 n^2 \pm O( \alpha n^{3/2} \sqrt{\log 1 / \delta} + n\sqrt{d} \log (1 / \delta) + \sqrt{n} \beta_1 + \beta_2) \; ,
    \end{align*}
    since $(1-a_i)^2 = w_i$ for $i \in G \setminus R$ and $0$ otherwise.
    
    Then, the regularity of $w$ implies that 
    \begin{align*}
        \norm{\Sum (w, S)}^2 &= \norm{\Sum (w, G \setminus R)}^2 + 2 \iprod{\Sum (w, G \setminus R), \Sum (w, B)} + \norm{\Sum (w, B)}^2 \\
        &= \norm{\Sum (w, G \setminus R)}^2 + 2 \iprod{\Sum (w, S), \Sum (w, B)} - \norm{\Sum(w, B)}^2 \\
        &= \norm{w}_1 d + \norm{\mu}^2 n^2 \pm O(\alpha n^{3/2} \sqrt{\log 1 / \delta} + n\sqrt{d} \log (1 / \delta) + \sqrt{n} \beta_1 + \beta_2) \; .
    \end{align*}
The second half of the claim then follows from straightforward calculations.
\end{proof}

\subsection{Additional preliminaries}
We first prove bounds on the eigenvalues of the Gram matrix of random Gaussian samples. To do so, we require properties about \emph{Wishart matrices}, which we now define.

\begin{definition}
    A Wishart matrix $W = \mathrm{W}_d(n)$ has distribution $W = H^\top H$, where $H \in \R^{n \times d}$ has every entry drawn as an i.i.d. standard Gaussian $\cN(0, 1)$. Note that $W \in \R^{d \times d}$, and $W$ is positive semidefinite.
\end{definition}

We will need the following concentration bound for the eigenvalues of a Wishart matrix.

\begin{lemma}(Follows from \cite[Theorem II.13]{wishartconcentration})
    If $W \sim \mathrm{W}_d(n),$ then with probability $1-\delta$,
\[\left\|S - n \cdot I\right\| \le O\left(\sqrt{nd} + \sqrt{n \log (1/\delta)} + d + \log (1/\delta)\right).\]
\end{lemma}

  We now provide eigenvalue bounds for samples drawn from $\cN(\mu, I)$.
\begin{fact}
\label{fact:gram-conc}
  Let $\mu \in \R^d$ have $\|\mu\| \le \alpha$ and let $\delta > 0$.
  If $n \leq d$ and $X_1, \dots, X_n \overset{i.i.d.}{\sim} \cN(\mu, I)$, then with probability at least $1-\delta$, we have 
\[\|\Gram(\{X_1, \dots, X_n\}) - d \cdot I\| \leq O(\max(\sqrt{n d}, \sqrt{d  \log(1/\delta)}, \log(1/\delta), \alpha^2 n)).\]
  If $n \geq d$ and $X_1, \dots, X_n \overset{i.i.d.}{\sim} \cN(\mu, I)$, then with probability at least $1-\delta$, we have 
\[\bigg\|\sum_{i \in [n]} X_i X_i^\top - n \cdot I\bigg\| \leq O(\max(\sqrt{n d}, \sqrt{n  \log(1/\delta)}, \log(1/\delta), \alpha^2 n)).\]    
\end{fact}
\begin{proof}
    First, note that when $n \le d$, the nonzero (top $n$) eigenvalues of $\sum X_i X_i^\top$ match the eigenvalues of $\Gram(\{X_1, \dots, X_n\})$. So, in the $n \le d$ case we can focus on the top $n$ eigenvalues of $\sum_{i \in [n]} X_i X_i^\top$. This will allow us to consolidate calculations for both the $n \le d$ and $n \ge d$ case.

  Let $Y_i := X_i-\mu$. We can write
  \[
  \sum_{i \leq n} X_i X_i^\top = \sum_{i \leq n} Y_i Y_i^\top + \sum_{i \leq n} Y_i \mu^\top + \sum_{i \leq n} \mu Y_i^\top + n \cdot \mu \mu^\top\mper
  \]
  Note that $\sum_{i \leq n} Y_i Y_i^\top \sim \mathrm{W}_d(n)$, so with probability at least $1-\delta$,
\[\bigg\|\sum_{i \in [n]} Y_i Y_i^\top - n \cdot I\bigg\| \le O\left(\sqrt{nd} + \sqrt{n \log (1/\delta)} + d + \log(1/\delta)\right).\]
    In the $n \le d$ case, we note that $\sum Y_i Y_i^\top$ has the same nonzero eigenvalues as $\Gram(\{Y_1, \dots, Y_n\}) \sim W_n(d)$. So, with probability at least $1-\delta$, the top $n$ eigenvalues of $\sum Y_i Y_i^\top$ are in the range
\[d \pm O\left(\sqrt{nd} + \sqrt{d \log(1/\delta)} + n + \log(1/\delta)\right).\]
    
  Next, $\sum_{i \leq n} Y_i \mu^\top$ is a rank-1 matrix with operator norm $\|\sum_{i \leq n} Y_i \| \cdot \|\mu\| \le \alpha \cdot \|\sum_{i \leq n} Y_i\|$. Since $\sum_{i \leq n} Y_i \sim \cN(0, n \cdot I)$, with probability at least $1-\delta$ it has norm at most $O(\sqrt{nd + n\log(1/\delta)})$, which means $\|\sum_{i \leq n} Y_i \mu^\top\| \le O(\alpha \sqrt{nd + n\log(1/\delta)}) \le O(\sqrt{nd} + \sqrt{n \log(1/\delta)})$. The same bound holds for $\sum_{i \leq n} \mu Y_i^\top$. 
  Finally, $\|n \mu \mu^\top \| = n \cdot \|\mu\| \cdot \|\mu^\top\| \le n \alpha^2$.
  
  In the $n \ge d$ case, adding the bounds together completes the proof. In the $n \le d$ case, adding the bounds together tells us the top $n$ eigenvalues of $\sum X_i X_i^\top$ are in the desired range, which completes the proof.
\end{proof}

The next fact we will need is a direct corollary of Lemma 4.1 in~\cite{dong2019quantum}.
\begin{fact}
\label{fact:small-subset-deviations}
    Let $\alpha = O(1)$, let $\mu \in \R^d$ have $\norm{\mu} \le \alpha$, and let $\e, \delta > 0$.
    Suppose that $n \geq \Omega (1 / \e)$.
    Then, there is some universal constant $c > 0$ so that with probability $1 - \delta$, we have that for all $v$ with $\norm{v} = 1$, and all $w$ supported on $G$ with $\norm{w}_1 \leq 10 \e n$, it holds that
    \[
    \sum_{i \in G} w_i \iprod{v, X_i}^2 \leq c \cdot (\e n \log 1 / \e + d + \log 1 / \delta) \; .
    \]
\end{fact}
Finally, we will require the following downweighting scheme:
\begin{fact}
    \label{fact:1d-filter}
    Let $w \in \fS_n$, and let $\tau_1, \ldots, \tau_n$ be a set of nonnegative scores satisfying $\sum_{i \in G \setminus R} w_i \tau_i < 5 \sum_{i \in B} w_i \tau_i$.
    Let $w' \in \Gamma_n$ be defined by 
    \[
    w'_i = \left(1 - \frac{\tau_i}{\max_{i \in S} \tau_i} \right) w_i \; .
    \]
    Then $\supp{w'} \subset \supp{w}$, and moreover $w' \in \fS_n$.
\end{fact}

\subsection{The filtering algorithm for $n \leq d$}
In this case, the filtering algorithm proceeds as follows.
Let $\delta > 0$, and let 
\begin{equation}
\label{eq:gamma2}
    \gamma_2 := C \left(\sqrt{nd} + \alpha^2 n + \sqrt{(n+d) \log (1/\delta)} + \log (1 / \delta) + \e n \log 1 / \e\right) \; ,
\end{equation}
for some constant $C$ sufficiently large.
Initialize weights $w^{(1)} = \bone$.
Then, for $t = 1$ until termination, we proceed as follows.
For any $w \in \Gamma_n$, let $D(w) = d \cdot \mathrm{diag} (w)$.
Let $\lambda$ denote the top singular value of $\Gram (w, S) - D(w)$, and let $v$ be its associated singular unit vector (if there are multiple, choose any).
If $\lambda < 5 \gamma_2$, then terminate.
Otherwise, for all $i \in S$, let $\tau_i = \frac{v_i^2}{w_i^{(t)}} \mathbb{I}[w^{(t)}_i > 0]$ (where $\tau_i$ defaults to $0$ when $w_i^{(t)} = 0$), and proceed to sort the samples in decreasing order of $\tau_i$.
Then, define $w^{(t + 1)}$ by 
\[
w^{(t + 1)}_i = \left( 1  - \frac{\tau_i}{\max_i \tau_i}\right) w^{(t)}_i \; .
\]
The formal pseudocode for this algorithm appears in Algorithm~\ref{alg:spectralfilter1}.

\begin{algorithm}[H]

\caption{Spectral filtering for $n \leq d$. Input: $X_1,\ldots,X_n \in \R^d$, $\gamma_2 > 0$.}
\label{alg:spectralfilter1}

    \begin{algorithmic}[1]
\State Let $w^{(1)} = \bone$, and let $t = 1$
\While{$\norm{\Gram (w^{(t)}, S) - D(w^{(t)})} \geq 5 \gamma_2$}
    \State Let $v$ be the top singular vector of $\Gram (w^{(t)}, S) - D(w^{(t)})$
    \State For all $i$, let $\tau_i = \frac{v_i^2}{w_i^{(t)}} \mathbb{I}[w^{(t)}_i > 0]$\Comment{If $w_i^{(t)} = 0$, we set $\tau_i = 0$.}
    \State Let
    \[
    w^{(t + 1)}_i = \left( 1  - \frac{\tau_i}{\max_i \tau_i}\right) w^{(t)}_i \; .
    \]
    \State Let $t \gets t + 1$
\EndWhile
\State \textbf{Return} $w^{(t)}$
\end{algorithmic}

\end{algorithm}

For the rest of the section, let us assume that $G$ is $(\e, \e n \sqrt{d} \log (n / \delta), (\e n)^2 \log 1 / \e + \e n \sqrt{\e n d \log 1 / \e})$-regular, and additionally assume that
\begin{equation}
\label{eq:gram-assume}
\norm{\Gram(G) - d I} \leq \frac{\gamma_2}{10} \; .
\end{equation}
By Lemma~\ref{lem:good-set-regularity} and Fact~\ref{fact:gram-conc}, these two conditions hold together with probability $1 - \delta$.
Then, our main claim for this algorithm is the following:
\begin{lemma}
\label{lem:spectralfilter1}
    Under the above assumptions, Algorithm~\ref{alg:spectralfilter1} terminates in $K$ iterations for some $K \leq 6 \e n$, runs in time $O(d n^2)$ per iteration, and moreover, at termination, we have that
    \begin{itemize}
        \item $w^{(K)} \in \fS_n$, and
        \item for all $T \subset S$ with $|T| \leq \e n$, we have that
        \[
            \norm{\Sum(w^{(K)}, T)}^2 = \norm{w^{(K)}(T)}_1 d \pm O(\e n \cdot \gamma_2) .
        \]
    \end{itemize}
\end{lemma}
\begin{proof}
    The runtime per iteration is clearly dominated by the time it takes to find the top singular vector of the centered gram matrix, which can be done in time $O(d n^2)$.
    
    We will show that for all $t = 1, \ldots, K$, we have that $w^{(t)} \in \fS_n$.
    First, we demonstrate how this proves the overall lemma.
    First, note that after each iteration, some new $w_i$ (with the maximum $\tau_i$) becomes $0$, so after $6 \eps n$ iterations, we have removed at least $6 \eps n$ mass from $w$. By definition of $\fS_n$, this means we have removed at least $\e n$ mass from the bad coordinates $w_B$, at which point no further updates can maintain the invariant that $w^{(i)} \in \fS_n$.
    
    Next, we observe that if $w^{(K)} \in \fS_n$, then since we terminated, we must have that
    \[\norm{\Gram (w^{(K)}) - D(w^{(K)})} \leq 5 \gamma_2 \; .\]
    But then, for all $T$ with $|T| \leq \e n$, let $\bone_T \in \R^n$ be the indicator vector $T$.
    Then, we have that 
    \begin{equation*}
        \norm{\Sum(w^{(K)}, T)}^2 = \bone_T^\top \Gram(w^{(K)}) \bone_T = \norm{w^{(K)}_T}_1 d \pm O(\e n \cdot \gamma_2) \; ,
    \end{equation*}
    as claimed.

    Thus, it suffices to prove the invariant that $w^{(t)} \in \fS_n$ for all $t = 1, \ldots, K$.
    We proceed by induction.
    Clearly $w^{(1)} \in \fS_n$.
    Now, suppose $w^{(t)} \in \fS_n$ for some $t < K$.
    Since we have not yet terminated, this implies that 
    \[
    \lambda = \norm{\Gram(w^{(t)}) - D(w^{(t)})} \ge 5 \gamma_2 \; .
    \]
    But by~\eqref{eq:gram-assume} and the Cauchy interlacing theorem, we have that
    \begin{align*}
        \norm{\Gram (w^{(t)}, G \setminus R) - D (w^{(t)}_{G \setminus R})} \leq \frac{\gamma_2}{10} \le \frac{\lambda}{50} \; .
    \end{align*}
    We claim that this implies that $5 \sum_{i \in B} v_i^2 > \sum_{i \in G \setminus R} v_i^2$.
    Indeed, suppose not, and let $v_G$ denote the restriction of $v$ onto the coordinates in $G \setminus R$, and let $v_B$ denote the restriction of $v$ onto the coordinates in $B$.
    This means that
    \begin{align*}
        \left| v^\top \left(\Gram (w^{(t)}) - D(w) \right) v \right| &= \left| v_G^\top \left( \Gram (w^{(t)}, G \setminus R) - D(w^{(t)}_{G \setminus R}) \right) v_G \right. \\
        &\left. +2 v_G^\top \left(\Gram (w^{(t)}) - D(w^{(t)})\right) v_B + v_B^\top \left( \Gram(w^{(t)}, B) - D(w^{(t)}_B) \right) v_B \right| \\
        &\leq \norm{v_G}^2 \cdot \frac{\lambda}{50} + 2 \norm{v_G} \norm{v_B} \lambda + \norm{v_B}^2 \lambda \leq 0.96 \lambda \; ,
    \end{align*}
    where the last inequality holds because $\|v_B\|^2+\|v_G\|^2 = \|v\|^2 = 1$ and $\|v_B\|^2 \le \frac{1}{6}.$
    But this is a contradiction since $v$ is the top singular vector of the centered Gram matrix.
    Therefore, by Fact~\ref{fact:1d-filter} and the definition of $\tau_i$, we obtain that $w^{(t + 1)} \in \fS_n$, as claimed.
    (Note that if $w_i^{(t)} = 0$, the $i$th row and column of both $\Gram(w^{(t)}, S)$ and $D(w^{(t)})$ are $0$ so $v_i = 0$, which means $w_i^{(t)} \cdot \tau_i = v_i^2$ even if $w_i^{(t)} = 0$.)
    This completes the proof.
\end{proof}
 
\subsection{The filtering algorithm for $n > d$}
The filtering algorithm proceeds similarly to above.
Let $\gamma_2$ be as in~\eqref{eq:gamma2}.
Initialize weights $w^{(1)} = \bone$.
Then, for $t =1$ until termination, we proceed as follows. 
Let $\lambda$ be the top singular value of $M(w^{(t)}) - n I$, and let $v$ be its associated singular value (if there are multiple, again choose one arbitrarily).
If $\lambda < 5 \gamma_2$, then terminate.
Otherwise, for all $i \in G$, let $\tau_i = \iprod{v, X_i}^2 \mathbb{I}[w_i > 0]$.
Proceed to sort the samples in decreasing order of $\tau_t$.
As before, by relabeling indices, assume that $\tau_1 \geq \tau_2 \geq \dots \geq \tau_n$.
Let $I$ be the smallest index so that $\sum_{i \leq I} w_i^{(t)} \geq 2 \e n$, and define $w^{(t + 1)}$ by 
\begin{equation}
\label{eq:filter-step}
w^{(t + 1)}_i = \left\{ \begin{array}{ll}
         \left( 1 - \frac{\tau_i}{\tau_1} \right) w^{(t)}_i & \mbox{if $i \leq I$};\\
        w^{(t)}_i & \mbox{if $i > I$}.\end{array} \right. 
\end{equation}
The formal pseudocode for the algorithm appears in Algorithm~\ref{alg:spectralfilter2}.

\begin{algorithm}[H]

\caption{Spectral filtering for $n > d$. Input: $X_1,\ldots,X_n \in \R^d$, $\gamma_2 > 0$.}
\label{alg:spectralfilter2}

    \begin{algorithmic}[1]
\State Let $w^{(1)} = \bone$, and let $t = 1$
\While{$\norm{M (w, S) - n I} \geq 5 \gamma_2$}
    \State Let $v$ be the top singular vector of $M(w, S) - n I$
    \State For all $i$, let $\tau_i = \iprod{v, X_i}^2 \mathbb{I}[w^{(t)}_i > 0]$
    \State Let $w^{(t + 1)}$ be given by~\eqref{eq:filter-step} 
    \State Let $t \gets t + 1$
\EndWhile
\State \textbf{Return} $w^{(t)}$
\end{algorithmic}

\end{algorithm}

As before, for the rest of this section, let us assume that $G$ is $(\e, \e n \sqrt{d} \log (n / \delta), (\e n)^2 \log 1 / \e + \e n \sqrt{\e n d \log 1 / \e})$-regular, and additionally assume that
\begin{equation}
    \label{eq:cov-assume}
    \norm{M (G) - n I} \leq \frac{\gamma_2}{10} \; ,
\end{equation}
and that Fact~\ref{fact:small-subset-deviations} holds.
As before, direct applications of Lemma~\ref{lem:good-set-regularity} and Facts~\ref{fact:gram-conc} and~\ref{fact:small-subset-deviations} immediately imply that these conditions hold together with probability at least $1 - \delta$.
Then, our main claim for this algorithm is the following:
\begin{lemma}
\label{lem:spectralfilter2}
    Under the above regularity conditions, Algorithm~\ref{alg:spectralfilter2} terminates in $K$ iterations for some $K \leq 6 \e n$, runs in time $O(nd^2)$ per iteration, and moreover, at termination, we have that
    \begin{itemize}
        \item $w^{(K)} \in \fS_n$, and
        \item for all $T \subset S$ with $|T| \leq \e n$, we have that
        \[
        \norm{\Sum (w^{(K)}, T)}^2 \leq 10 \gamma_2 \cdot \e n.
        \]
    \end{itemize}
\end{lemma}
\begin{proof}
    As before, the per-iteration runtime is dominated by the runtime of PCA, which is $O(n d^2)$.

    We will again inductively show that for all iterations $t$, we have that $w^{(t)} \in \fS_n$.
    We first show how to prove the lemma, assuming this claim.
    In this case, the number of iterations $K$ can be bounded identically as in Lemma~\ref{lem:spectralfilter1}.
Moreover, by construction, at termination we have that $\norm{M(w^{(K)}) - n I} \le 5 \gamma_2$.
    Now suppose that there was some subset $T$ with $|T| \leq \e n$ that had
    \[
    \norm{\Sum (w^{(K)}, T)}^2 > 10 \gamma_2 \cdot \e n.
    \]
    Then, there is a unit vector $v \in \R^d$ so that
\begin{equation} \label{eq:unit-v-big}
    \sum_{i \in T} (w^{(K)}_i)^{1/2} \iprod{v, X_i} > \sqrt{10 \gamma_2 \cdot \e n} \; .
\end{equation}
    Thus, we have that
\begin{equation} \label{eq:v-bound}
    \sum_{i \in T \cup B} w_i^{(K)} \langle v, X_i \rangle^2 \ge \sum_{i \in T} w^{(K)}_i \iprod{v, X_i}^2 \geq 10 \gamma_2 \; .
\end{equation}
    Above, the first inequality holds because every $w_i^{(K)}$ is nonnegative, and the second inequality follows from \eqref{eq:unit-v-big} and the Cauchy-Schwarz inequality.
    However, as $|R \cup T| \le 2 \eps$ and $\|\bone-w^{(K)}\|_1\le 6 \eps n$, \eqref{eq:cov-assume} and Fact~\ref{fact:small-subset-deviations} together imply
\begin{equation} \label{eq:M-RT-bound}
    \norm{ M \left(w^{(K)}, G \setminus (R \cup T)\right) - n I } \leq \frac{\gamma_2}{5} \; . 
\end{equation}
    Since $G \backslash (R \cup T) = S \backslash (B \cup T)$, the inequalities \eqref{eq:v-bound} and \eqref{eq:M-RT-bound} together imply that
    \[
    \norm{ M \left(w^{(K)}, S\right) - n I} > 5 \gamma_2 \; ,
    \]
    which is a contradiction.

    Thus, as before, it suffices to prove that $w^{(t)} \in \fS_n$ for all iterations until termination.
    We will do so inductively.
    As before, the base case $t = 1$ is trivial.
    Now suppose that $w^{(t)} \in \fS_n$ for some $t < K$.
    Since we have not yet terminated, this means that $\norm{M (w^{(t)}, S) - n I} > 5 \gamma_2$.
    Then,~\eqref{eq:cov-assume} and Fact~\ref{fact:small-subset-deviations} together immediately imply that
\begin{equation} \label{eq:B_bound}
    \sum_{i \in B} w_i^{(t)} \iprod{v, X_i}^2 \geq 3 \gamma_2 \; ,
\end{equation}
    for $v$ the top eigenvector of $M(w^{(t)}, S)-n \cdot I$. On the other hand, since $\sum_{i \le I} w_i^{(t)} \le 2 \eps n + 1$ and since $w^{(t)} \in \fS_n$ means we have removed at most $6 \eps n$ mass from all samples, this means $I \le 8 \eps n + 1 \le 10 \eps n$. So, Fact~\ref{fact:small-subset-deviations} implies that 
\begin{equation} \label{eq:G_bound}
    \sum_{i \leq I, i \in G} w_i^{(t)} \iprod{v, X_i}^2 < \gamma_2 \; .
\end{equation}
    By definition of $I$, every $\langle v, X_i \rangle^2$ for $i \le I$ is larger than every $\langle v, X_i \rangle^2$ for $i \in B \backslash [I]$. Therefore, since $\sum_{i \in B \backslash [I]} w_i^{(t)} \le |B \backslash [I]| \le \eps n$ but $\sum_{i \le I} w_i^{(t)} \ge 2 \eps n$, we have
\begin{equation} \label{eq:BI_bound}
    2 \cdot \left(\sum_{i \in B} w_i^{(t)} \langle v, X_i \rangle^2 - \sum_{i \le I} w_i^{(t)} \langle v, X_i \rangle^2\right) \le 2 \cdot \sum_{i \in B \backslash [I]} w_i^{(t)} \langle v, X_i \rangle^2 \le \sum_{i \in I} w_i^{(t)} \langle v, X_i \rangle^2.
\end{equation}
    Along with \eqref{eq:B_bound}, \eqref{eq:BI_bound} implies that
\begin{equation} \label{eq:I_bound}
    \sum_{i \le I} w_i^{(t)} \langle v, X_i \rangle^2 \ge 2 \gamma_2.
\end{equation}
    Hence, by combining \eqref{eq:I_bound} with \eqref{eq:G_bound}, we have
    \[
    \sum_{i \leq I, i \in B} w_i^{(t)} \iprod{v, X_i}^2 \geq \gamma_2 \geq \sum_{i \leq I, i \in G} w_i^{(t)} \iprod{v, X_i}^2 \; ,
    \]
    and so the result for $w^{(t+1)}$ immediately follows from Fact~\ref{fact:1d-filter}.
\end{proof}

\subsection{Bounding row sums}
We now have a way to ensure that small subsets of points have means with bounded norm.
We also need to enforce that row sums are bounded.
To do so, we will simply remove the set of $O(\e n)$ points whose row sums have largest deviation from what we expect.
More formally, given a set of weights $w \in \fS_n$, we will let 
\begin{equation}
\label{eq:rowsumscore}
\tau_i = \abs{\iprod{\sqrt{w_i} X_i, \sum_{j \in S} \sqrt{w_j} X_j} - w_i d} \cdot \mathbb{I} [w_i > 0] \; .
\end{equation}
We then sort the indices in decreasing order by $\tau_i$.
Again for simplicity of notation, assume that after some suitable reindexing we have that $\tau_1 \geq \tau_2 \geq \ldots \geq \tau_n$.
Then, we replace $w_i$ with $0$ for all $i \le \eps n$.
We give the formal pseudocode for this algorithm in Algorithm~\ref{alg:rowsumfilter}.

\begin{algorithm}[H]

\caption{Bounding row sums. Input: $X_1,\ldots,X_n \in \R^d$}
\label{alg:rowsumfilter}

    \begin{algorithmic}[1]
        \State For all $i$, let $\tau_i$ be as in~\eqref{eq:rowsumscore}.
        \State Sort the indices in decreasing order by $\tau_i$.
        \Comment{By relabeling indices, for simplicity of notation assume that the $i$'s are initally sorted}
\State Set $w_i = 0$ for all $i \leq \eps n$.
        \State \textbf{return} $w$
    \end{algorithmic}

\end{algorithm}

\begin{lemma}
\label{lem:rowsumfilter-1}
    Assume $G$ is $(12 \eps, \beta_1, \beta_2)$-regular, that Fact~\ref{fact:small-subset-deviations} holds for $G$, and $S = (G \backslash R) \cup B$, where $|R| = |B| = \eps n$.
    Let $w \in \fS_n$, and assume that for all $T \subset S$ with $|T| \leq 2 \e n$, we have that $\norm{\Sum(w, T)}^2 = \norm{w_T}_1 d \pm O( \beta_2 )$.
Then for all $T \subset S \backslash B$ with $|T| \leq \e n$, we have that 
\[
\sum_{i \in T, j \in S} \sqrt{w_i w_j} \iprod{X_i, X_j} = d \cdot \|w_T\|_1 \pm O(\sqrt{n} \beta_1 + \beta_2) \; .
\] 
\end{lemma}

\begin{proof}
    Fix $T \in S \backslash B$ with $|T|\le \eps n$. Since $S = (G \backslash R) \cup B$, we can write
\begin{align*}
    \sum_{i \in T, j \in S} \sqrt{w_i w_j} \iprod{X_i, X_j} &= \underbrace{\sum_{i \in T, j \in G} \sqrt{w_i w_j} \iprod{X_i, X_j}}_{A_1} - \underbrace{\sum_{i \in T, j \in R} \sqrt{w_i w_j} \iprod{X_i, X_j}}_{A_2} + \underbrace{\sum_{i \in T, j \in B} \sqrt{w_i w_j} \iprod{X_i, X_j}}_{A_3} \, .
\end{align*}    
    Above, for $j \in R$, $w_j$ is defined to equal $w_{j'}$ for the $j' \in B$ that replaces $j$.

    Let $b_i := \sqrt{w_i}$ and $a_i := 1-\sqrt{w_i}$. Since $|T| \le \eps n$, we know that $\sum_{i \in T} b_i \le \eps n$. Moreover, $\sum_{j \in G} w_j \ge (1-6 \eps) n$, so $\sum_{j \in G} a_i \le 6 \eps n.$ Because $G$ is $(12 \eps, \beta_1, \beta_2)$-regular and $T \subset G$, Proposition~\ref{prop:weighted-inner-product-regularity} implies that
\begin{align}
    A_1 &= \left\langle \sum_{i \in T} b_i X_i, \sum_{j \in G} X_j \right\rangle - \left\langle \sum_{i \in T} b_i X_i, \sum_{j \in G} a_j X_j \right\rangle \nonumber \\
    &= d \cdot \sum_{i \in T} b_i - d \sum_{i \in T} a_i b_i \pm O(\sqrt{n} \beta_1 + \beta_2) \nonumber \\
    &= d \cdot \sum_{i \in T} w_i \pm O(\sqrt{n} \beta_1 + \beta_2). \label{eq:a1-bound}
\end{align}

    Next, we bound $A_2$. Since $T, R \subset G$ are disjoint and $|T|, |R| \le \eps n$, Proposition~\ref{prop:weighted-inner-product-regularity} implies that
\begin{align}
    A_2
    &= \left\langle \sum_{i \in T} b_i X_i, \sum_{j \in R} b_j X_j \right\rangle = \pm O(\sqrt{n} \beta_1 + \beta_2), \label{eq:a2-bound}
\end{align}
    since the $b_i$ terms in $T$ and the $b_j$ terms in $R$ are from disjoint sets.

    Finally, we bound $A_3$. We will only use the fact that $T, B$ are disjoint sets in $S$ of size at most $\eps n$ and our assumption on $w$ in the lemma statement. We can write
\begin{align}
    A_3 &= \frac{1}{2} \left(\left\|\Sum(w, T \cup B)\right\|^2 - \left\|\Sum(w, T)\right\|^2 - \left\|\Sum(w, B)\right\|^2\right) \nonumber \\
    &= \frac{1}{2} \left(\|w_{T \cup B}\|_1 \cdot d - \|w_T\|_1 \cdot d - \|w_B\|_1 \cdot d\right) \pm O(\beta_2) \nonumber \\
    &= \pm O(\beta_2). \label{eq:a3-bound}
\end{align}
    Adding \eqref{eq:a1-bound}, \eqref{eq:a2-bound}, and \eqref{eq:a3-bound} completes the proof.
\end{proof}

\begin{lemma}
\label{lem:rowsumfilter}
    Assume Lemma~\ref{lem:rowsumfilter-1} holds, and let $w'$ be the output of Algorithm~\ref{alg:rowsumfilter}. Then, for all $T \subset S$ with $|T| \leq \e n$, we have that 
\[
\sum_{i \in T, j \in [n]} \sqrt{w_i' w_j'} \iprod{X_i, X_j} = d \cdot \|w_T\|_1 \pm O(\sqrt{n} \beta_1 + \beta_2) \; .
\]
\end{lemma}

\begin{proof}
    Recall the definition of $\tau_i$ from \eqref{eq:rowsumscore}. First, we note that for any $T \subset S \backslash B$ with $|T| \le \eps n$, $\sum_{i \in T} \tau_i \le O(\sqrt{n} \beta_1 + \beta_2)$. To see why, we can split $T$ into $T^+$ and $T^-$, where $i \in T^+$ if $\langle \sqrt{w_i} X_i, \sum_{j \in S} \sqrt{w_j} X_j \rangle \ge w_i d$ and $i \in T^-$ otherwise. Then, since $|T^+|, |T^-| \le \eps n$, Lemma~\ref{lem:rowsumfilter-1} implies that both $\sum_{i \in T^+} \tau_i$ and $\sum_{i \in T^-} \tau_i$ are at most $O(\sqrt{n} \beta_1 + \beta_2)$.

    Since $\sum_{i \in T} \tau_i \le O(\sqrt{n} \beta_1 + \beta_2)$ for any subset $T$ of $S \backslash B$ of size at most $\eps n$, and since we sorted the $\tau_i$'s in decreasing order, this implies $\sum_{i \in T} \tau_i \le O(\sqrt{n} \beta_1 + \beta_2)$ for any subset $T$ of $S \backslash [\eps n]$ of size at most $\eps n$. If $w_i$ represents the values of $w$ before setting the top $\eps n$ indices to $0$, and $w_i'$ represents the values of $w$ afterwards (i.e., $w_i' = 0$ for $i \le \eps n$ and $w_i' = w_i$ for $i > \eps n$), then 
\begin{align}
    \left|\sum_{i \in T, j \in S} \sqrt{w_i' w_j} \langle X_i, X_j \rangle - d \cdot \|w'_T\|_1\right| 
    &\le \sum_{i \in T} \left|\left\langle \sqrt{w_i'} X_i, \sum_{j \in S} \sqrt{w_j} X_j \right\rangle - w_i' d\right| \nonumber \\
    &= \sum_{i \in T \backslash [\eps n]} \tau_i \le O(\sqrt{n} \beta_1 + \beta_2). \label{eq:w'-bound-1}
\end{align}
    Next, we have
\begin{align}
    \sum_{i \in T, j \in [\eps n]} \sqrt{w_i' w_j} \langle X_i, X_j \rangle &= \sum_{i \in T \backslash [\eps n], j \in [\eps n]} \sqrt{w_i w_j} \langle X_i, X_j \rangle = \pm O(\beta_2), \label{eq:w'-bound-2}
\end{align}
    by the same argument as in \eqref{eq:a3-bound}, since $T \backslash [\eps n]$ and $[\eps n]$ are disjoint sets in $S$ and have size at most $\eps n$.

    By subtracting \eqref{eq:w'-bound-2} from \eqref{eq:w'-bound-1}, we obtain the desired bound
\[\left|\sum_{i \in T, j \in S} \sqrt{w_i' w_j'} \langle X_i, X_j \rangle - d \|w'_T\|_1\right| = \left|\sum_{i \in T, j \in S \backslash [\eps n]} \sqrt{w_i' w_j} \langle X_i, X_j \rangle - d \|w'_T\|_1\right| \le O(\sqrt{n} \beta_1 + \beta_2). \qedhere\]
\end{proof}

\subsection{Putting it all together}
We now have all the necessary pieces for the full algorithm.
It will proceed as follows.
First, remove any points whose norms differ from $d$ by too much.
Then, run the appropriate spectral filtering algorithm depending on whether or not $n \leq d$.
Next, use Algorithm~\ref{alg:rowsumfilter} to bound row sums.
Finally, check the sum of all entries in the (centered) Gram matrix, and if it is too large, then output \textsc{YES}, otherwise output $\textsc{NO}$.
The full pseudocode is given in Algorithm~\ref{alg:efficient-mean-tester}.

\begin{algorithm}[H]

\caption{Robust mean testing $X_1,\ldots,X_n \in \R^d$}
\label{alg:efficient-mean-tester}
    \begin{algorithmic}[1]
        \State Remove any $i$ satisfying $\abs{\norm{X_i}^2 - d} \geq O(\sqrt{d} \log n / \delta)$
        \State Let $\gamma_2$ be as in~\eqref{eq:gamma2}.
        \If{$n \leq d$}
            \State Let $w$ be the output of Algorithm~\ref{alg:spectralfilter1} with parameter $\gamma_2$.
        \Else 
            \State Let $w$ be the output of Algorithm~\ref{alg:spectralfilter2} with parameter $\gamma_2$.
        \EndIf 
        \State Let $w'$ be the output of Algorithm~\ref{alg:rowsumfilter} with input $w$
         \If{$\abs{\norm{\Sum (w', S)}^2 - d \norm{w}_1} \geq 0.7 \alpha^2 n^2$} 
            \State \textbf{return} $\textsc{YES}$ 
        \Else
            \State \textbf{return} $\textsc{NO}$
        \EndIf
    \end{algorithmic}
\end{algorithm}

\begin{proof}[Proof of Theorem~\ref{thm:poly-time-main}]
    The runtime of the algorithm is clearly dominated by the runtime of the spectral filters, which both run in time $O(\e n^2 d \min (n, d))$, and the runtime of computing $\Sum (w, S)$, which is $O(n d)$.
    We now prove correctness.

    By standard arguments, the first step (Line 1) of Algorithm~\ref{alg:efficient-mean-tester} removes no uncorrupted points with probability $1 - \delta / 3$.

    Assume that $G$ is $(12 \e, \beta_1, \beta_2)$-regular for $\beta_1 = \e n \sqrt{d} \log (n / \delta)$ and $\beta_2 = \e n \sqrt{n d} \log (n / \delta) + (\e n)^2 \log 1 / \e$.
    By Lemma~\ref{lem:good-set-regularity}, this occurs with probability $1 - \delta / 3$.
    As argued above, for our choice of $n$, the conditions for Lemma~\ref{lem:spectralfilter1} and Lemma~\ref{lem:spectralfilter2} are also satisfied with probability $1 - \delta / 3$ (and will even hold for all $|T| \le 2 \e n$).
    Thus, we obtain that $w \in \fS_n$.
    Finally, regularity and standard sub-exponential moment bounds imply that the conditions for Lemma~\ref{lem:rowsumfilter} are satisfied, as long as $\beta_1 = \eps n \sqrt{d} \log(n/\delta)$ and $\beta_2 = O(\eps n \cdot \gamma_2) = O(\eps n \sqrt{nd} + \alpha^2 \eps n^2 + \eps n \sqrt{(n+d) \log(1/\delta)} + \eps n \log(1/\delta) + (\eps n)^2 \log (1/\eps)$.
    Therefore, the resulting set of weights $w'$ is $(O(\e), \beta_1 + \frac{\beta_2}{\sqrt{n}}, \beta_2)$-regular.
    Plugging this into Lemma~\ref{lem:regular-weights-imply-tester} yields the claim, since by our choice of $n$, one can verify that $\sqrt{n} \beta_1, \beta_2 \leq \frac{1}{C} \cdot \alpha^2 n^2.$
\end{proof} 
\section{Computational Lower Bound}

We refer to \cite{kunisky2022notes} for definitions concerning the low-degree method and the low-degree likelihood ratio, reviewing here only a little background.
Originating in \cite{barak2019nearly,hopkins2018statistical,hopkins2017efficient}, the \emph{low-degree method} is a heuristic for understanding computational complexity of average-case problems, in this case a hypothesis testing problem.
The heuristic unconditionally rules out a class of algorithms based on low-degree evaluating low-degree polynomials in the input; in this case, low-degree polynomials in the $nd$-length vector $(X_1,\ldots,X_n)$.
This low-degree model captures a surprisingly broad range of algorithms, including spectral methods, making it a powerful heuristic for detecting computational hardness.

\begin{theorem}
    \label{thm:lb:lowdegree}
  For $\ns, \dims \in \N$ and $\cor, \dst > 0$, consider the following probability distributions on $D_0, D_1$ on $(\R^{\ns})^\dims$.
  \begin{itemize}
      \item $D_0 = \cN(0,1)^{\otimes \ns \dims}$
      \item $D_1$: first, sample a random unit vector $v \in \R^{\dims}$. Then, draw $\ns$ i.i.d. vectors $X_1,\ldots,X_{\ns}$ from the distribution $(1-\cor) \cN(\dst v, I) + \cor ( -\dst (1-\cor) \cor^{-1} v, I)$.
  \end{itemize}
  For $D \in \N$, $D > 1$, let $L^{\leq D}$ be the degree-$D$ truncated likelihood ratio for $D_1$ with respect to $D_0$.
  Then
  \[
    \|L^{\leq D} - 1\| \leq D^{O(D)} \cdot \frac{\sqrt{\ns} \dst^2}{\sqrt{\dims}{\cor}} \cdot \exp \left ( \frac {\sqrt{\ns} \dst^2}{\sqrt{\dims} \cor}  + \frac{\dst^2}{\sqrt{\dims} \cor^2} \right ) \, .
  \]
  Consequently, if $\ns \leq o(D^{-O(D)} \cdot \tfrac{\dims \cor^2}{\dst^4})$ and $\ns \geq \sqrt{\dims}/{\dst^2}$ (since otherwise testing is information-theoretically impossible), we have $\|L^{\leq D} - 1 \| \leq o(1)$.
\end{theorem}

\begin{proof}
  We define the following auxiliary distribution $P$ over $\ns \times \dims$ matrices -- to draw a sample from $P$, first draw a random unit vector $v \in \R^{\dims}$, then sample each column of $P$ independently to be equal to $\dst \cdot v$ with probability $(1-\cor)$ and otherwise equal to $-\dst (1-\cor) \cor^{-1}$.
  Note that an equivalent way to sample from $D_1$ is to first draw $X \sim P$ and output $X + G$, where $G \in \R^{\ns \times \dims}$ has independent entries distributed as $\cN(0,1)$.
  Thus, $D_1$ fits into the \emph{Gaussian additive model.}
  
  Using Theorem 2.6 of \cite{kunisky2022notes}, 
  concerning low-degree likelihood ratio for Gaussian additive models, we have
  \[
    \|L^{\leq D} - 1\|^2 = \sum_{t = 1}^D \frac 1 {t!} \cdot \E_{X, X' \sim P} \iprod{X,X'}^t \, ,
  \]
  where $X,X'$ are independent draws from $P$ and $\iprod{ \cdot, \cdot}$ is the Euclidean inner product in $\ns \dims$ dimensions.

  Let $v,w$ be the independent random unit vectors associated to separate draws $X,X' \sim P$, let $S_v \subseteq [n]$ be the columns of $X$ equal to $\dst v$, and similarly for $S_w$.
  Then
  \[
  \iprod{X,X'}^t = \iprod{v,w}^t \cdot (\dst^2 |S_v \cap S_w| - \dst^2 (1- \cor) \cor^{-1} (|S_v \cap \overline{S_w}| + |S_w \cap \overline{S_v}) + \dst^2 (1-\cor)^2 \cor^{-2} |\overline{S_v} \cap \overline{S_w}|)^t \, .
  \]
  Furthermore, $v,w$ are independent from $S_v,S_w$, so
  \[
  \E \iprod{X,X'}^t = \E \iprod{v,w}^t \cdot \E (\dst^2 |S_v \cap S_w| - \dst^2 (1- \cor) \cor^{-1} (|S_v \cap \overline{S_w}| + |S_w \cap \overline{S_v}|) + \dst^2 (1-\cor)^2 \cor^{-2} |\overline{S_v} \cap \overline{S_w}|)^t  \, .
  \]
  The whole quantity is equal to zero for odd $t$, and for even $t$ it's at most
  \[
  \frac{O(t)^{t/2}} {\dims^{t/2} } \cdot \dst^{2t} \cdot \E (|S_v \cap S_w| - (1- \cor) \cor^{-1} (|S_v \cap \overline{S_w}| + |S_w \cap \overline{S_v}|) + (1-\cor)^2 \cor^{-2} |\overline{S_v} \cap \overline{S_w}|)^t  \, ,
  \]
  using Fact~\ref{fact:unif-sphere-moments}.
  We have
  \begin{align*}
    \E |S_v \cap S_w| & = n (1-\cor)^2 \\
    \E |S_v \cap \overline{S_w}| & = n \cor (1-\cor) \\
    \E |\overline{S_v} \cap S_w| & = n \cor (1-\cor) \\
    \E |\overline{S_v} \cap \overline{S_w}| & = n \cor^2
  \end{align*}
  and hence in particular
  \[
  \E |S_v \cap S_w| - (1-\cor)\cor^{-1} ( \E |S_v \cap \overline{S_w}| + \E |S_w \cap \overline{S_v}|) + (1-\cor)^2 \cor^{-2} \E |\overline{S_v} \cap \overline{S_w}| = 0 \, .
  \]
  Therefore,
  \begin{align*}
  & \E \iprod{X,X'}^t \leq \frac{O(t)^{t/2} \alpha^{2t}}{d^{t/2}} \cdot \\
   & \sqrt{\E \left ( |S_v \cap S_w| - \E |S_v \cap S_w| \right )^{2t}}\\
   & + \sqrt{(1-\cor)^{2t} \cor^{-2t} \E \left ( |S_v \cap \overline{S_w}| - \E |S_v \cap \overline{S_w}| \right )^{2t}}\\
   &+ \sqrt{(1-\cor)^{4t} \cor^{-4t} \E \left ( |\overline{S_v} \cap \overline{S_w}| - \E |\overline{S_w} \cap \overline{S_v}| \right)^{2t} }
  \end{align*}
  Observe that:
  \begin{itemize}
      \item $|S_v \cap S_w|$ follows a Binomial distribution $\operatorname{Bin}(\ns,(1-\cor)^2)$, so 
      \[
      \E (|S_v \cap S_w| - \E |S_v \cap S_w|)^{2t} \leq O(t)^{2t} \ns (1-\cor)^2 + O(t)^t (\ns (1-\cor)^2 (2 \cor - \cor^2))^{t}
      \]
      using Fact~\ref{fact:binomial-moments}.
      \item $|S_v \cap \overline{S_w}|$ follows a Binomial distribution $\operatorname{Bin}(\ns, \cor(1-\cor))$, so
      \[
        \E (|S_v \cap \overline{S_w}| - \E |S_v \cap \overline{S_w}|)^{2t} \leq O(t)^{2t} \ns \cor (1-\cor) + O(t)^t (\ns \cor (1-\cor) (1 - \cor (1-\cor)))^t
      \]
      using Fact~\ref{fact:binomial-moments}.
      \item $|\overline{S_v} \cap \overline{S_w}|$ follows a Binomial distribution $\operatorname{Bin}(\ns, \cor^2)$, so
      \[
        \E (|\overline{S_v} \cap \overline{S_w}| - \E |\overline{S_v} \cap \overline{S_w}|)^{2t} \leq O(t)^{2t} \ns \cor^2 + O(t)^t (\ns \cor^2 (1-\cor^2))^t
      \]
      using Fact~\ref{fact:binomial-moments}.
  \end{itemize}
  Substituting these moment bounds and simplifying using $t \geq 1$, we get
  \[
    \E \iprod{X,X'}^t \leq \frac{O(t)^{O(t)} \dst^{2t}}{\dims^{t/2}} \left (  \sqrt{\ns} \cdot \eta^{-2t+1}  + \ns^{t/2} \cdot \cor^{-t}  \right ) \, .
  \]
  Summing across $t \in [1,D]$ gives the result.
\end{proof}

\begin{fact}
\label{fact:unif-sphere-moments}
Let $u,v$ be independent random unit vectors in $\dims$ dimensions and let $t \in \N$.
Then $\E \iprod{u,v}^t \leq O(t)^{t/2} \cdot O(\dims)^{-t/2}$.
\end{fact}
\begin{proof}
    The random variable $\iprod{u,v}$ has the same distribution as $\iprod{x,g} / \|g\|$ where $x$ is any fixed unit vector and $g \sim \cN(0,I)$. We have
    \[
    \E \frac{\iprod{x,g}^t}{\|g\|^t} \leq \left ( \E \iprod{x,g}^{2t} \right )^{1/2} \cdot \left ( \E \|g\|^{-2t} \right )^{1/2} \, .
    \]
    Since $\iprod{x,g}$ is distributed as $\cN(0,1)$, the first term is at most $\sqrt{(2t)^{t}} = O(t)^{t/2}$.
    For the second term, $\|g\| \geq \Omega(\sqrt{\dims})$ with probability at least $0.9$, so $\E \|g\|^{-2t} \leq O(1)^t \cdot \dims^{-t}$.
\end{proof}

The following is a special case of Rosenthal's inequality; see e.g. \cite{pinelis1994optimum}.
\begin{fact}[Moments of binomial distribution]
   \label{fact:binomial-moments}
   There is a constant $C > 0$ such that for all $n,t \in \N$ and $p \in [0,1]$, 
   if $Y \sim \operatorname{Bin}(n,p)$ be a binomial random variable,
   $\E (Y - \E Y)^t \leq (Ct)^t \cdot np + (Ct)^{t/2} \cdot (np(1-p))^{t/2}$.
\end{fact} 
\section*{Acknowledgments}

The authors would like to thank Guy Blanc and Gautam Kamath for some helpful suggestions. 

\printbibliography

\appendix

\section{Mathematica code to verify the computation from~\cref{sec:final:computation:mathematica}}
    \label{app:mathematica}
For completeness, we here provide some Mathematica code which can be used to verify the computations from the proof of~\cref{lemma:tedious:compputation}:
\begin{lstlisting}[extendedchars=true,language=Mathematica]
s11[\[Alpha]_, \[Beta]_, \[Epsilon]_, n_, d_] = 
  2*(1 - \[Epsilon])*n*\[Beta]^2*d;
s12[\[Alpha]_, \[Beta]_, \[Epsilon]_, n_, 
   d_] = (1 - \[Epsilon])*n*\[Beta]^2*d + \[Beta]*d;
s13[\[Alpha]_, \[Beta]_, \[Epsilon]_, n_, 
   d_] = (1 - \[Epsilon])*n*\[Beta]^2*d + \[Beta]*d;
s14[\[Alpha]_, \[Beta]_, \[Epsilon]_, n_, d_] = 2*\[Beta]*d;
s21[\[Alpha]_, \[Beta]_, \[Epsilon]_, n_, 
   d_] = (1 - \[Epsilon])*n*\[Beta]^2*d + \[Beta]*d;
s22[\[Alpha]_, \[Beta]_, \[Epsilon]_, n_, 
   d_] = (1 - \[Epsilon])*n*\[Beta]^2*d - \[Alpha]^2;
s23[\[Alpha]_, \[Beta]_, \[Epsilon]_, n_, d_] = 2*\[Beta]*d;
s24[\[Alpha]_, \[Beta]_, \[Epsilon]_, n_, d_] = \[Beta]*d - \[Alpha]^2;
s31[\[Alpha]_, \[Beta]_, \[Epsilon]_, n_, 
   d_] = (1 - \[Epsilon])*n*\[Beta]^2*d + \[Beta]*d;
s32[\[Alpha]_, \[Beta]_, \[Epsilon]_, n_, d_] = 2*\[Beta]*d;
s33[\[Alpha]_, \[Beta]_, \[Epsilon]_, n_, 
   d_] = (1 - \[Epsilon])*n*\[Beta]^2*d - \[Alpha]^2;
s34[\[Alpha]_, \[Beta]_, \[Epsilon]_, n_, d_] = \[Beta]*d - \[Alpha]^2;
s41[\[Alpha]_, \[Beta]_, \[Epsilon]_, n_, d_] = 2*\[Beta]*d;
s42[\[Alpha]_, \[Beta]_, \[Epsilon]_, n_, d_] = \[Beta]*d - \[Alpha]^2;
s43[\[Alpha]_, \[Beta]_, \[Epsilon]_, n_, d_] = \[Beta]*d - \[Alpha]^2;
s44[\[Alpha]_, \[Beta]_, \[Epsilon]_, n_, d_] = -2*\[Alpha]^2;
\end{lstlisting}
\begin{lstlisting}[extendedchars=true,language=Mathematica]
diag[\[Epsilon]_, \[Gamma]_, n_] = 
  DiagonalMatrix[{\[Gamma]*n, (\[Epsilon] - \[Gamma])*
     n, (\[Epsilon] - \[Gamma])*n, (1 - 2*\[Epsilon] + \[Gamma])*n}];
I4 = DiagonalMatrix[{1, 1, 1, 1}];
sigma[\[Alpha]_, \[Beta]_, \[Epsilon]_, n_, 
   d_] = {{s11[\[Alpha], \[Beta], \[Epsilon], n, d], 
    s12[\[Alpha], \[Beta], \[Epsilon], n, d], 
    s13[\[Alpha], \[Beta], \[Epsilon], n, d], 
    s14[\[Alpha], \[Beta], \[Epsilon], n, 
     d]}, {s21[\[Alpha], \[Beta], \[Epsilon], n, d], 
    s22[\[Alpha], \[Beta], \[Epsilon], n, d], 
    s23[\[Alpha], \[Beta], \[Epsilon], n, d], 
    s24[\[Alpha], \[Beta], \[Epsilon], n, 
     d]}, {s31[\[Alpha], \[Beta], \[Epsilon], n, d], 
    s32[\[Alpha], \[Beta], \[Epsilon], n, d], 
    s33[\[Alpha], \[Beta], \[Epsilon], n, d], 
    s34[\[Alpha], \[Beta], \[Epsilon], n, 
     d]}, {s41[\[Alpha], \[Beta], \[Epsilon], n, d], 
    s42[\[Alpha], \[Beta], \[Epsilon], n, d], 
    s43[\[Alpha], \[Beta], \[Epsilon], n, d], 
    s44[\[Alpha], \[Beta], \[Epsilon], n, d]}};
final[\[Alpha]_, \[Beta]_, \[Epsilon]_, n_, d_, \[Gamma]_] = 
  I4 + Dot[diag[\[Epsilon], \[Gamma], n], 
     sigma[\[Alpha], \[Beta], \[Epsilon], n, 
      d]]/((1 - \[Epsilon])*\[Alpha]^2*n + d);
\end{lstlisting}
\begin{lstlisting}[extendedchars=true,language=Mathematica]
FullSimplify[( (d + \[Beta]^2 d (\[Epsilon]^2 \[Minus] \[Gamma]) \
n^2)^2 \[Minus] (\[Alpha]^2 n \[Minus] 
       2 \[Beta] d \[Epsilon] n + \[Beta]^2 d \[Epsilon] n^2 \[Minus] 
       2 \[Alpha]^2 \[Epsilon] n + \[Alpha]^2 \[Gamma] n + 
       2 \[Beta] d \[Gamma] n \[Minus] 
       2 \[Beta]^2 d \[Epsilon]^2 n^2 + \[Beta]^2 d \[Epsilon]\
  \[Gamma]  n^2)^2 ) == 
  Factor[Det[
     final[\[Alpha], \[Beta], \[Epsilon], n, 
      d, \[Gamma]]]*(d + \[Alpha]^2*n - \[Alpha]^2 *\[Epsilon]*
        n)^2] ] 
\end{lstlisting} 
\end{document}